\pdfoutput=1
\RequirePackage{ifpdf}
\ifpdf 
\documentclass[pdftex]{sigma}
\else
\documentclass{sigma}
\fi

\numberwithin{equation}{section}

\newtheorem{Theorem}{Theorem}[section]
\newtheorem*{Theorem*}{Theorem}
\newtheorem{Corollary}[Theorem]{Corollary}
\newtheorem{Lemma}[Theorem]{Lemma}
\newtheorem{Proposition}[Theorem]{Proposition}

\theoremstyle{definition}
\newtheorem{Definition}[Theorem]{Definition}
\newtheorem{Example}[Theorem]{Example}
\newtheorem{Remark}[Theorem]{Remark}

\usepackage{mathtools}
\usepackage{tabularx}
 \newcolumntype{Y}{>{\centering\arraybackslash}X}
\usepackage{makecell}
\usepackage{tikz}
 \usetikzlibrary{calc,decorations.pathmorphing,shapes}
 \newcounter{sarrow}
 \newcommand\longrightsquigarrow[1]{%
 \stepcounter{sarrow}%
 \mathrel{
 \begin{tikzpicture}[baseline={($(current bounding box.south)+(0,-0.5ex)$)}]
 \node[inner sep=.5ex] (\thesarrow) {$\scriptstyle #1$};
 \path[
 draw,<-,decorate,
 decoration={zigzag,amplitude=0.7pt,segment length=1.2mm,pre=lineto,pre length=4pt}
 ]
 (\thesarrow.south east) -- (\thesarrow.south west);
 \end{tikzpicture}
 }%
 }

\newcommand{\thalf}{\tfrac{1}{2}}

\newcommand{\del}{\partial}
\DeclareMathOperator{\Tr}{Tr}

\newcommand{\abs}[1]{\left| #1 \right|}
\newcommand{\norm}[1]{\left\lVert #1 \right\rVert}

\newcommand{\ket}[1]{\left| #1 \right>}
\newcommand{\vev}[1]{\langle #1 \rangle}
\renewcommand{\vec}[1]{\ensuremath{\mathbf{#1}}} 

\newcommand{\hodge}{\star}
\newcommand{\hash}{\ensuremath\raisebox{-.08em}{\scalebox{1.4}{\#}}}

\DeclareMathOperator{\Ric}{Ric}

\newcommand\Restrict[2]{{\left. \kern-\nulldelimiterspace #1 \right\vert_{#2} }}								 
\newcommand\restr{\raisebox{-.2ex}{\ensuremath \vert}} 					 

\newcommand{\N}{\mathbb{N}}			
\newcommand{\Z}{\mathbb{Z}} 		
\newcommand{\R}{\mathbb{R}} 		
\newcommand{\C}{\mathbb{C}}			

\DeclareMathOperator{\ad}{ad}
\DeclareMathOperator{\Ad}{Ad}
\DeclareMathOperator{\Hom}{Hom}

\newcommand{\HW}[1][]{{\mathop{\mathbf{HW}}{\hspace{-0.2em}}_{#1\, }}}
\newcommand{\KW}[1][]{{\mathop{\mathbf{KW}}{\hspace{-0.2em}}_{#1\, }}}
\newcommand{\VW}[1][]{{\mathop{\mathbf{VW}}{\hspace{-0.2em}}_{#1\, }}}

\newcommand{\EBE}{{\mathop{\mathbf{EBE}}}}
\newcommand{\TEBE}[1][]{{\mathop{\mathbf{TEBE}}{}_{#1\, }}}
\newcommand{\Nahm}[1][]{{\mathop{\mathbf{Nahm}}{}_{#1\, }}}

\begin{document}

\allowdisplaybreaks

\newcommand{\arXivNumber}{2412.13285}

\renewcommand{\PaperNumber}{028}

\FirstPageHeading

\ShortArticleName{A Family of Instanton-Invariants for Four-Manifolds}

\ArticleName{A Family of Instanton-Invariants for Four-Manifolds\\ and Their Relation to Khovanov Homology}

\Author{Michael BLEHER}

\AuthorNameForHeading{M.~Bleher}

\Address{Institute for Mathematics, Heidelberg University,\\ Im Neuenheimer Feld 205, Heidelberg, Germany}
\Email{\mail{mbleher@mathi.uni-heidelberg.de}}

\ArticleDates{Received September 30, 2025, in final form March 05, 2026; Published online March 23, 2026}

\Abstract{This article provides a review of the gauge-theoretic approach to Khovanov homology, framed in terms of a generalisation of Witten's original proposal. Concretely, the physical arguments underlying Witten's insights suggest that there is a one-parameter family of Haydys--Witten instanton Floer homology groups \smash{$HF_{\theta}\bigl(W^4\bigr)$} for four-manifolds. At~the~heart of the proposal is a systematic investigation of the dimensional reductions of the Haydys--Witten equations. It is shown that on the five-dimensional cylinder $M^5=\mathbb{R}_s\times W^4$ with nowhere-vanishing vector field $v=\cos\theta \partial_s+\sin\theta w$, the Haydys--Witten equations provide flow equations for the $\theta$-Kapustin--Witten equations on $W^4$. Similar reductions to lower dimensions include the twisted extended Bogomolny equations on three-manifolds and the twisted octonionic Nahm equations on one-manifolds, whose solutions provide natural boundary conditions along the boundary and corners of $W^4$. These reductions determine the indicial roots of the Haydys--Witten and $\theta$-Kapustin--Witten equations with twisted Nahm-pole boundary conditions, which are required to establish elliptic regularity. Motivated by these insights, the groups \smash{$HF_{\theta}\bigl(W^4\bigr)$} are defined in analogy with Yang--Mills instanton Floer theory: solutions of the $\theta$-Kapustin--Witten equations on~$W^4$ modulo Haydys--Witten instantons on the cylinder $\mathbb{R}_s\times W^4$ interpolating between them. The relation to knot invariants observed by Witten arises when the four-manifold is the geometric blow-up $W^4=\smash{\bigl[X^3\times\mathbb{R}^+,K\bigr]}$ along a~knot $K\subset X^3\times{0}$ in its three-dimensional boundary. This yields a~precise restatement of Witten's conjecture as the equality between \smash{$HF^\bullet_{\pi/2}\bigl(\bigl[S^3\times\mathbb{R}^+,K\bigr]\bigr)$} and Khovanov homology $\mathrm{Kh}^\bullet(K)$.}


\Classification{57R58; 53C07; 81T13}

\section{Introduction}

Haydys--Witten instanton Floer theory is a natural five-dimensional analogue of the classical four-dimensional Donaldson--Floer theory -- also known as Yang--Mills instanton Floer theory.
In Donaldson--Floer theory, one constructs topological invariants of three-manifolds $X^3$ from the set of flat connections on a $G$-principal bundle over $X^3$.
One identifies pairs of flat connections that are related by an anti-self-dual connection (Yang--Mills instanton) on $W^4 = \mathbb{R}_s \times X^3$, which interpolates between them along the flow direction $\mathbb{R}_s$.
In the language of physics, Donaldson--Floer theory describes the quantum mechanical ground states of a topological quantum field theory obtained by coupling three-dimensional Chern--Simons theory to a topologically twisted four-dimensional $\mathcal{N}=2$ super Yang--Mills (SYM) theory.
Floer showed that this construction yields topological invariants of $X^3$, known as Floer groups \smash{$HF_\bullet\bigl(X^3\bigr)$}~\cite{Floer1988,Floer1989}.

Haydys--Witten Floer theory applies the same idea one dimension higher.
One considers topologically twisted five-dimensional $\mathcal{N}=2$ SYM theory on a five-manifold $M^5$, coupled to a topological twist of four-dimensional $\mathcal{N}=4$ SYM theory on its boundary.
Classical ground states in this theory are described by $\theta$-Kapustin--Witten solutions on $W^4$, while instanton corrections correspond to Haydys--Witten solutions on $\R_s \times W^4$ that interpolate between them.
In analogy with Donaldson--Floer theory, Haydys--Witten Floer theory then constructs from the Morse--Smale--Witten complex $CF_{\theta}\bigl(W^4\bigr)$ of $\theta$-Kapustin--Witten solutions the homology groups
\begin{align*}
	HF_{\theta}\bigl(W^4\bigr) = H\bigl( CF_{\theta}\bigl(W^4\bigr), {\rm d}_v\bigr) .
\end{align*}
The differential ${\rm d}_v$ counts solutions of the Haydys--Witten equations on the cylinder $\R_s \times W^4$ with respect to a distinguished vector field $v=\cos\theta \del_s + \sin\theta w$, where $w$ is a non-vanishing vector field on $W^4$:
\begin{align*}
	{\rm d}_v [x] = \sum_{\mu(x,y)=1} \hash \mathcal{M}(x,y) [y].
\end{align*}
Since the physical incarnation of these homology groups is as Hilbert space $\mathcal{H}_{\mathrm{BPS}}\bigl(W^4\bigr)$ of a~topologically twisted theory, it is expected that they are topological invariants of $W^4$.

A particularly interesting application of Haydys--Witten Floer theory arises for $W^4 = S^3 \times \R_y^+$ together with a knot embedded in its boundary $K \subset \del W^4 = S^3\times \{0\}$.
In the analytic context of Floer theory, the presence of the knot $K$ is encoded by imposing a particular asymptotic behaviour near the boundary and knot, known as Nahm pole boundary conditions with knot singularities along $K$.
According to an influential conjecture by Witten, the resulting topological invariants provide a gauge-theoretic interpretation of Khovanov homology~\cite{Witten2011}.

While various aspects of Haydys--Witten Floer theory have been investigated through its relation to Khovanov homology (e.g.,~\cite{Gaiotto2012a, Witten2011}), to the best of the author's knowledge the explicit description of a one-parameter family of Haydys--Witten Floer groups $HF_{\theta}\bigl(W^4\bigr)$ for general four-manifolds is new.
This article details how this one-parameter family arises naturally from the $\R_s$-invariant reduction of the Haydys--Witten equations on $\R_s \times W^4$: when the distinguished vector field is of the form $v=\cos\theta \del_s+\sin\theta w$ with \smash{$w \in \Gamma\bigl(TW^4\bigr)$} one obtains the $\theta$-Kapustin--Witten equations on $W^4$.
Put differently, the Haydys--Witten equations serve as the instanton flow equations for the $\theta$-Kapustin--Witten equations.
Moreover, the geometric configuration of~$v$ relative to $\R_s \times W^4$ has further crucial implications.
For instance, the incidence angle of~$v$ at a~boundary component dictates the necessary \emph{twisting} (or tilting) of appropriate Nahm pole boundary conditions and knot singularities.
Consequently, the five-dimensional geometry of the Haydys--Witten equations clarifies the interplay of various boundary conditions and offers insights into the elliptic regularity of the system on polycylindrical manifolds with boundaries and corners.

Haydys--Witten theory on four-manifolds has recently been the focus of independent research by Er, Ong and Tan~\cite{Er2023}.
For instance, they provide a physics-based derivation of the fact that when $\theta=0$ (i.e., $v$ coincides with the instanton flow direction $\del_s$), the Floer chain complex is generated by Vafa--Witten solutions.
Indeed, for compact four-manifolds, our proposed Haydys--Witten Floer groups for $\theta \in (0,\pi)$ are trivial, since the underlying $\theta$-Kapustin--Witten equations admit only trivial solutions (Theorem~\ref{thm:background-Kapustin--Witten-vanishing}, cf.\ Section~\ref{sec:background-HWF-theory-general-manifolds}).
This leaves the $\theta=0$ case, which is equivalent to Vafa--Witten theory~\cite{Vafa1994}, as the primary source of non-trivial Floer homology for compact four-manifolds~\cite{Haydys2015,Ong2022}.
Er, Ong, and Tan also briefly discuss a continuous family of Haydys--Witten theories on four-manifolds of the form $X^3 \times S^1$.
Such product manifolds permit a natural non-vanishing vector field $w$ along the $S^1$-factor and are thus closely related to the constructions discussed here~\cite[Section~7.4]{Er2023}.

While most of the presentation provides a review of previous results and arguments, this article also includes several new contributions:
First, we present a detailed discussion of the dimensional reductions of the Haydys--Witten equations to four, three, and one dimension in~Propositions~\ref{prop:background-dimensional-reduction-KW}--\ref{prop:background-dimensional-reduction-Nahm}.
Second, we establish the elliptic regularity of the Haydys--Witten equations together with the naturally induced boundary conditions on four-manifolds with boundaries and cylindrical ends by calculating their indicial set in Proposition~\ref{prop:background-indicial-set}.
Finally, these insights allow us to propose a tentative definition of a one-parameter family of Haydys--Witten Floer groups for four-manifolds in Definition~\ref{def:Haydys-Witten-Floer-homology}.
Along the way, we provide a thorough review of the underlying results from physics that provide the motivation for defining these as topological invariants, unify the mathematical ingredients into a coherent framework, and highlight necessary analytic details that need to be established to make the theory work.

The article is structured as follows.
We first recall some background from physics, starting with a short overview of {$4d$} $\mathcal{N}=4$ SYM theory with boundaries and line operators in~Section~\ref{sec:background-super-Yang-Mills-theory}.
Then we review topological twists in Section~\ref{sec:background-topological-twist} and specify the twists relevant for Haydys--Witten theory.
Section~\ref{sec:background-BPS-Kh-Homology} explains that the partition function of the four-dimensional theory admits a~categorification in terms of a Hilbert space of BPS states in a $5d$ $\mathcal{N}=2$ SYM theory.
These considerations explain the origin of the Haydys--Witten and Kapustin--Witten equations, and why they are expected to give rise to an interesting Floer theory.
In Section~\ref{sec:background-equations}, we give short individual introductions to a slightly confusing number of differential equations in various dimensions.
It is shown in Section~\ref{sec:background-dimensional-reduction} that these equations can all be viewed as dimensional reductions of the Haydys--Witten equations.
The Kapustin--Witten equations exhibit surprisingly restrictive vanishing results for finite energy solutions and it seems to be an important aspect of the theory to consider field configurations with singular boundary conditions.
This explains the relevance of Nahm pole boundary conditions with knot singularities, which we review in Section~\ref{sec:background-Nahm-pole-boundary-condition}.
The article concludes with a definition of Haydys--Witten Floer theory and an explanation how it captures Witten's proposal in Section~\ref{sec:background-HWF-theory}.

\section{Supersymmetric Yang--Mills theory with boundaries}
\label{sec:background-super-Yang-Mills-theory}

To set the stage and in view of later sections, we start by introducing some notation for general manifolds.
Let $G$ be a Lie group, and let $E \to W^4$ be a principal $G$-bundle with connection~${A\in\mathcal{A}(E)}$ over an oriented Riemannian four-manifold $\bigl(W^4, g\bigr)$.
We denote the Lie algebra of $G$ by $\mathfrak{g}$ and the adjoint bundle ${E \times_{\Ad} \mathfrak{g}}$~by~${\ad E}$.

There are spinor bundles $S^\pm$ associated to the ${\rm SO}(4)$-frame bundle over $W^4$, respectively of positive and negative chirality.
The underlying Weyl spinor representations are complex conjugates \smash{$\overline{S^+} = S^-$}.
Write $V$ for the complexified vector representation of ${\rm SO}(4)$.
A standard construction\footnote{Since the Clifford algebra acts on $\wedge^\bullet L$, the latter carries a representation of $\mathfrak{so}(4)\simeq \wedge^2 V \subset \mathrm{Cl}(V,g)$. Spinor representations are usually \emph{defined by} this representation, see, for example,~\cite{Deligne1999a}.} identifies ${S := S^+ \oplus S^-}$ with ${\wedge^\bullet L = \wedge^{\scriptscriptstyle \mathrm{even}} L \oplus \wedge^{\scriptscriptstyle \mathrm{odd}} L }$ for some choice of maximal totally isotropic subspace $L$ of $(V,g)$.
The Clifford algebra $\mathrm{Cl}(V,g)$, viewed as a ``deformation quantization'' of the exterior algebra of~$V$, acts naturally on $\wedge^\bullet L$.
In particular, since ${V \subset \mathrm{Cl}(V,g)}$, this induces a complex linear map $ \mathrm{cl}\colon V \otimes S^\pm \to S^\mp $ called Clifford multiplication.
Moreover, there is a $\C$-valued inner product $(s,t)$ on $S$, defined by restriction to the top-degree component of the element $s\wedge t \in \wedge^\bullet L$.
The inner product pairs each of $S^\pm$ with itself.
In combination with Clifford multiplication it induces a bilinear map $\Gamma\colon S^+ \otimes S^- \to V$, defined by duality: $g( \Gamma(s,t), v) := (s, \mathrm{cl}(v, t) )$.

The connection induces a covariant exterior derivative ${\rm d}_A$ on $\Omega^\bullet \bigl(W^4, \ad E\bigr)$ and a Dirac operator for spinors, given by composing the covariant derivative with Clifford multiplication
\begin{align*}
	D^A\colon\ \Gamma\bigl(W^4, S^\pm \otimes \ad E\bigr)\ \stackrel{\nabla^A}{\longrightarrow} \ \Gamma\bigl(W^4, T^\ast W^4 \otimes S^\pm \otimes \ad E\bigr) \ \stackrel{\mathrm{cl} \,\circ\, \sharp}{\longrightarrow}\ \Gamma\bigl(W^4, S^\mp \otimes \ad E\bigr).
\end{align*}

Let $\Tr(\cdot)$ denote the trace on $\mathfrak{g}$.
For \smash{$\alpha,\beta \in\Omega^k\bigl(W^4, \ad E\bigr)$}, introduce the density-valued inner product $\langle \alpha , \beta \rangle := \Tr \alpha\wedge\hodge \beta$ with associated norm \smash{$\norm{\alpha}^2 = \langle \alpha, \alpha \rangle$}.
For \smash{$s,t \in \Gamma\bigl(W^4, S^\pm \otimes \ad E\bigr)$}, we similarly write $\langle s, t \rangle := \Tr(s,t)\mu_{W^4}$, where $\mu_{W^4} = \sqrt{g}\, {\rm d}x^1 \cdots {\rm d}x^n$ denotes the volume form.

With that notation in place, we now specify the Lagrangian of $d=4$ $\mathcal{N} = 4$ super Yang--Mills theory.
The field content consists of a connection $A \in \mathcal{A}(E)$, six scalar fields $\phi_i \in \Omega^0\bigl(W^4,\ad E\bigr)$, $i=1,\ldots,6$, and four Weyl spinors ${\psi_a \in \Gamma\bigl(W^4, S^+ \otimes \ad E\bigr)}$, $a=1,\ldots,4$.
The action is the sum of a kinetic and a topological term
\begin{align}
	\label{eq:background-SYM-action}
	S = \frac{1}{g_{\text{YM}^2}} \int_{W^4} \mathcal{L}_{\text{kin}} + \frac{\theta_\mathrm{YM}}{32 \pi^2} \int_{W^4} \Tr F_A \wedge F_A.
\end{align}
For the rest of this section, we assume that $W^4$ is a region in Euclidean space $\R^4$.
In that case, the kinetic Lagrangian is given by the sum of the following two parts, where the first is purely bosonic and the second contains all contributions that involve fermions
\begin{align*}
&		\mathcal{L}_{\text{kin}}^{A,\phi}= \tfrac{1}{2} \norm{F_A}^2 + \sum_{i=1}^6 \norm{ {\rm d}_A \phi_i }^2 + \thalf \sum_{i,j=1}^6 \norm{[\phi_i, \phi_j]}^2, \\
&		\mathcal{L}_{\text{kin}}^{\psi}= \sum_{a=1}^4 \bigl\langle \bar{\psi}_a, D^A \psi_a \bigr\rangle + \sum_{\substack{i=1,\ldots,6 \\ a,b=1,\ldots,4}} C^{iab} \langle \psi_a , [\phi_i, \psi_b] \rangle.
\end{align*}
The coefficients $C^{iab}$ are related to the structure constants of ${\rm SU}(4)_R$, the $R$-symmetry of the~${\mathcal{N}=4}$ super-Poincar\'e algebra explained in more detail below.

The definitions of both the classical and quantum theory rely on a well-defined variational principle for the action.
For this it is necessary to specify boundary conditions.
For ${W^4 = \R^4}$, one typically assumes that fields and their derivatives fall off sufficiently fast at infinity.
In~general, however, admissible boundary conditions are determined by the requirement that any boundary terms that arise in a variation of $S$ vanish.

Under a variation $\delta \psi_a$ of one of the spinor fields, the variation of the action contains a~boundary term \smash{$\int_{\del W^4} \Gamma^\perp\bigl( \bar{\psi}_a, \delta \psi_a\bigr)$}, where $\Gamma^\perp$ denotes the post-composition of $\Gamma$ with projection to the direction perpendicular to the boundary.
This boundary term vanishes as long as any non-zero part of $\psi_a$ is orthogonal, with respect to $\Gamma^\perp(\bar{\cdot}, \cdot)$, to its variation $\delta \psi_a$.
Put differently, admissible boundary conditions for $\psi_a$ are determined by a choice of totally isotropic subspace $\mathfrak{S}_a \subset S^+$ and imposing \smash{$\psi_a\restr_{\del W^4} \in \Gamma\bigl(\del W^4, \mathfrak{S}_a \otimes \ad E\bigr)$}.
A similar analysis for the connection and scalar fields shows that their boundary conditions are in general given by Robin-type conditions, which relate normal derivatives and boundary values.
All in all, admissible boundary conditions are given by configurations that satisfy
\begin{align}
&	(F_A)_{y\mu} + \frac{\theta_\mathrm{YM} g_{\text{YM}}^2}{32 \pi^2} \epsilon_{y\mu\nu\lambda} (F_A)^{\nu \lambda}= 0, \nonumber\\
&	\nabla^A_y \phi_i - \sum_{j=1}^6 c_{ij} \phi_j= 0, \nonumber\\
&	\psi_a\restr_{\del W^4}\in \Gamma\bigl(\del W^4, \mathfrak{S}_a\bigr).
	\label{eq:background-admissible-boundary-conditions}
\end{align}
Note that the boundary conditions for the gauge field are completely fixed by a combination of the $\theta_\mathrm{YM}$-angle and coupling constant $g_{\mathrm{YM}}$.
In contrast, the coefficient functions $c_{ij}$ are generally not restricted, though a generic choice will break invariance under the action of ${\rm SU}(4)_R \simeq {\rm SO}(6)$ that rotates the six scalar fields $\phi_i$ into linear combinations of each other.

For $W^4=\R^4$, the theory is invariant under the action of the $4d$ $\mathcal{N}=4$ super-Poincar\'e algebra.
This is the $\mathbb{Z}/2\mathbb{Z}$-graded Lie algebra $A = A^0 \oplus A^1$ with bosonic and fermionic part given by
\begin{align*}
&	A^0= ( V \rtimes \mathfrak{so}(V) ) \times \mathfrak{su}(4)_R, \qquad	A^1= \bigl(S^+ \otimes \C^4\bigr) \oplus \bigl(S^- \otimes \bigl(\C^4\bigr)^\ast \bigr).
\end{align*}
The inclusion of $\C^4$ in the fermionic part provides $\mathcal{N}=4$ copies of the minimal super-Poincar\'e algebra in four dimensions.
The $\mathfrak{su}(4)_R$ part in the bosonic part is the Lie algebra of $R$-symmetry, which is defined to be any transformation that is represented non-trivially on $A^1$ and commutes with the action of the Lorentz group ${\rm SO}(V)$.
With respect to ${\rm SU}(4)_R$, the four spinors $\psi_a$ are in the defining representation $\C^4$ and the six scalars $\phi_i$ in the six-dimensional vector representation.

The super-Poincar\'e algebra is equipped with a $\mathbb{Z}/2\mathbb{Z}$ graded Lie bracket
\begin{align*}
	[x,y] = (-1)^{\abs{x}\abs{y}+1} [y,x],
\end{align*}
where $\abs{x}$ denotes the degree of homogeneous elements.
As a consequence $A^1$ carries a representation of $A^0$, while on $A^1$ the Lie bracket $[\cdot,\cdot]\colon A^1 \times A^1 \to A^0$ yields an intertwiner of $A^0$-representations.
For the spinorial part of $A^1$ this intertwiner is given by $\Gamma\colon S^+ \times\, S^- \to V$ extended by zero to all of $S$, together with the natural pairing on the \smash{$\mathbb{C}^4 \times \bigl(\mathbb{C}^4\bigr)^\ast$} factor.
It~follows that the anti-commutator of fermionic generators $[Q_1, Q_2]$ is always an element of $V$, corresponding to the common adage that supersymmetry squares to translations.

Consider now $W^4=\R^3 \times \R^+$.
The factor $\R^+ = [0,\infty)$ introduces a spacetime boundary and explicitly breaks translation invariance in the direction perpendicular to the boundary.
Accordingly, only fermionic symmetries in $A^1$ that do not square to a translation in the normal direction~${\R^\perp \subset V}$ can be preserved.
To that end, observe that \smash{$[Q_1,Q_2]\restr_{\R^\perp}$} is non-isotropic on~\smash{$\C^4 \times \bigl(\C^4\bigr)^\ast$} and reduces to $\Gamma^\perp(\cdot, \cdot)$ on $S^+ \times S^-$.
It follows that unbroken fermionic symmetries are elements of a totally isotropic subspace of the form \smash{$\bigl(\mathfrak{S} \otimes \C^4\bigr) \oplus \bigl(\overline{\mathfrak{S}} \otimes \bigl(\C^4\bigr)^\ast \bigr) \subset A^1$}.
As~a~consequence, the remaining supersymmetry algebra can contain at most half of the original fermionic generators.

A generic choice of boundary conditions satisfying~\eqref{eq:background-admissible-boundary-conditions} will not be invariant under the remaining super-Poincar\'e algebra.
A complete classification of ``half-BPS'' boundary conditions, those that are invariant under the remaining half of the fermionic generators, is described in~\mbox{\cite{Gaiotto2009b, Gaiotto2009a}}.
In the following, we only reproduce the result most relevant to us.

It turns out that for any choice of maximal totally isotropic subspace $\mathfrak{S}\subset S^+$, there exists a unique half-BPS boundary condition for $(A, \phi_i, \psi_a)$ that preserves the remaining half of the super-Poincar\'e algebra and full gauge symmetry.
There is a basis of $S^+$ in which the choice of $\mathfrak{S}$ is equivalent to fixing a generator of the form \smash{$Q = \bigl(\begin{smallmatrix} 1 \\ t \end{smallmatrix}\bigr)$} with $t\in\R$ (possibly infinite, interpreted as vanishing of the top component).
Invariance of the boundary values of $\psi_a$ under the action of $Q$ fully determines the choice of what was called \smash{$\mathfrak{S}_a \subset S^+$} earlier.
Namely, for each fermion we must have \smash{$\psi_a = \lambda_a \bigl(\begin{smallmatrix} t \\ 1 \end{smallmatrix}\bigr)$} for some $\lambda_a \in \C$.
Invariance under the remaining Lorentz group~${\rm SO}(3)$, $R$-symmetry ${\rm SO}(3)\times {\rm SO}(3) \subset {\rm SU}(4)_R$, and supersymmetry then also fixes the boundary conditions of the bosonic fields.
\begin{align*}
&	(F_A)_{y\mu} + \frac{t}{1-t^2} \epsilon_{y\mu\nu\lambda} F^{\nu \lambda}= 0 \\
&	\nabla^A_y \phi_i - \frac{t}{1+t^2} \epsilon_{ijk} [\phi_j, \phi_k]= 0 , \qquad i,j,k = 1,2,3, \\
&	\phi_{i+3}= 0 , \qquad i = 1,2,3, \\
&	\psi_a= \lambda_a \otimes \begin{pmatrix} t \\ 1 \end{pmatrix}.
\end{align*}
By comparison with the admissible boundary conditions of \eqref{eq:background-admissible-boundary-conditions}, it is clear that \smash{$\frac{t}{1-t^2} = \frac{\theta_\mathrm{YM} g_{\text{YM}}^2}{32 \pi^2}$}.
This condition has two roots, so for any value of SYM parameters $(g_\text{YM}, \theta_\mathrm{YM})$ there are two half-BPS boundary conditions that preserve the full gauge symmetry and which differ only by a~reversal of orientation.
Conversely, for any choice of preserved supersymmetries $\mathfrak{S} \subset S^+$, there exists a set of SYM parameters $(g_\text{YM}, \theta_\mathrm{YM})$ for which the equations above determine half-BPS boundary conditions.

\begin{Remark}
The half-BPS equations for the scalar fields $(A, \phi_i)$ form the basis of a set of conditions known as Nahm pole boundary conditions.
These conditions allow the inclusion of~'t Hooft operators along a knot $K\subset \del W^4$, by encoding a monodromy of $A$ and a certain singular behaviour of $\phi_i$ near the boundary.
Since the boundary conditions encode the presence of a~knot $K$, this plays a crucial role in the gauge theoretic approach to Khovanov homology and is discussed in detail in Section~\ref{sec:background-Nahm-pole-boundary-condition}.
\end{Remark}

\section{The Kapustin--Witten twists and localization}
\label{sec:background-topological-twist}

While the topological term $\int_{W^4} \Tr F_A \wedge F_A$ in the action functional~\eqref{eq:background-SYM-action} is manifestly independent of the metric on spacetime and only depends on the topology and smooth structure of the principal bundle $E$, the same is not true for $\mathcal{L}_\text{kin}$.
As a consequence, the partition function and observables of $4d$ $\mathcal{N}=4$ super Yang--Mills theory are generally not topological invariants.
However, in supersymmetric theories it is often possible to restrict to a subsector of the theory that depends only topologically on spacetime by a procedure known as twisting.

Below we briefly recall the twisting procedure and subsequently present the relevant twists of~$4d$~$\mathcal{N}=4$ SYM theory.
Importantly, the half-BPS boundary conditions discussed in Section~\ref{sec:background-super-Yang-Mills-theory} retain enough supersymmetry that the theory still admits a topological twist.
The section concludes with an explanation of how twisting leads to the Kapustin--Witten equations.

{\bf Topological twists of supersymmetric theories.}
Topological twists are a standard tool to extract topological field theories from supersymmetric ones~\cite{Witten1988, Witten1988a, Witten1991}.
Here we closely follow expositions in~\cite{Eager2021, Elliott2013}; for a more thorough introduction see~\cite{Vonk2005}.

First, recall from Noether's theorem that for every continuous symmetry of the action functional there is an associated conserved current $j$.
Applying this to translation symmetry of a~field theory on $\R^n$ results in a conserved current for each element of $\R^n$.
Choosing a basis~$x^\mu$ for~$\R^n$ the components of the associated conserved current are $j_\mu = T_{\mu\nu} {\rm d}x^\nu$.
The 2-tensor $T$ is the (canonical) energy-momentum tensor associated to the field theory and is related to variations of the metric by \smash{$\delta_g S = \int_{\R^4} T_{\mu\nu}\, \delta_{g_{\mu\nu}} \mathcal{L}$}.
The observables of a theory are independent of the metric -- hence, topological invariants in the sense of physics -- if the energy-momentum tensor vanishes.
It follows, in particular, that in topological theories translations must act trivially.

Even if a given field theory generally depends on the metric, it is still possible that certain protected subsectors of the theory are purely topological.
This happens, for example, if there is a~non-anomalous symmetry $Q$ with the following two properties:
\begin{itemize}\itemsep=0pt
	\item $Q$ is nilpotent (more precisely, square-zero): $Q^2 = 0$.
	\item $T$ is $Q$-exact: there exists a functional $V$ such that $T = [Q, V]$.
\end{itemize}
Now recall the general fact that in any quantum theory the expectation value of a symmetry-exact operator $\vev{ [Q,\mathcal{O}] }$ vanishes.
This can be seen, at least formally, by the following argument.
If $Q$ is a non-anomalous symmetry, i.e., the path-integral measure is invariant under transformations induced by $Q$ and $[Q, S] =0$, then the following expression must be independent of the infinitesimal parameter~$\epsilon$:
\begin{align*}
	\int \mathcal{D}\Phi \exp(\epsilon Q) ( \mathcal{O} \exp(-S) )
	= \int \mathcal{D}\Phi \exp(-S) ( \mathcal{O} + \epsilon [Q, \mathcal{O}] ).
\end{align*}
This can only be the case if the second term on the right vanishes, which is exactly the statement $\vev{ [Q, \mathcal{O}] } = 0$.
As a particular consequence we observe that when $Q^2 = 0$, the subsector given by $Q$-cohomology, i.e., the set of $Q$-invariant operators modulo those operators that are invariant for the uninteresting reason of being $Q$-exact, determines a viable quantum field theory in its own right.
Moreover, the fact that the energy-momentum tensor is $Q$-exact implies that the expectation value of any $Q$-closed observable is topologically invariant.
This follows from studying the effect of a continuous variation of the metric on the expectation value of a $Q$-closed operator $\mathcal{O}$:
\begin{align*}
	\delta_g \vev{\mathcal{O}}
	= \int \mathcal{D} \Phi \mathcal{O} \delta_g \exp({\rm i} S)
	= \int \mathcal{D} \Phi \mathcal{O} [Q,V] \exp({\rm i} S)
	= \vev{[Q,\mathcal{O}V]} = 0.
\end{align*}
In summary, if $Q$ satisfies the two properties stated above, $Q$-cohomology provides a topological subsector of the theory.

Since supersymmetric theories are invariant under the super-Poincar\'e algebra $A = A^0 \oplus A^1$, which contains nilpotent symmetries, one can ask if there are $Q\in A^1$ that satisfy the two properties described above.
With regards to the first property, the set of nilpotent fermionic elements of any $\mathbb{Z}/2\mathbb{Z}$-graded algebra forms a projective variety \smash{$Y:= \bigl\{ Q^2 = 0\bigr\} \subset A^1$}, known as the nilpotence variety~\cite{Eager2021}.
The nilpotence variety of the super-Poincar\'e algebra depends only on the spacetime dimension and quantity of supersymmetry $\mathcal{N}$.
In view of the discussion above, every $Q\in Y(d,\mathcal{N})$ gives rise to an associated $Q$-invariant subsector of the theory.

For this to be a topological sector the energy momentum tensor must be $Q$-exact, which it cannot be, because $Q$ is in the spinor representation of the Lorentz group.
This problem can sometimes be circumvented by ``twisting'' the action of the Lorentz algebra $\mathfrak{so}(d)$ in such a way that $Q$ can be reinterpreted as a scalar operator.
To that end note that upon restriction to $Q$-cohomology the symmetry algebra is broken to the stabilizer of the line spanned by $Q$,
\begin{align*}
	I(Q) = \{x \in A \mid [x, Q] \propto Q\}.
\end{align*}
Other elements of $A$ do not preserve the kernel and image of $Q$, and therefore do not act on $Q$-cohomology.
With respect to $I(Q)$, $Q$ is tautologically a scalar (perhaps up to some currently irrelevant $\mathrm{U}(1)$ charges).
Since the Lorentz symmetry and the semi-simple part of $R$-symmetry act non-trivially on $Q$, $I(Q)$ cannot contain generators of the corresponding subalgebras.
However, it can contain combinations of the two, where the action of the Lorentz algebra is undone by the action of $R$-symmetry.
Indeed, $I(Q)$ may contain a bosonic subalgebra $\mathfrak{so}^\prime(d)$ that is isomorphic to the Lorentz algebra and is embedded in $A^0$ as follows:
\begin{align*}
	\mathfrak{so}(d) \stackrel{1\times \alpha}{\to} \mathfrak{so}(d) \times R \subset A^0.
\end{align*}
The map $\alpha$ is (the derivative of) a non-trivial homomorphism from the Lorentz group to the $R$-symmetry group and is commonly known as the twisting map.
Note that for any $Q\in Y$, there might exist several twisting maps and conversely the graph of a fixed twisting map might appear in the unbroken symmetries $I(Q)$ of several $Q$'s.

When such a twisted Lorentz symmetry exists one can view the $Q$-invariant subsector as a~field theory with $\mathfrak{so}^\prime(d)$ invariance in its own right.
Crucially, in this reinterpretation $Q$ is a~Lorentz scalar and the energy momentum tensor can be $Q$-exact.
This in turn is the case if all translations are $Q$-exact, i.e., if the subalgebra $E(Q):= \bigl[Q,A^1\bigr]$ is all of $V$.

In summary, to twist a supersymmetric theory means to restrict to the subsector of a theory given by $Q$-cohomology of some nilpotent supercharge $Q \in Y$, together with a choice of twisting homomorphism that identifies a Lorentz subgroup $\mathfrak{so}^\prime(d) \subset I(Q)$.
If $E(Q) = V$, the result is a~topological theory and is referred to as topological twist of the original theory.
Topologically twisted theories can manifestly be defined on general Riemannian manifolds $M^n$, while their supersymmetric ``parents'' may only make sense on very special manifolds: e.g., $\R^n$, or manifolds that admit covariantly constant spinors.
Although the metric of the Riemannian manifold is needed in the definition of the action, the existence of the nilpotent symmetry $Q$ makes sure that the theory is independent of the metric.

{\bf Kapustin--Witten twists.}
$\mathcal{N}=4$ super Yang--Mills theory on $\R^4$ admits several interesting twisting maps~\cite{Marcus1995, Yamron1988}; also see~\cite{Labastida1997} for a more complete discussion.
The topological twist that is relevant to Haydys--Witten Floer theory and Khovanov homology is commonly dubbed Kapustin--Witten or Geometric Langlands twist.
Here we cite the most relevant results; more detailed explanations can be found in~\cite{Kapustin2007, Witten2011}.

The four-dimensional Lorentz algebra is isomorphic to $\mathfrak{so}(4) \simeq \mathfrak{su}(2)_\ell \times \mathfrak{su}(2)_r$, where the subscripts $\ell$ and $r$ stand for \emph{left} and \emph{right} chiral part.
The Kapustin--Witten twisting homomorphism is given by a diagonal embedding
\begin{align*}
	\alpha\colon \ \mathfrak{so}(4) \simeq \mathfrak{su}(2)_\ell \times \mathfrak{su}(2)_r \hookrightarrow
	\begin{pmatrix} \mathfrak{su}(2)_\ell & 0 \\ 0 & \mathfrak{su}(2)_r \end{pmatrix} \subset \mathfrak{su}(4)_R.
\end{align*}
The centralizer of the graph of $\alpha$ in $\mathfrak{so}(4)\times \mathfrak{su}(4)_R$ is an additional $U(1)$ that acts on the two blocks by multiplication with $e^{\pm {\rm i} \omega}$, respectively.
This gives rise to a $U(1)$ symmetry in the twisted theory and an associated $\mathbb{Z}$-grading on the twisted fields.
We will later see that it plays the role of a fermion number $F$ and will thus be denoted $U(1)_F$ and $\mathbb{Z}_F$ from now on.

The full nilpotence variety in $4d$ $\mathcal{N}=4$ SYM theory is an 11-dimensional stratified variety~\cite{Eager2021}.
This variety contains a $\mathbb{CP}^1$-family of supercharges for which on the one hand $I(Q)$ contains the (fixed) twisted Lorentz algebra $\mathfrak{so}^\prime(4)$ and on the other hand $E(Q) = V$, such that $Q$-cohomology becomes topological after twisting.
From the point of view of the ${\rm SO}^\prime(4)$ Lorentz group, the fields transform in the following representations:
\begin{itemize}\itemsep=0pt

	\item The connection $A$ transforms trivially under $R$-symmetry and thus remains unchanged.

	\item Four of the scalar fields $\phi_i$ combine into a vector representation $\phi$ with $F=0$, while the remaining two are ${\rm SO}(4)^\prime$ scalars that are rotated into each other by $U(1)_F \simeq {\rm SO}(2)$. Combining the latter two into $\sigma = \phi_5 + {\rm i}\phi_6$ and $\bar{\sigma} = \phi_5 - {\rm i}\phi_6$ makes these into $\mathfrak{g}_{\C}$-valued Lorentz scalars with fermion number $F=\pm 2$.

	\item The sixteen real components of the four Weyl fermions $\psi_a$ distribute into two vector representations $\lambda_1$, $\lambda_2$ of ${\rm SO}(4)^\prime$ with $F=1$, an antisymmetric representation $\chi$ with $F=-1$, and two Lorentz scalars $\eta_1$, $\eta_2$ with $F=-1$.

\end{itemize}

On a general Riemannian manifold the field content of the twisted theory thus arranges into a $\mathbb{Z}$-graded chain complex, where $Q$ acts as differential of degree $1$,
\begin{center}
\fontsize{8.5}{8.8}\selectfont
\setlength{\tabcolsep}{1pt}
\begin{tabularx}{\textwidth}{YYYYY}
$F=-2$ & $-1$ & $0$ & $1$ & $2$ \\[0.4em]
\makecell{${\bar{\sigma} \in \Omega^1(W^4, \ad E_{\C}})$}&%
\makecell{${\eta_1, \eta_2 \in \Omega^0(W^4,\ad E)}$ \\ ${\chi \in \Omega^2(W^4,\ad E)}$}&%
\makecell{${A \in \mathcal{A}(E)}$ \\ ${\phi \in \Omega^1(W^4, \ad E)}$}&%
\makecell{${\lambda_1,\lambda_2 \in \Omega^1(W^4,\ad E)}$}&%
\makecell{${\sigma \in \Omega^1(W^4, \ad E_{\C} )}$}
\end{tabularx}
\end{center}

Recall from Section~\ref{sec:background-super-Yang-Mills-theory} that in the presence of a boundary, a totally isotropic subspace \smash{$\bigl(\mathfrak{S} \otimes \C^4\bigr) \oplus \bigl(\overline{\mathfrak{S}} \otimes \bigl(\C^4\bigr)^\ast \bigr) \subset A^1$} of the fermionic symmetries can be preserved by choosing half-BPS boundary conditions.
The preserved supersymmetry determined by $\mathfrak{S}$ and the $\mathbb{CP}^1$-family of nilpotent charges that admit a compatible Kapustin--Witten twist intersect in a single point~$Q$~\cite{Witten2011}.

{\bf Localization of topologically twisted partition function.}
Topologically twisted ${\mathcal{N} = 4}$ super Yang--Mills theory is defined on a general Riemannian four-manifold with boundary $\bigl(W^4, g\bigr)$.
As before, the action is given by the expression
\begin{align*}
	S = \frac{1}{g_{\text{YM}^2}} \int_{W^4} \mathcal{L}_{\text{kin}} + \frac{\theta_\mathrm{YM}}{32 \pi^2} \int_{W^4} \Tr F_A \wedge F_A,
\end{align*}
where the Lagrangian arises from the one on flat Euclidean space by rewriting it in terms of the ${\rm SO}^\prime(4)$ fields on $W^4$, and adding curvature terms as necessary to make the action $Q$-invariant.
For example, the part of the Lagrangian that contains the connection $A$ and the one-form \smash{$\phi\in \Omega^1\bigl(W^4, \ad E\bigr)$} is given by{\samepage
\begin{align*}
		&\mathcal{L}_{\text{kin}}^{A,\phi}= \frac{1}{2} \norm{F_A}^2 + \norm{ {\rm d}_A \phi }^2 + \frac{1}{2} \norm{[\phi\wedge\phi]} + \langle \Ric \phi, \phi \rangle,
\end{align*}
where $\Ric$ denotes the Ricci tensor, viewed as linear map on $\Omega^1\bigl(W^4\bigr)$.}

Let us now investigate the partition function of the topologically twisted theory, possibly including the insertion of a 't Hooft operator supported on $K \subset \del W^4$ as specified by the BPS boundary conditions of Section~\ref{sec:background-super-Yang-Mills-theory}.
We denote the partition function of this theory by
\begin{align*}
	Z^Q_{\mathrm{SYM}}\bigl(W^4, K\bigr) = \int \mathcal{D}\Phi \exp(-S).
\end{align*}
It turns out that the path integral localizes on field configurations that obey $[Q, \Psi] = 0$, for all~$\Psi$ with odd fermion number $F$.
Among the fields with fermion number $F = -1$ is the two-form~$\chi$, with (anti-\nolinebreak)self-\nolinebreak{}dual parts $\chi^\pm$, and a scalar $\eta_1$.
The action of $Q$ on these fields produces
\begin{align*}
&	\bigl[Q, \chi^+\bigr]= ( F_A - [\phi\wedge\phi] + t {\rm d}_A \phi )^+ , \qquad	[Q, \chi^-]= \big( F_A - [\phi\wedge\phi] - t^{-1} {\rm d}_A \phi \big)^- , \\
&	[Q, \eta_1]= {\rm d}_A^{\hodge_4} \phi.
\end{align*}
Since the topologically twisted theory only depends on $Q$-cohomology, we can modify the action by adding $Q$-exact terms without changing the partition function,
\begin{align*}
	S^\prime
	&= S -\biggl[ Q ,\ \frac{1}{\epsilon} \int_{W^4} \bigl( \bigl\langle \chi^+, \bigl[Q,\chi^+\bigr] \bigr\rangle + \langle \chi^- , [Q,\chi^-] \rangle + \langle \eta_1 , [Q,\eta_1] \rangle \bigr) \biggr] \\
	&= S - \frac{1}{\epsilon} \int_{W^4} \bigl( \norm{ \bigl[Q, \chi^+\bigr] }^2 + \norm{ [Q,\chi^-] }^2 + \norm{ [Q,\eta_1] }^2 + \cdots \bigr).
\end{align*}
The omitted terms on the right are fermion bilinears and can be neglected, since in a supersymmetric ground state fermionic fields vanish.
Taking $\epsilon \to 0$, the action diverges except when~${\bigl[Q,\chi^\pm\bigr] = [Q, \eta] = 0}$, such that the path integral localizes on field configurations that satisfy these equations.
These are known as Kapustin--Witten equations and will be discussed in much more detail in Section~\ref{sec:background-Kapustin-Witten-equations}.
There is a similar localization argument for the complex scalar~$\sigma$ and its complex conjugate that provides the equations ${\rm d}_A \sigma = [\phi, \sigma] = [\sigma, \bar\sigma] = 0$.
These equations imply that~$\sigma$ and~$\bar{\sigma}$ vanish, at least as long as $(A,\phi)$ are irreducible solutions of the Kapustin--Witten equations.

A standard argument relates the expected dimension of the moduli space of solutions of the localization equations to the index of a Dirac-like operator on $W^4$.
In the path-integral formalism this index is equal to the anomaly in the fermion number $F$, so the partition function of the topologically twisted theory on $W^4$ vanishes except when the expected dimension of the moduli space of solutions is zero.
As a result, the partition function reduces to a sum over classical solutions of the localization equations on $W^4$, which additionally satisfy the supersymmetric boundary conditions described in Section~\ref{sec:background-super-Yang-Mills-theory} if $\del W^4 \neq \varnothing$.

Consider now the contribution to the partition function from a given solution of the localization equations.
In expanding the path integral around the solution we assume that there are no bosonic or fermionic zero modes and the solution fully breaks gauge symmetry, which is the expected case if the index is zero.
Taking $g_\text{YM} \to 0$, the calculation reduces to a one-loop approximation.
On the one hand this results in a factor of $\exp(-\theta_\mathrm{YM} P) =: q^P$ from the topological part of the action.
Here we abbreviated \smash{$P=\frac{1}{32\pi^2}\int_{W^4} \Tr F\wedge F$}, which is the integral of the second Chern class of the principal bundle $E$ and is known as instanton number.
On the other hand, we pick up the ratio of the fermion and boson determinant, which due to supersymmetry are equal up to a sign $(-1)^F$, where $F$ is interpreted as the fermion number of the solution, as~alluded to earlier.
It follows that any classical solution of the localization equations contributes to the partition function with a term $(-1)^F q^P$.

Writing $\mathcal{M}^{\mathrm{KW}}$ for the set of classical solutions of the Kapustin--Witten equations, we find that the topological partition function is given by
\begin{align*}
	Z_\text{SYM}^Q \bigl(W^4, K ;\, q\bigr) = \sum_{\Phi \in \mathcal{M}^{\mathrm{KW}}} (-1)^{F(\Phi)} q^{P(\Phi)} = \sum_{n \in \mathbb{Z}} a_n q^n.
\end{align*}
In rewriting the sum in the second equation, we assume that there are only finitely many solutions and denote by $a_n$ the number of solutions with $P=n$.
Since the presence of a 't Hooft operator along a knot $K \subset \del W^4$ is encoded in the boundary conditions, it affects the coefficients $a_n$.

To conclude this section, we cite a further key result of~\cite{Witten2010, Witten2011}, which states that the action of SYM theory on $W^4 = S^3 \times \R^+$ is equivalent to
\begin{align*}
	S = [ Q , \,\cdot\, ] + {\rm i} \psi\ CS( A + {\rm i} \omega \phi ),
\end{align*}
where $\psi$ and $\omega$ are determined by $\theta_\mathrm{YM}$ -- or, depending on preference, the choice of preserved supersymmetry.
As was the case for the instanton number $P$, the definition of the Chern--Simons functional in that equation needs a bit of care, see~\cite{Witten2010, Witten2011, Witten2011a}.

The twisted partition function is independent of $Q$-exact operators, so topologically twisted SYM theory on $W^4=S^3\times \R^+$ together with a 't Hooft operator along $K\subset \del W^4$ is equivalent to Chern--Simons theory, which famously calculates the Jones polynomial of $K$~\cite{Witten1989}:
\begin{align*}
	Z_{\rm SYM}^Q\bigl(S^3\times \R^+ , K\bigr) = Z_{\rm CS}\bigl(S^3, K\bigr) = J(K).
\end{align*}

\section{Hilbert space of BPS states in five dimensions}
\label{sec:background-BPS-Kh-Homology}

The Jones polynomial admits a categorification known as Khovanov homology~\cite{Khovanov1999}.
Categorification involves replacing a classical mathematical object with a richer, more structured object ``one step up'' in category theory.
The new object usually captures more information and often allows for more powerful tools and insights.

In the case of Khovanov homology and the Jones polynomial, categorification replaces a~polynomial invariant of knots and links with a homological invariant that assigns to each knot or link $K$ a~bigraded vector space $\mathrm{Kh}(K) = \bigoplus_{i,j} \mathrm{Kh}^{i,j}(K)$.
The Jones polynomial is obtained from Khovanov homology as the Euler characteristic of the graded vector space
\begin{align*}
	J(K) = \sum_{i,j} (-1)^i q^j \dim \mathrm{Kh}^{i,j}(K).
\end{align*}
A similar phenomenon exists in supersymmetric quantum field theories when the localization procedure reduces the path integral to a sum over supersymmetric vacua.
In such situations the partition function can often be expressed as a trace over the Hilbert space of BPS states in one dimension higher.

Indeed, it is well-known that $4d$ $\mathcal{N}=4$ super Yang--Mills theory can be viewed as the low-energy effective theory of a $5d$ $\mathcal{N}=2$ super Yang--Mills theory\footnote{Five-dimensional super Yang--Mills theory is not UV-complete and from the physics point of view it may be more satisfying to consider compactifications of the $6d$ $\mathcal{N}=(2,0)$ superconformal theory or of $10d$ $\mathcal{N}=1$ super Yang--Mills theory, see~\cite{Witten2011}.
From the Floer theory point of view, as adopted here, the five-dimensional interpretation is more natural.} compactified on a small circle $S^1$.
The Nahm pole boundary condition lifts to the five-dimensional theory by a translation-invariant continuation in the direction of the circle.
In particular, any 't Hooft operator supported on a knot $K$ in $\del W^4$ lifts to an $S^1$-invariant surface operator supported on $\Sigma_K = S^1 \times K$ in the boundary of the five-manifold.
The Nahm pole boundary condition and 't Hooft operator preserve the same supersymmetry generator as in the four-dimensional theory.
Furthermore, the Kapustin--Witten twisting homomorphism of the four-dimensional theory corresponds to an~analogous topological twist in five dimensions.
In particular, the topological supercharge $Q$ that defined the topological subsector in four dimensions, remains a nilpotent symmetry when the model is lifted to five dimensions.

In the five-dimensional theory on $S^1 \times W^4$ one constructs the Hilbert space of states $\mathcal{H}$ by canonical quantization on a ``temporal'' slice, i.e., on a codimension one submanifold $\{s\} \times W^4$, $s \in S^1$, in the background determined by the boundary conditions.
The partition function of the topologically twisted four-dimensional theory on $W^4$ is then equivalent to a trace in $\mathcal{H}$ with certain operator insertions:
\begin{align*}
	Z^Q_{\text{SYM}}(q) = \Tr_{\mathcal{H}} (-1)^F q^P.
\end{align*}
As before, $F$ and $P$ denote fermion and instanton number, but in the five-dimensional theory are interpreted as operators on Hilbert space.
$P$ is given in terms of the classical integral of the four-dimensional fields and promoted to an operator by classical quantization.
The operator $F$ in the expansion around a classical solution is given by summing over the $F$-eigenvalues of the filled Dirac sea, i.e., the total fermion number of the combination of all negative energy states.

In this approach $\mathcal{H}$ plays the role of the chain complex underlying Khovanov homology.
Indeed, since $Q$ is a nilpotent fermionic symmetry, it satisfies $Q^2=0$ and accordingly acts as a~differential on the Hilbert space.
The cohomology of $Q$ is commonly known as the space of BPS states $\mathcal{H}_\text{BPS} := H^\bullet(\mathcal{H}, Q)$ and, in the current situation, is spanned by supersymmetric ground states.
While the physical Hilbert space generally depends on various parameters and perhaps on choices during quantization, BPS states are protected by supersymmetry and are independent of continuous deformations of the theory.
As a consequence, \smash{$\mathcal{H}_\text{BPS}=\mathcal{H}_\text{BPS}\bigl(W^4, K\bigr)$} is a knot invariant and one expects that it is a gauge-theoretic manifestation of Khovanov homology.

In a first approximation, the quantum ground states of the five-dimensional theory are determined by time-independent classical solutions of the equations of motion.
The construction outlined above results in the fact that these correspond to solutions of the localization equations of the four-dimensional theory, i.e., solutions of the Kapustin--Witten equations.
For simplicity, assume that there is a finite, non-degenerate set of solutions $\mathcal{M}^{\mathrm{KW}}$, in particular assume that after gauge fixing the solutions do not have bosonic zero modes.
Then an expansion around a~given solution $\Phi \in \mathcal{M}^{\mathrm{KW}}$ produces a single perturbative ground state $\ket{\Phi}$ of zero energy.

To move beyond perturbation theory, consider the space $\mathcal{H}_0$ spanned by all perturbative supersymmetric ground states.
States in $\mathcal{H}_0$ can fail to be true quantum mechanical ground states if they are lifted from zero energy by some non-perturbative process.
It is a well-known property of supersymmetric quantum theories that eigenstates of the Hamiltonian with non-zero energy must appear in Bose--Fermi pairs that only differ with respect to their spin statistics, which in the current situation is given by $\mathbb{Z}_F / 2\mathbb{Z}$.
Recall that $\mathcal{H}_0$ carries a $\mathbb{Z}_P \times \mathbb{Z}_F$ grading\footnote{On general five-manifolds there is no $U(1)_F$ symmetry, but there is still a $\mathbb{Z}/2\mathbb{Z}$ grading by spin-statistics.} with respect to the instanton and fermion numbers.
Since the fermion operator satisfies $[F,Q] = Q$, the fermion numbers of such a Bose--Fermi pair differ by exactly one, while their instanton numbers coincide.

In the present context, quantum corrections can only arise by tunneling from one classical (perturbative) solution to another.
To understand the fundamentally quantum mechanical process of tunneling in the context of field theories, it is helpful to slightly change perspective.
The~result is the well-known application of Morse theory to Yang--Mills theory pioneered by Floer in the context of flat connections on three-manifolds and anti-self-dual connections (Yang--Mills instantons) on four-manifolds~\cite{Floer1988, Floer1989}.
For this, one reinterprets super Yang--Mills theory on~${\R_s \times W^4}$ as supersymmetric quantum mechanics on the space of field configurations over $W^4$, where the real coordinate $s$ takes the place of the time coordinate in quantum mechanics.
If one interprets the action functional of super Yang--Mills theory as potential energy, one finds that a~tunneling event corresponds to a solution of gradient flow equations~\cite{Witten1982}.
Following this line of argument in the current context and formulating everything in terms of five-dimensional fields, one arrives at the Haydys--Witten equations on $\R_s \times W^4$~\cite{Witten2011}.
We defer a detailed description of these equations and their specializations to the upcoming Section~\ref{sec:background-equations}.

The space of BPS ground states $\mathcal{H}_\text{BPS}$ is now given by $Q$-cohomology of the approximate space $\mathcal{H}_0$, where $Q$ acts by instanton corrections as explained above.
Since $Q$ has $F$-degree $+1$, or by the Bose--Fermi pair argument from above, only gradient flows that interpolate between solutions whose fermion numbers $f_{i}$ and $f_{j}$ differ by 1 can contribute to the correction.
The action of $Q$ is then given by
\begin{align*}
	Q \ket{\Phi_i} = \sum_{ \substack{\Phi_j \in \mathcal{M}^{\text{KW}} \\ \abs{f_i-f_j} = 1} } n_{ij} \ket{\Phi_j},
\end{align*}
where $n_{ij}$ is the signed count of instanton solutions that interpolate between the solutions $\Phi_i$ and $\Phi_j$.

Conceptually, the approximate Hilbert space $\mathcal{H}_0$ corresponds to the Morse--Smale--Witten complex, quantum corrections are given by the instanton Floer differential associated to the Haydys--Witten equations, and the Hilbert space of BPS states yields the associated Floer cohomology of the manifold.
The mathematical formulation of this instanton Floer theory is standard and will be summarized in Section~\ref{sec:background-HWF-theory}.
The key insight provided by physics and bottom line of this section is that this Floer theory based on the Haydys--Witten and Kapustin--Witten equations gives rise to interesting topological invariants of four-manifolds.

\section{The Haydys--Witten equations and their specializations}
\label{sec:background-equations}

The instanton equations mentioned in the preceding section are conveniently summarized in a formulation due to Haydys~\cite{Haydys2015}.
Below we first introduce Haydys' geometric setup and the Haydys--Witten equations, which provide a covariant lift of anti-self-dual equations in four dimensions to five-manifolds $M^5$ that are equipped with a non-vanishing vector field $v$.
Subsequently, we review the definitions and relevant properties of several closely related differential equations on lower-dimensional manifolds.
Namely, the $\theta$-Kapustin--Witten and Vafa--Witten equations on four-manifolds, the twisted extended Bogomolny equations on three-manifolds, and Nahm's equations on one-dimensional manifolds.
All these equations can be viewed as dimensional reductions of the Haydys--Witten equations.
We dedicate Section~\ref{sec:background-dimensional-reduction} to a detailed discussion of this fact.

\subsection{The Haydys--Witten equations}
\label{sec:background-Haydys--Witten-equations}

The Haydys--Witten equations are a set of partial differential equations on Riemannian five-manifolds that are equipped with a nowhere-vanishing vector field~$v$.
The equations were introduced by Haydys on general five-manifolds~\cite{Haydys2015} and at roughly the same time discovered independently by Witten in the special case $\R_s \times X^3 \times \R^+_y$ with $v=\del_y$~\cite{Witten2011}.
This section closely follows the original exposition of~\cite{Haydys2015}.

Let \smash{$\bigl(M^5,g\bigr)$} be a Riemannian five-manifold and $v$ be a nowhere-vanishing vector field of pointwise unit norm.
Consider a principal $G$-bundle $E\to M^5$ for $G$ a compact Lie group, write~$\mathcal{A}(E)$ for the space of connections, and denote by $\ad E$ the adjoint bundle associated to the Lie algebra $\mathfrak{g}$ of $G$.
Furthermore, for a connection $A \in \mathcal{A}(E)$ we denote the associated covariant derivative by $\nabla^A$ and the exterior covariant derivative by ${\rm d}_A$.

Write ${\eta = g(v,\cdot) \in \Omega^1(M)}$ for the one-form dual to the vector field $v$ and observe that the pointwise linear map
\begin{align*}
	T_\eta\colon\ \Omega^2(M) \to \Omega^2(M), \qquad
	\omega \mapsto \hodge_{5} ( \omega \wedge \eta)
\end{align*}
has eigenvalues $\{ -1, 0, 1\}$, such that $\Omega^2(M)$ decomposes into the corresponding eigenspaces:{\samepage
\begin{align*}
	\Omega^2(M) = \Omega^2_{v,-}(M) \oplus \Omega^2_{v,0}(M) \oplus \Omega^2_{v,+}(M).
\end{align*}
Below we will use the notation $\omega^+$ to denote the part of $\omega$ that lies in $\Omega^2_{v,+}(M)$.}

At every point $p\in M^5$, the fiber \smash{$\Omega^2_{v,+}(M)\restr_p$} is an oriented three-dimensional vector space, with orientation induced by that of \smash{$M^5$}, and thus carries a natural Lie algebra structure given by the usual cross product $(\cdot \times \cdot)$ of $\mathbb{R}^3$.
The map \smash{$\sigma(\cdot, \cdot) = \frac{1}{2} (\cdot \times \cdot) \otimes [ \cdot, \cdot ]_{\mathfrak{g}}$} determines a corresponding cross product on \smash{$\Omega^2_{v,+}(M, \ad E) \simeq \R^3 \otimes \mathfrak{g}$}.

\begin{Example}
To parse these constructions, consider their incarnation in a small neighbourhood~$U$ of a point $p\in M$, see Figure~\ref{fig:background-Haydys-geometry}.
Choose orthonormal coordinates $(x_i, y)_{i=0,1,2,3}$ based at~$p$ in~such a way that $v = \del_{y}$.
An explicit basis of $\Omega_{v,+}^2(U)\restr_p$ is given by
\begin{align*}
&	e_1= {\rm d}x^0\wedge {\rm d}x^1 + {\rm d}x^2 \wedge {\rm d}x^3 , \\
&	e_2= {\rm d}x^0\wedge {\rm d}x^2 - {\rm d}x^1 \wedge {\rm d}x^3 , \\
&	e_3= {\rm d}x^0\wedge {\rm d}x^3 + {\rm d}x^1 \wedge {\rm d}x^2 .
\end{align*}
Observe that this is also a basis of the self-dual two-forms with respect to $\hodge_4$ acting on the orthogonal complement of $v$ in $\Omega^2(U)\restr_p$.
We can extend the $e_i$ to a local frame over all of $U$ such that an element of \smash{$\Omega_{v,+}^2(U, \ad E)$} is of the form $B = e_1 \otimes \phi_1 + e_2 \otimes \phi_2+ e_3 \otimes \phi_3$ for some $\mathfrak{g}$-valued functions $\phi_a$, $a=1,2,3$.
Equivalently, the non-vanishing components of $B$ are $B_{0a} =\phi_a$, $B_{ab} = \epsilon_{abc} \phi_c$.
Finally, the cross product of $B$ with itself is given by
\begin{align*}
	\sigma(B, B) = e_1 \otimes [\phi_2, \phi_3] + e_2 \otimes [\phi_3, \phi_1] + e_3 \otimes [\phi_1,\phi_2],
\end{align*}
which in 2-form components corresponds to \smash{$\sigma(B,B)_{\mu\nu} = \frac{1}{4} g^{\rho\tau} [B_{\mu\rho}, B_{\nu\tau}]$}.
\end{Example}
\begin{Example}
Consider $\smash{\bigl(M^5, v\bigr) = \bigl(W^4 \times \R_{y\ } , \del_y \bigr)}$ with product metric and denote by $i\colon W^4 \hookrightarrow W^4 \times \R_y$ inclusion at $y=0$.
Then
\begin{align*}
	 \Omega^2_\pm\bigl(W^4\bigr) \simeq i^\ast \bigl( \Omega_{\del_y, \pm}^2\bigl(W^4 \times \R_y\bigr) \bigr),
\end{align*}
where \smash{$\Omega_\pm^2\bigl(W^4\bigr)$} denotes (anti-)self-dual two-forms on $W^4$ with respect to $\hodge_4$.
Indeed, whenever the metric is of product type, the Hodge star operator factorizes as
\begin{align*}
	\hodge_5 ( \alpha \wedge \beta ) = (-1)^{k\ell} \hodge_{W^4} \alpha \wedge \hodge_{\R_y} \beta,
\end{align*}
when \smash{$\alpha \in \Omega^k\bigl(W^4\bigr)$}, $\beta \in \Omega^\ell(\R_y)$ and this gives rise to the stated isomorphisms.
\end{Example}

\begin{figure}[t]
\centering
\includegraphics[width=\textwidth, height=0.16\textheight, keepaspectratio]{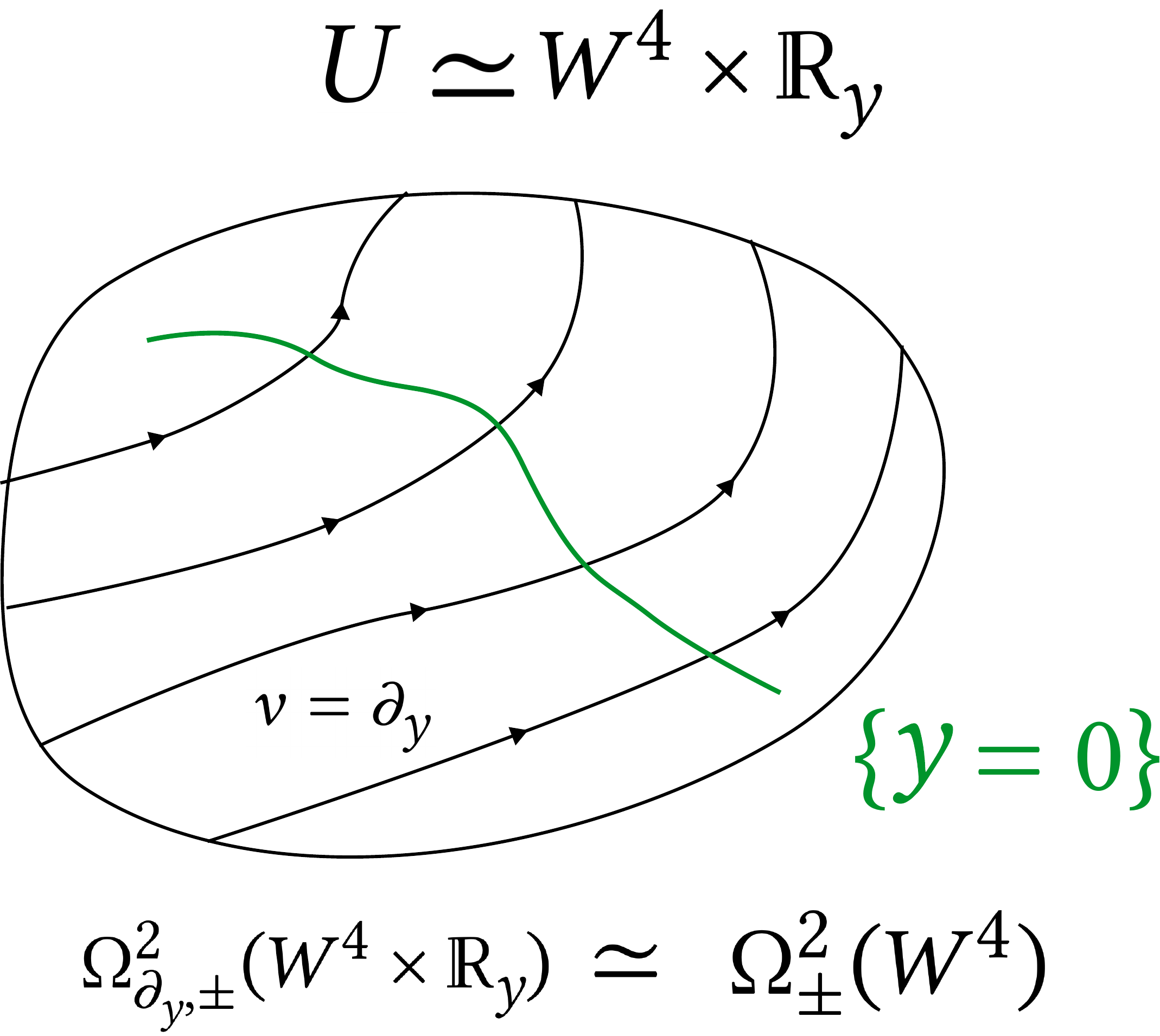}
\caption{}
\label{fig:background-Haydys-geometry}
\end{figure}

The two examples demonstrate that Haydys' setup provides a covariant lift of (anti-)self-dual 2-forms into 5-manifolds.
Given the relevance of gauge theory and anti-self-dual connections in the study of 4-manifolds, Haydys then suggests to consider a closely related set of equations in five-dimensional gauge theory.
The equations make use of the following differential for Haydys-self-dual two-forms
\begin{align*}
\begin{split}
	\delta_A^+ \colon\ \Omega^2_{v,+}(M, \ad E)
 \ \xrightarrow{\nabla^{A}}\ {}&T^\ast M \otimes \Omega^2_{v,+}(M, \ad E)
\simeq TM \otimes \Omega^2_{v,+}(M,\ad E)\\
\xrightarrow{-\imath} \ {} &\Omega^1(M, \ad E).
\end{split}
\end{align*}
Here, as usual, $\nabla^{A}$ is the covariant derivative with respect to both, the connection and the Levi-Civita connection, while we use $\imath$ to denote contraction $\imath(u\otimes \omega) := \imath_u \omega$.
In normal coordinates $(x^\mu, y)_{\mu=0,1,2,3}$ with $v=\del_y$, the action of this differential is given by \smash{$\delta_A^+ B = - \sum_{\mu=0}^{3} \nabla^{A}_\mu\; \imath_\mu B$}, since by construction $\imath_{\del_y} B=0$.

\begin{Definition}[Haydys--Witten equations]
Let \smash{$\bigl(M^5, v\bigr)$} be a Riemannian 5-manifold together with a preferred vector field and ${E \to M^5}$ a principal $G$-bundle.
With notation as above, consider a pair \smash{${(A,B) \in \mathcal{A}(E)\times \Omega^2_{v,+}(\ad E)}$}.
The Haydys--Witten equations for $(A,B)$ are given by
\begin{align}
	&F_A^+ - \sigma(B, B) - \nabla_v^A B= 0, \qquad \imath_v F_A - \delta_A^+ B= 0. \label{eq:background-Haydys--Witten-equations}
\end{align}
We denote the corresponding differential operator by
\begin{align*}
	\HW[v]\colon\ \mathcal{A}(E) \times \Omega_{v,+}^2(\ad E) \to \Omega_{v,+}^2(\ad E) \times \Omega^1(\ad E).
\end{align*}
\end{Definition}

If $B=0$ the Haydys--Witten equations provide a five-dimensional analogue of the anti-self-dual equations that underlie Donaldson--Floer theory on 4-manifolds.
This special case was studied from a slightly different perspective in~\cite{Fan1996}.

As mentioned already above, the perspective from supersymmetric Yang--Mills theory suggests that the Haydys--Witten equations have an interpretation in the six-dimensional $\mathcal{N}=(2,0)$ SCFT.
In fact they can be obtained via dimensional reduction in several, closely related ways: from six-dimensional equations~\cite[Section~5]{Witten2011}, from octonionic monopole equations on seven-manifolds with $G_2$ holonomy, or from $\operatorname{Spin}(7)$-instantons on eight-manifolds~\cite{Cherkis2015}.
The octonionic structure is closely related to the 10d $\mathcal{N}=1$ super-Poincar\'e algebra.

\subsection{The Kapustin--Witten equations}
\label{sec:background-Kapustin-Witten-equations}
The Kapustin--Witten equations are a one-parameter family of partial differential equations on Riemannian four-manifolds $W^4$.
They were first described by Kapustin and Witten in the context of the geometric Langlands program and its interpretation in quantum field theory, in~which case one considers a product of two Riemann surfaces $W^4 = \Sigma \times C$~\cite{Kapustin2007}.
A few years later the equations resurfaced in the context of the gauge theoretic approach to Khovanov homology on manifolds of the form $W^4 = X^3 \times \R^+$~\cite{Witten2011}.
Since then their study has grown into an active field of research, for an incomplete list of recent developments see, e.g.,~\cite{Gagliardo2012, He2015, He2019d, Nagy2021, Tanaka2019, Taubes2013, Taubes2017a}.

Let $\bigl(W^4,g\bigr)$ be a smooth Riemannian four-manifold and $G$ a compact Lie group.
Consider a principal $G$-bundle $E\to W^4$ together with a principal connection $A$, and denote by $\ad E$ the adjoint bundle associated to the Lie algebra $\mathfrak{g}$ of $G$.
Furthermore, consider an $\ad E$-valued one-form $\phi \in \Omega^1\bigl(W^4, \ad E\bigr)$, usually called the Higgs field.
The $\theta$-Kapustin--Witten equations for the pair $(A,\phi)$ and an angle $\theta \in [0,2\pi]$ are the following family of differential equations.
\begin{Definition}[Kapustin--Witten equations] 
\begin{align}
&	\bigl( \cos \tfrac{\theta}{2} \bigl( F_A - \thalf [\phi \wedge \phi] \bigr) - \sin \tfrac{\theta}{2} {\rm d}_A \phi \bigr)^+= 0,\nonumber \\
&	\bigl( \sin \tfrac{\theta}{2} \bigl( F_A - \thalf [\phi \wedge \phi] \bigr) + \cos \tfrac{\theta}{2} {\rm d}_A \phi \bigr)^-= 0,\nonumber \\
&	{\rm d}_A^{\hodge_4} \phi= 0.\label{eq:background-Kapustin--Witten-equations}
\end{align}
The corresponding differential operator is denoted as
\begin{align*}
\KW[\theta]\colon \ \mathcal{A}(E) \times \Omega^1\bigl(X^4, \ad E\bigr) \to \Omega^2_+\bigl(X^4, \ad E\bigr)\times \Omega^2_-\bigl(X^4, \ad E\bigr) \times \Omega^0\bigl(X^4, \ad E\bigr).
\end{align*}
\end{Definition}

For $\theta \neq 0 \pmod{\pi}$, the self-dual and anti-self-dual parts in \eqref{eq:background-Kapustin--Witten-equations} can be combined into the following single expression
\begin{align} \label{eq:background-Kapustin--Witten-equations-combined}
	F_{A} - \frac12 [\phi\wedge\phi] + \cot\theta {\rm d}_{A} \phi - \csc\theta \hodge_4 {\rm d}_A \phi	= 0.
\end{align}
Furthermore, as observed already by Kapustin and Witten~\cite{Kapustin2007} and discussed in more detail by Gagliardo and Uhlenbeck~\cite{Gagliardo2012}, this is equivalent to phase-shifted conjugate anti-self-dual equations for the $G_{\C}$-connection $A + {\rm i}\phi$:
\begin{align}\label{eq:background-Kapustin--Witten-equations-cplx-asd}
	F_{A + {\rm i} \phi} = {\rm e}^{{\rm i} (\pi - \theta)} \hodge_4 \overline{ F_{A + {\rm i} \phi} }.
\end{align}
This point of view suggests the applicability of several powerful results from the theory of self-dual Yang--Mills connections and geometric analysis in general.
For example, it is apparent from~\eqref{eq:background-Kapustin--Witten-equations-cplx-asd} that in the case $\phi = 0$ the one-parameter family of Kapustin--Witten equations interpolates between Donaldson's anti-self-dual equations ($\theta = 0$) and self-dual equations~(${\theta = \pi}$).

\begin{Remark}
In the literature, it is more common to parametrize the family of Kapustin--Witten equations by ${\theta_{\rm GU} = \theta/2 \in [0,\pi/2]}$, see, e.g.,~\cite{Gagliardo2012}.
However, when viewed as dimensional reduction of the Haydys--Witten equations as explained in Section~\ref{sec:background-dimensional-reduction}, the parameter $\theta$ obtains a~geometric interpretation as angle between the non-vanishing vector field $v$ and the direction of invariance.
This motivates the slightly non-standard choice of normalization used throughout this article.
The naturality of $\theta$, as opposed to $\theta_{\mathrm{GU}}$, can also be seen from equations~\eqref{eq:background-Kapustin--Witten-equations-combined}~and~\eqref{eq:background-Kapustin--Witten-equations-cplx-asd}, where our normalization avoids an additional factor of two.
\end{Remark}

At the midpoint $\theta = \pi/2$, the equations are usually referred to as ``the'' Kapustin--Witten equations and written more succinctly as
\begin{align*}
	&F_A - \tfrac12 [\phi \wedge \phi] - \hodge_4 {\rm d}_A \phi= 0, \qquad
	{\rm d}_A^{\hodge_4} \phi = 0.
\end{align*}

In their original article, Kapustin and Witten realized that on closed manifolds the study of solutions is relatively simple.

\begin{Theorem}[\cite{Gagliardo2012, Kapustin2007}]
\label{thm:background-Kapustin--Witten-vanishing}
Let $E\to W^4$ be an ${\rm SU}(2)$ principal bundle over a compact manifold without boundary.
Assume $(A,\phi)$ satisfies the $\theta$-Kapustin--Witten equations with $\theta \in (0,\pi)$.
If~${E\to W^4}$ has non-zero Pontryagin number, then $A$ and $\phi$ are identically zero.
Otherwise~${A+ {\rm i} \phi}$ is a flat ${\rm PSL}(2,\mathbb{C})$ connection; equivalently $F_A = \thalf [\phi\wedge\phi]$ and $\nabla^A \phi = 0$.
\end{Theorem}

As a consequence, one mostly concentrates on open manifolds or manifolds with boundary.
Recently, solutions of Kapustin--Witten equations have been studied on $\R^4$ and it turns out that also in this case solutions are remarkably constrained.
\begin{Theorem}[Taubes' dichotomy~\cite{Taubes2017a}]
\label{thm:background-taubes-dichotomy}
Let $W^4 = \R^4$, $G={\rm SU}(2)$, and define
\[
\kappa^2(r) := r^{-3} \int_{\del B_r} \norm{\phi}^2.
\]
Assume that $(A,\phi)$ is a solution of the Kapustin--Witten equations, then either there is an $a>0$ such that ${\liminf_{r\to\infty} \kappa/r^a > 0}$ or $A$ is flat, $\nabla^A \phi=0$, and $[\phi\wedge\phi] = 0$.
\end{Theorem}
In a similar way, the Kapustin--Witten energy
\begin{align*}
	E_{\text{KW}} = \int_{W^4} \bigl( \norm{F_A}^2 + \norm{\nabla^A \phi}^2 + \norm{[\phi\wedge\phi]}^2 \bigr)
\end{align*}
offers some degree of control over the asymptotic behaviour of the Higgs field $\phi$ on manifolds with controlled asymptotic volume growth (ALX spaces).
\begin{Theorem}[Kapustin--Witten solutions on ALE and ALF spaces {\cite[Main Theorem 1]{Nagy2021}}]
Assume $W^4$ is an ALE or ALF gravitational instanton and $(A,\phi)$ is a finite energy solution of the $\theta$-Kapustin--Witten equations, $\theta\neq 0 \pmod{\pi}$.
Then $\phi$ has bounded $L^2$-norm.
\end{Theorem}
Since bounded $L^2$-norm in particular implies bounded $L^2$-average on spheres, it follows that for finite energy solutions on $\R^4$ and $S^1 \times \R^3$ the function $\kappa$ in Theorem~\ref{thm:background-taubes-dichotomy} converges to $0$ as~${r\to\infty}$.
Taubes' dichotomy then yields the following corollary.
\begin{Corollary}\label{cor:background-Nagy-Oliveira-Conjecture}
On $\R^4$, $\R^3\times S^1$ and compact manifolds, solutions of the $\theta$-Kapustin--Witten equations with finite positive energy are such that $A$ is flat, $\nabla^A \phi = 0$, and $[\phi\wedge\phi]=0$.
\end{Corollary}
A generalization of Taubes' dichotomy to any ALE or ALF space was proposed in~\cite{Bleher2023a}.
With this Corollary~\ref{cor:background-Nagy-Oliveira-Conjecture} applies to ALE, ALF, and compact manifolds.

\subsection{The Vafa--Witten equations}
The Vafa--Witten equations are partial differential equations on a four-manifold $W^4$.
They were first discovered by Vafa and Witten in the context of 4d $\mathcal{N}=2$ super Yang--Mills theory~\cite{Vafa1994}.
Their solutions give rise to topological invariants of four-manifolds.
The equations have since been subject to close scrutiny both in physics and mathematics, see~\cite{Gukov2022, Mares2010, Ong2022, Tanaka2017, Tanaka2019, Taubes2017} and references therein.

Let $\bigl(W^4,g\bigr)$ be a smooth Riemannian four-manifold and $G$ a compact Lie group.
Consider a principal $G$-bundle $E\to W^4$ together with a connection $A$, and denote by $\ad E$ the adjoint bundle associated to the Lie algebra $\mathfrak{g}$ of~$G$.
Let \smash{$B\in \Omega^2_+\bigl(W^4\bigr)$} be a self-dual two-form (with respect to the Hodge star operator) and \smash{$C\in\Omega^0\bigl(W^4,\ad E\bigr)$} a section of $\ad E$.
In complete analogy to the situation in Haydys' setup, there is a three-dimensional cross product $\sigma(\cdot,\cdot) = \thalf (\cdot \times \cdot) \otimes [\cdot, \cdot]_{\mathfrak{g}}$ on \smash{$\Omega^2_+\bigl(W^4, \ad E\bigr)$}.

The Vafa--Witten equations for the triple $(A,B,C)$ are the following partial differential equations.
\begin{Definition}[Vafa--Witten equations] 
\begin{align}\label{eq:background-Vafa--Witten-equations}
	&F_A^+ - \sigma(B,B) - \tfrac12 [C,B]= 0,\qquad
	{\rm d}_A^{\hodge 4} B + {\rm d}_A C = 0.
\end{align}
We denote the associated differential operator by $\VW(A,B,C)$.
\end{Definition}

The Vafa--Witten equations are closely related to the Kapustin--Witten equations.
For one, we will later describe in detail that both arise as dimensional reduction of the Haydys--Witten equations.
Furthermore, it is well known that on Euclidean space the $\theta=0$ version of the Kapustin--Witten equations and the Vafa--Witten equations are equivalent.
The correspondence arises by combining the components of \smash{$B = \sum_{i=1}^3 \phi_i \bigl({\rm d}x^0\wedge {\rm d}x^i + \frac{1}{2} \epsilon_{ijk} {\rm d}x^j\wedge {\rm d}x^k\bigr)$} and $C$ into a~one-form $\phi = C {\rm d}x^0 + \phi_i {\rm d}x^i$.

\subsection{The extended Bogomolny equations}

The extended Bogomolny equations (EBE) are a set of partial differential equations on a three-manifold $X^3$.
As the name suggests, they are an extension of Bogomolny's magnetic monopole equations on three-manifolds~\cite{Atiyah1988a, Bogomolny1976, Jaffe1980}.
The latter are a dimensional reduction of the self-dual Yang--Mills equations from four dimensions to three dimensions.
From a purely three-dimensional point of view, the EBE can be viewed as complexification of the Bogomolny equations~\cite{Nagy2020}.

In the context of Haydys--Witten theory, the EBE appear as a one-parameter family, called twisted extended Bogomolny equations (TEBE), obtained by dimensional reduction of the Haydys--Witten (or Kapustin--Witten) equations to three dimensions.
They play an important role in the definition of the Nahm pole boundary conditions in the presence of knots, because the 't Hooft operator, upon reduction to three dimensions, reduces to the insertion of a~monopole-like solution of the EBE.

Let \smash{$\bigl(X^3,g\bigr)$} be a Riemannian three-manifold and $E\to X^3$ a $G$-principal bundle.
Consider a connection $A$, a one-form \smash{$\phi \in \Omega^1\bigl(X^3, \ad E\bigr)$}, and two functions \smash{$c_1, c_2 \in \Omega^0\bigl(X^3, \ad E\bigr)$}.
The extended Bogomolny equations are the following set of differential equations for the tuple~$(A,\phi,c_1, c_2)$:
\begin{Definition}[extended Bogomolny equations]
\begin{align}
	&F_A - \tfrac12 [\phi\wedge\phi] + \hodge_3 ( {\rm d}_A c_1 - [c_2, \phi] ) = 0, \nonumber\\
	&{\rm d}_A c_2 + [c_1, \phi ] - \hodge_3 {\rm d}_A \phi = 0, \qquad
	{\rm d}_A^{\hodge_3} \phi - [c_1,c_2]=0.
	\label{eq:background-EBE}
\end{align}
\end{Definition}
We write $\EBE(A,\phi,c_1,c_2)$ for the associated differential operator.

As a side remark, note that the equations admit several specializations, each of which had important reverberations in contemporary mathematics.
\begin{itemize}\itemsep=0pt
	\item If $\phi = c_2 = 0$, these are the Bogomolny equations for a magnetic monopole $(A,c_1)$ on a~three-manifold $X^3$.
	\item If $X^3 = \Sigma \times \mathbb{R}_y$, $c_1 = c_2 = 0$, and $(A,\phi)$ are independent of $y$, the equations reduce to Hitchin's equations for a Higgs bundle over $\Sigma$~\cite{Hitchin1987}.
	Note that solutions to the $y$-invariant equations provide natural stationary boundary conditions at non-compact ends.
	\item If $X^3 = \Sigma \times \mathbb{R}_y$, $A = c_2 =0$, and $\phi$ is independent of $\Sigma$, they reduce to Nahm's equations~\cite{Nahm1983}.
	These will be described in more detail in the upcoming Section~\ref{sec:background-Nahm-equations}.
\end{itemize}

Since the extended Bogomolny equations can be obtained by a dimensional reduction of the~${\theta=\pi/2}$ version of the Kapustin--Witten equations (cf.\ Section~\ref{sec:background-dimensional-reduction}), it is clear that there should also be a one-parameter family of EBEs, defined by dimensional reduction for any value of~${\theta \in [0,\pi]}$.
The result is known as $\theta$-twisted extended Bogomolny equations (TEBE)~\cite{Dimakis2022}.
\begin{Definition}[twisted extended Bogomolny equations]
\begin{align}
	&F_A - \tfrac12 [\phi\wedge\phi] + \cot\theta\, {\rm d}_A \phi + \csc\theta \hodge_3 ( {\rm d}_A c_1 - [c_2, \phi] )= 0, \nonumber \\
	&{\rm d}_A c_2 + [c_1, \phi] + \cot\theta ( {\rm d}_A c_1 - [c_2, \phi] ) - \csc\theta \hodge_3 {\rm d}_A \phi= 0, \qquad
	{\rm d}_A^{\hodge_3} \phi - [c_1, c_2] = 0.
	\label{eq:background-TEBE}
\end{align}
We write $\TEBE[\theta](A,\phi,c_1,c_2)$ for the associated differential operator.
\end{Definition}
Observe that for $\theta = \pi/2$, one obtains the untwisted EBE of~\eqref{eq:background-EBE}, i.e., $\TEBE[\pi/2] = \EBE$.

Let us also mention the work of Nagy and Oliveira in~\cite{Nagy2020, Nagy2021}, where the TEBE are investigated at $\theta=0$ and $\pi/2$, respectively.
From their point of view, the two sets of equations arise from two ways to extend the Hodge star operator to the complexification of the principal bundle.
From our point of view, these are special points of a one-parameter family of extended Bogomolny equations, and their findings are analogues in three dimensions of the interpretation of the Kapustin--Witten equations as phase-shifted anti-self dual equations of a complex connection in four dimensions, see~\eqref{eq:background-Kapustin--Witten-equations-cplx-asd}.

Sometimes the EBE are defined only for a triple $(A,\phi,c_1)$.
The reason for this is the following result.
\begin{Proposition} \label{prop:background-EBE-vanishing}
Assume $(A,\phi, c_1, c_2)$ is an irreducible solution of the extended Bogomolny equations on a Riemannian three-manifold $X^3$.
If $X^3$ is closed, or if it has ends at which the fields satisfy boundary conditions such that $\int_{X^3} {\rm d} \Tr \left( c_1 \wedge {\rm d}_A c_2 - [c_1, c_2] \wedge \hodge_3 \phi \right) = 0$, then $c_2=0$.
\end{Proposition}
\begin{proof}
The proof was originally outlined in~\cite[p.~58]{Witten2011}.
The following, more explicit presentation of the argument is taken from~\cite{He2019c}.

If we set
\begin{align*}
	&I_1= \int_{X^3} \bigl\| F_A - \tfrac12 [\phi\wedge\phi] + \hodge_3 {\rm d}_A c_1 \bigr\|^2 + \norm{\hodge_3 {\rm d}_A \phi - [c_1, \phi]}^2 + \norm{{\rm d}_A^{\hodge_3} \phi}^2 , \\
	&I_2= \int_{X^3} \norm{[c_2, \phi]}^2 + \norm{ {\rm d}_A c_2}^2 + \norm{[c_1,c_2]}^2, \\
	&I_3= \int_{X^3} {\rm d} \Tr ( c_1 \wedge {\rm d}_A c_2 - [c_1, c_2] \wedge \hodge_3 \phi ),
\end{align*}
then there is a Weitzenb\"ock formula
\begin{align*}
	\int_{X^3} \norm{\EBE(A,\phi,c_1,c_2)}^2 = I_1 + I_2 + I_3 .
\end{align*}
By assumption, the boundary term $I_3 = 0$, either because $\del M= \varnothing$ or because the boundary conditions on $(c_1,c_2)$ are exactly such that $I_3$ vanishes.
The remaining terms in the Weitzenb\"ock formula are non-negative, so any solution of the extended Bogomolny equations also satisfies $I_1 = I_2 = 0$.
If $c_2\neq 0$, vanishing of the terms in $I_2$ implies that the pair $(A,\phi)$ is reducible, which is in contradiction to our assumption and concludes the proof.
\end{proof}

An astoundingly important fact about the EBE is that over three-manifolds of the form $X^3=\Sigma\times\R$, for some Riemann surface $\Sigma$, they admit a Hermitian Yang--Mills structure~\cite{Witten2011}.
To see this, one introduces three differential operators
\begin{align*}
	\mathcal{D}_1 = \nabla^A_1 + {\rm i} \nabla^A_2, \qquad
	\mathcal{D}_2 = [\phi_1, \,\cdot\,] + {\rm i} [\phi_2, \,\cdot\,], \qquad
	\mathcal{D}_3 = \nabla^A_3 + {\rm i} [\phi_3, \,\cdot\,].
\end{align*}
The extended Bogomolny equations are then equivalent to
\begin{align*}
	[\mathcal{D}_i, \mathcal{D}_j] = 0,
	\qquad
	\sum_{i=1}^3 [\overline{\mathcal{D}_i}, \mathcal{D}_i] = 0 .
\end{align*}
These equations are equivalent to the Hermitian Yang--Mills equation for a Hermitian connection on a holomorphic vector bundle over $\Sigma\times \R$.

More generally, the $\theta$-TEBE exhibit a Hermitian Yang--Mills structure if $c_2 - \tan(\beta/3) c_1 = 0$, where $\beta$ denotes the complementary angle $\beta = \pi/2-\theta$~\cite{Gaiotto2012a}.
Note that for $\theta=\pi/2$ this condition is automatically satisfied for all irreducible solutions due to Proposition~\ref{prop:background-EBE-vanishing}.

The Hermitian Yang--Mills structure is an important tool in the classification of solutions of the EBE and TEBE.
It has been used to prove a Kobayashi--Hitchin correspondence between solutions of the EBE and Higgs bundles with certain extra structure~\cite{Gaiotto2012a, He2019c, He2020b}.
First analogous results regarding the classification of TEBE-solutions have been achieved in~\cite{Dimakis2022, He2019b}.
In fact, in favourable circumstances a closely related Hermitian Yang--Mills structure also exists for a~decoupled version of the Haydys--Witten and Kapustin--Witten equations~\cite{Bleher2023b}.
A first attempt to utilize this Hermitian Yang--Mills structure to classify solutions of the decoupled equations in terms of geometric objects is described in the final chapter of the author's PhD thesis~\cite{Bleher2023c}.

\subsection{Nahm's equations}
\label{sec:background-Nahm-equations}

Nahm's equations are ordinary, non-linear differential equations for a vector bundle of $\mathfrak{g}$-valued functions over a one-dimensional interval $I$.
The equations rely on the existence of a cross product on the vector bundle.
As is well known, real vector spaces with a (bilinear) cross product are in correspondence with the imaginary part of normed division algebras.
Correspondingly, the associated equations are referred to as complex, quaternionic, and octonionic Nahm equations, depending on the rank of the vector bundle.

The equations were originally introduced in the quaternionic case by Nahm and play an~important role in the classification of monopoles~\cite{Donaldson1984, Hitchin1983, Nahm1983}.
They can be viewed as dimensional reductions either from anti-self dual Yang--Mills equations on four-manifolds or equivalently from Bogomolny's monopole equations on three-manifolds.
The octonionic Nahm equations first appear in~\cite{Grabowski1993}.
They have recently attracted renewed attention in the context of $\operatorname{Spin}(7)$-instantons on eight-dimensional manifolds and monopoles on seven-dimensional manifolds with~$G_2$ holonomy~\cite{Charbonneau2022a, Cherkis2015}.
Finally, in a recent article He discusses some properties of the moduli space of solutions for the octonionic Nahm equations~\cite{He2020a}.

In the context of Haydys--Witten theory, a dimensional reduction of the Haydys--Witten equations from five to one dimension in general gives rise to a twisted version of the octonionic Nahm equations.
Since the Haydys--Witten equations can be viewed as a lift of self-dual Yang--Mills equations to five dimensions, it is not too surprising that a variant of Nahm's equations appears.
The twisted octonionic Nahm equations will play an important role in the definition of the twisted (or tilted) Nahm pole boundary conditions with knot singularities in~Section~\ref{sec:background-Nahm-pole-boundary-condition}.

Below, we first provide the general definition of Nahm's equations and state explicit formulae for each case, mostly following a similar exposition in~\cite{He2020a}.
Subsequently, the twisted octonionic equations are introduced, where we explain the underlying structure from the point of view of octonionic multiplication.

Let $G$ be a compact Lie group with Lie algebra $\mathfrak{g}$ and $I$ a real interval with coordinate~$y$.
Consider a trivial $G$-principal bundle $E$ over $I$ with connection $A = A_y {\rm d}y$ and denote the associated covariant derivative by \smash{$\nabla^A_y = \frac{{\rm d}}{{\rm d}y} + [A_y, \cdot]$}.
Let $\mathbb{V}$ be one of the normed division algebras~$\C$,~$\mathbb{H}$,~$\mathbb{O}$.
Write $(\operatorname{Im} \mathbb{V}, \times)$ for its imaginary part together with the cross product induced by multiplication.

We consider the trivial bundle $\operatorname{Im} \mathbb{V}\otimes \ad E \to I$.
Let the tuple $\vec{X} = (X_1, \ldots, X_k)$ denote a~section of this bundle, where $k = \dim_{\mathbb{R}}\mathbb{V} -1$.
In particular, the components $X_i$ are $\mathfrak{g}$-valued functions over $I$.
This bundle admits a cross product, induced by the cross product on $\operatorname{Im} \mathbb{V}$ and the Lie bracket on $\mathfrak{g}$ and given by
\begin{align*}
	\vec{X} \times \vec{X}
	= \sum_{i,j,k} f_{ijk}\ e_i \otimes [X_j, X_k],
\end{align*}
where $f_{ijk}$ are the structure constants of $\mathbb{V}$.

Let $E$ be a $G$-principal bundle over an interval $I$ with connection $A$ and denote a section $\vec{X} \in \Gamma(I, \operatorname{Im} \mathbb{V}\otimes \mathfrak{g})$.
The Nahm equations associated to $\mathbb{V}$ are the following system of non-linear, ordinary differential equations for the pair $(A,\vec{X})$.

\begin{Definition}[Nahm equations]
\samepage{\begin{align*}
	\nabla^A_y \vec{X} + \vec{X} \times \vec{X} = 0.
\end{align*}
We will occasionally write $\Nahm[\mathbb{V}](A, \vec{X})$ for the associated differential operator.}
\end{Definition}
\begin{Remark}
The division algebras only appear fiberwise, in the multiplication of sections.
In particular, the underlying differential equations are based in real analysis, as opposed to complex, quaternion, or octonion analysis.
Furthermore, Nahm's equations only make use of pairwise products, such that the non-associativity of the octonions fortunately need not be taken into consideration.
\end{Remark}

For $\mathbb{V} = \C$, the section $\vec{X}$ corresponds to a single $\ad E$-valued function $X_1$, while the cross product on $\operatorname{Im} \C$ is the zero map.
Hence, the complex Nahm equation is the single equation
\begin{align*}
	\nabla^A_y X_1 = 0,
\end{align*}
which is just the statement that $X_1$ is covariantly constant along $I$.

For $\mathbb{V} = \mathbb{H}$, the section consists of three components $\vec{X} = (X_1 , X_2, X_3)$ and the structure constants are the completely anti-symmetric tensor $\epsilon_{ijk}$, for $i,j,k \in \{1,2,3\}$, which is~1 when~${ijk=123}$,
\begin{align*}
	\nabla^A_y X_i + \tfrac{1}{2} \epsilon_{ijk} [X_j, X_k] = 0.
\end{align*}
The quaternionic Nahm equations are typically simply known as ``the'' Nahm equations.
An~important set of solutions are Nahm pole solutions $(A,\vec{X}) = (0, X_1, X_2, X_3)$, with $\smash{X_i = \frac{\mathfrak{t}_i}{y}}$, where $\mathfrak{t}_i \in \mathfrak{g}$ satisfy $\mathfrak{su}(2)$ commutation relations.
	
For $\mathbb{V} = \mathbb{O}$, the section has seven components $\vec{X} = (X_1, \ldots, X_7)$.
A possible choice of structure constants\footnote{We use a slightly different convention than~\cite{He2020a}; the two choices differ by $e_7 \mapsto -e_7$.} for the octonions is given by the completely antisymmetric tensor $f_{ijk}$, $i,j,k \in \{1,\ldots 7\}$, that is $+1$ when $ijk$ is any of 123, 145, 176, 246, 257, 347, or 365.
Explicitly, the octonionic Nahm equations are given by
\begin{align*}
	&\nabla^A_y X_1 + [X_2, X_3] + [X_4, X_5] - [X_6, X_7] = 0, \\
	&\nabla^A_y X_2 - [X_1, X_3] + [X_4, X_6] + [X_5, X_7] = 0, \\
	&\nabla^A_y X_3 + [X_1, X_2] + [X_4, X_7] - [X_5, X_6] = 0, \\
	&\nabla^A_y X_4 - [X_1, X_5] - [X_2, X_6] - [X_3, X_7] = 0, \\
	&\nabla^A_y X_5 + [X_1, X_4] - [X_2, X_7] + [X_3, X_6] = 0, \\
	&\nabla^A_y X_6 + [X_1, X_7] + [X_2, X_4] - [X_3, X_5] = 0, \\
	&\nabla^A_y X_7 - [X_1, X_6] + [X_2, X_5] + [X_3, X_4] = 0.
\end{align*}
Note that solutions of the quaternionic Nahm equations give rise to solutions of the octonionic Nahm equations:
If $(A,X_1,X_2,X_3)$ is a solution of the quaternionic equations, then $(A,X_1,X_2,X_3,0,0,0,0)$ is a solution of the octonionic equations.
For more about the moduli space of the octonionic Nahm equations, see~\cite{He2020a}.

We will also come across a twisted version of the octonionic equations.
To explain this, it is convenient to first rename components as $\vec{X} = (X_1, X_2, X_3, Y, Z_1, Z_2, Z_3)$.
We can then express the octonionic Nahm equations in terms of $\epsilon_{ijk}$ with $i,j,k \in \{1,2,3\}$ and $\epsilon_{123} = 1$ as
\begin{align*}
\begin{split}
	&\nabla^A_y X_i + [Y, Z_i] + \epsilon_{ijk} \bigl( \tfrac12 [X_j, X_k] - \tfrac12 [Z_j, Z_k] \bigr) = 0, \\
	&\nabla^A_y Y - [X_i,Z_i] = 0 ,\\
	&\nabla^A_y Z_i - [Y, X_i] - \epsilon_{ijk} [X_j, Z_k] = 0.
\end{split}
\end{align*}
In this notation the twisted equations arise from mixing the terms that appear in the last column by a rotation of angle $\beta \in [0,\pi]$, as follows:
\begin{align*}
	\begin{split}
	&\nabla^A_y X_i + [Y, Z_i] + \epsilon_{ijk} \bigl( \cos\beta \bigl( \tfrac12 [X_j, X_k] - \tfrac12 [Z_j, Z_k] \bigr) + \sin\beta [X_j, Z_k] \bigr) = 0 ,\\
	&\nabla^A_y Y - [X_i,Z_i] = 0, \\
	&\nabla^A_y Z_i - [Y, X_i] - \epsilon_{ijk} \bigl( -\sin\beta \bigl( \tfrac12 [X_j, X_k] - \tfrac12 [Z_j, Z_k] \bigr) + \cos\beta [X_j, Z_k] \bigr) = 0.
	\end{split}
\end{align*}

The twisted equations may be viewed as the result of deforming the cross product on $\operatorname{Im} \mathbb{O}$.
To explain this, first recall that by the Cayley--Dickson construction octonions can be viewed as a product of the quaternions equipped with a particular multiplication.
Explicitly, if we denote the basis of the quaternions by $(1,i,j,k)$, we can identify the basis elements of $\mathbb{O}$ with $1 = (1,0)$, $e_1 = (i,0)$, $e_2=(j,0)$, $e_3=(k,0)$, $h=(0,1)$, $f_1 =(0,i)$, $f_2=(0,j)$, $f_3=(0,k)$.
As real vector spaces, the imaginary octonions can then be identified with the direct sum $\operatorname{Im}\mathbb{O} = \operatorname{Im} \mathbb{H} \oplus \mathbb{R} \oplus \operatorname{Im} \mathbb{H}$.
Observe that this corresponds precisely to the previous renaming of components $\vec{X} = ( \vec{X}, Y, \vec{Z})$.
Octonionic multiplication can be summarized as follows: $1$ is the unit element, for any other basis element set $(x_i)^2= -1$, and for the remaining imaginary products specify $x\times y := \frac{1}{2}(xy - yx)$ to be given by $e_i \times f_i = -h$ and when $i\neq j$ by
\begin{align*}
\begin{split}
	&e_i \times e_j = \epsilon_{ijk} e_k,\qquad
	f_i \times f_j = - \epsilon_{ijk} e_k,\qquad
	e_i \times f_j = - \epsilon_{ijk} f_k,\\
	&h \times e_i = - f_i,\qquad
	h \times f_i = e_i.
\end{split}
\end{align*}
Clearly, octonionic multiplication does not preserve the decomposition as real vector spaces.
In~particular, $h$ maps one of the $\operatorname{Im} \mathbb{H}$ factors into the other.
This is used to deform the cross product by a rotation between the two factors of $\operatorname{Im} \mathbb{H}$.
More precisely, whenever the product has values in one of the $\operatorname{Im} \mathbb{H}$'s, we post-compose it with the left-action of $\cos\beta 1 + \sin\beta h$, which corresponds to adjusting the following multiplications:
\begin{align*}
&e_i \times_\beta e_j
 = ( \cos\beta 1 + \sin\beta h ) (e_i \times e_j)
		 = \epsilon_{ijk} ( \cos\beta e_k - \sin\beta f_{k} ),\\
	&f_{i} \times_\beta f_{j}
 = ( \cos\beta 1 + \sin\beta h ) (f_i \times f_j)
	 = - \epsilon_{ijk} ( \cos\beta e_{k} - \sin\beta f_{k} ) ,\\
	&e_{i} \times_\beta f_{j}
 = ( \cos\beta \; 1 + \sin\beta h ) (e_i \times f_j)
		 = - \epsilon_{ijk} ( \sin\beta e_k + \cos\beta f_{k} ).
\end{align*}
Here $i,j,k \in \{1,2,3\}$ and we sum over repeated indices.
All other products remain unchanged.

With this deformation the twisted version of the octonionic Nahm equations can be succinctly defined by the following equations.
\begin{Definition}[twisted octonionic Nahm equations] 
\begin{align*}
	\nabla^A_y \vec{X} + \vec{X} \times_\beta \vec{X} = 0.
\end{align*}
More explicitly, writing $\vec{X} = (X_1,X_2,X_3,Y,Z_1,Z_2,Z_3)$, the $\beta$-twisted octonionic equations expand to
\begin{align}
	&\nabla^A_y X_i + [Y,Z_i] + \epsilon_{ijk} \bigl( \cos\beta \bigl(\tfrac12 [X_j, X_k] - \tfrac12 [Z_j, Z_k] \bigr) + \sin\beta [X_j, Z_k] \bigr)= 0 ,\nonumber\\
	&\nabla^A_y Z_i - [Y,X_i] - \epsilon_{ijk} \bigl( -\sin\beta \bigl(\tfrac12 [X_j, X_k] - \tfrac12 [Z_j, Z_k] \bigr) + \cos\beta [X_j, Z_k] \bigr)= 0 ,\nonumber\\
	&\nabla^A_y Y - [X_i,Z_i]= 0 	.
	\label{eq:background-twisted-Nahm-equations-expanded}
\end{align}
We will occasionally denote the associated differential operator by $\Nahm[\mathbb{O}, \beta](A, \vec{X})$.
\end{Definition}

The embedding of quaternionic solutions into the moduli space of octonionic solutions carries over to the twisted case by rotating $\vec{X}$ into $\vec{Z}$:
If $(A,\vec{X}) = (A,X_1,X_2,X_3)$ is a solution of the quaternionic Nahm equations, then
\begin{align*}
	&(A,\cos\beta \vec{X},0,\sin\beta \vec{X}^\tau)
:= (A,\cos\beta X_1,\cos\beta X_2,\cos\beta X_3,0, \sin\beta X_1,\sin\beta X_3,\sin\beta X_2)
\end{align*}
is a solution of the $\beta$-twisted octonionic Nahm equations.
Note that the identification of $\vec{X}$ with~$\vec{Z}$ involves an anti-cyclic permutation of components, denoted by $\tau = (132)$.

\section{Dimensional reductions of the Haydys--Witten equations}
\label{sec:background-dimensional-reduction}

As mentioned before, each of the equations presented in Section~\ref{sec:background-equations} can be viewed as a dimensional reduction of the Haydys--Witten equations.
Here we explicitly perform the reduction steps and explain how the various one-parameter families of equations arise naturally from Haydys' five-dimensional geometry.

Throughout this section, we denote the Haydys--Witten fields by $\bigl(\hat A , \hat B\bigr)$ to explicitly distinguish them from four-dimensional fields $A$, $B$ and $\phi$ and three-dimensional fields $\tilde A$ and $\tilde \phi$.
For convenience, let us repeat the Haydys--Witten equations, as defined in equation~\eqref{eq:background-Haydys--Witten-equations}.
\begin{align}
	&F_{\hat A}^+ - \sigma\bigl(\hat B, \hat B\bigr) - \nabla^{\hat A}_v \hat B= 0,\qquad
	\imath_v F_{\hat A} - \delta_{\hat A}^+ \hat B = 0.
	\label{eq:background-dim-red-haydys-witten-equations}
\end{align}
Dimensional reduction on $M^5 = \R^k \times Y^{5-k}$ assumes that the fields $\bigl(\hat A, \hat B\bigr)$ and all gauge transformations are independent of the position in $\R^k$.
Equivalently, if we write $u_i$, $i=1,\ldots, k$, for a set of orthogonal unit vector fields on $\R^k$, then $\hat A$ and $\hat B$ are invariant under the action of all~$u_i$'s.
Note that this only makes sense if the glancing angles $g(u_i,v)= \cos \theta_i$ between~$u_i$ and the distinguished vector field $v$ are constant, since otherwise the varying angles will introduce an~explicit position-dependence in $\R^k$.
Below, we discuss dimensional reduction for $k=1$, $2$ and~$4$, which leads to the $\theta$-Kapustin--Witten equations, $\theta$-twisted extended Bogomolny equations, and $\beta$-twisted octonionic Nahm equations, respectively.

\subsection{\texorpdfstring{$\mathbb{R}$-invariant solutions}{R-invariant solutions}}\label{sec:background-dimensional-reduction-KW}

Consider a product space $M^5 = \R_s \times W^4$ equipped with a product metric and denote the inclusion of $W^4$ at $s=0$ by ${i\colon W^4 \hookrightarrow \R_s \times W^4 }$.
Let $u:=\del_s$ be the unit vector field tangent to $\mathbb{R}_s$ and assume that ${g(u,v) = \cos \theta}$ is constant.
The angle $\theta$ can either be thought of as the glancing angle between $v$ and the direction of invariance $\R_s$, or equivalently as incidence angle of $v$ on the hypersurface $W^4$, see Figure~\ref{fig:background-incidence-angles-0}.

\begin{figure}[!ht]
\centering
\includegraphics[width=\textwidth, height=0.1\textheight, keepaspectratio]{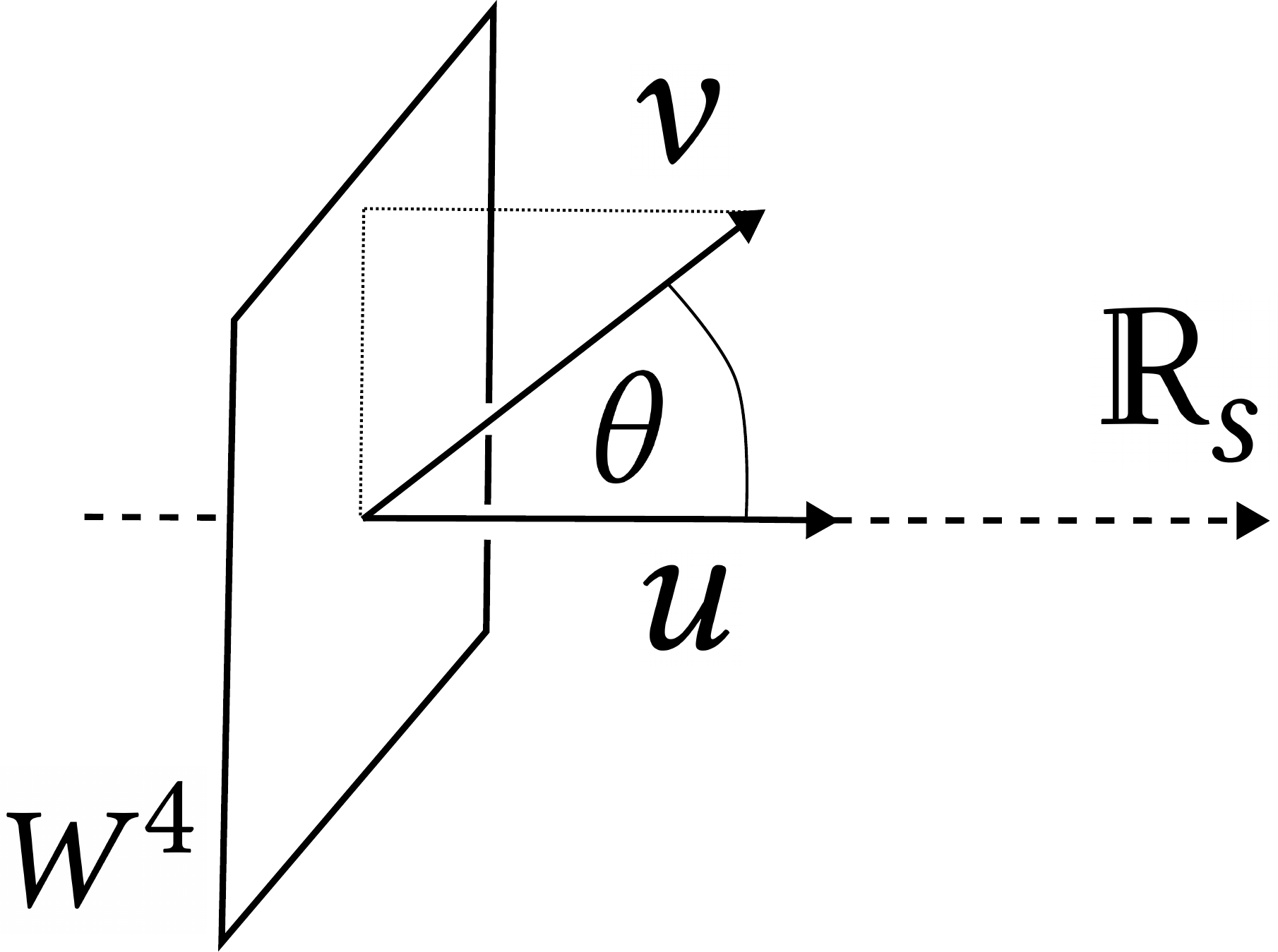}%
\caption{}%
\label{fig:background-incidence-angles-0}%
\end{figure}

As explained in Section~\ref{sec:background-Haydys--Witten-equations}, Haydys' geometric setup provides a lift of four-dimensional \mbox{(anti-)self-dual} two-forms into five dimensions.
Loosely speaking, this was achieved by identifying the orthogonal complement of $v$ with the tangent space of a four-manifold $W^4$ that foliates~$M^5$.
When $u$ and $v$ are aligned, dimensional reduction simply recovers the self-dual two-forms.
However, in general $u$ and $v$ are not necessarily aligned and the consequences of imposing $u$-invariance depend on the interplay between the orthogonal complements of $u$ and $v$.
To that end observe that $u$ and $v$ define a distribution $\Delta_{(u,v)} \subset TM$.
Since $u$ and $v$ are non-vanishing and $\theta$ is constant, $\Delta_{(u,v)}$ is regular, i.e., a vector bundle.
We now have to distinguish two cases: ${\theta \equiv 0 \pmod{\pi}}$ and $\theta\neq 0$.

If $\theta\equiv 0 \pmod{\pi}$ or equivalently if $u=\pm v$, the vector bundle $\Delta_{(u,v)}$ has rank one.
In this case, the orthogonal complements of $u$ and $v$ are identical, perhaps up to a reversal of orientation, and dimensional reduction is fairly straightforward.
The connection splits into $\hat A = \hat{A}_s {\rm d}s + A $, where $A = i^\ast \hat A$ is a connection over $W^4$.
The component $\hat{A}_s$ can be reinterpreted as an $\ad E$-valued function $C \in \Omega^0(\ad E)$, since gauge transformations are assumed to be $u$-invariant.\footnote{In general $\imath_u \hat{A}$ transforms as \smash{$\imath_u \hat{A} \mapsto g^{-1} \imath_u \hat{A} g + g^{-1} \nabla^{\hat{A}}_u g$} under gauge transformations.
If the gauge transformation $g$ is $u$-invariant the second term vanishes and $\imath_u \hat{A}$ is equivalent to an $\ad E$-valued function.}
Regarding~$\hat B$, recall from Section~\ref{sec:background-Haydys--Witten-equations} that there is an isomorphism \smash{$i^\ast \bigl( \Omega^2_{\del_s, +}\bigl(\R_s \times W^4\bigr) \bigr) \simeq \Omega^2_+\bigl(W^4\bigr)$}.
As~a~result \smash{${\bigl(\hat A, \hat B\bigr)}$} pull back to a triple of fields $(A, B, C)$ on the four-manifold $W^4$, where $A$ is a connection, $B$ is a self-dual two-form, and $C$ is an $\ad E$-valued function.
With these identifications in place, note that \smash{$\nabla^{\hat A}_v \hat{B} = [C\wedge B]$}, \smash{$\imath_v F_{\hat A} = {\rm d}_A C$}, and \smash{$\delta_{\hat A}^+ \hat{B} = {\rm d}_A^{\hodge_4} B$}.
Plugging this into the Haydys--Witten equations for \smash{$\bigl(\hat A,\hat B\bigr)$} immediately yields the Vafa--Witten equations~\eqref{eq:background-Vafa--Witten-equations} for the triple $(A,B,C)$.

If $\theta \not\equiv 0 \pmod{\pi}$, the vector bundle $\Delta_{(u,v)}$ has rank two.
The tangent bundle splits into orthogonal complements \smash{\raisebox{0.5pt}{$TM = \Delta_{(u,v)} \oplus \Delta_{(u,v)}^\perp$}} and this induces a similar decomposition for~one-forms as sections of \smash{$\bigl(\Delta_{(u,v)}\bigr)^\ast \oplus \bigl(\Delta_{(u,v)}^\perp\bigr)^\ast$}.
Moreover, since $\Delta_{(u,v)}$ admits the two linearly independent global sections $u$ and $v$, it is trivial.
It follows that there is a non-vanishing vector field $w$ that together with $u$ provides an orthonormal basis of \smash{$\Delta_{(u,v)}$}.
In this basis, $v$ is given as $v= \cos\theta u + \sin\theta w$ and we can define $v^\perp := -\sin\theta u + \cos\theta w$, which is the unique (up to a sign) unit vector field in $\Delta_{(u,v)}$ that is orthogonal to $v$.

Crucially, contraction with $v^\perp$ provides an isomorphism between $u$-invariant self-dual two-forms and sections of \smash{$\bigl(\Delta_{(u,v)}^\perp\bigr)^\ast$}.
One way to see this is to observe that locally the following two-forms provide a basis of \smash{$\Omega^2_{v,+}$} (cf.\ Section~\ref{sec:background-Haydys--Witten-equations}):
\begin{align*}
	e_i = \eta^\perp \wedge {\rm d}x^i + \tfrac12 \epsilon_{ijk} {\rm d}x^j \wedge {\rm d}x^k ,\qquad i = 1,2,3 ,
\end{align*}
where $\eta^\perp$ is the (global) one-form dual to $v^\perp$ and ${\rm d}x^i$ are (local) sections of \smash{$\bigl(\Delta_{(u,v)}^\perp\bigr)^\ast$}.
If~we locally write \smash{$\hat B = \sum_i \phi_i e_i$}, contraction with $v^\perp$ yields a one-form \smash{$\imath_{v^\perp} \hat B = \sum_i \phi_i {\rm d}x^i$}.
Using this~isomorphism, the $u$-invariant fields \smash{$\bigl(\hat A, \hat B\bigr)$} on $\R_s \times W^4$ can be reinterpreted as a connection~$A$ and an~$\ad E$-valued one-form $\phi$ on \smash{$W^4$}.

To make this more explicit, consider for the moment the example of Euclidean space $\R_s\times \R^4$ with Cartesian coordinates $(s,t,x^a)_{a=1,2,3}$, chosen in such a way that
\begin{align*}
	u=\del_s ,\qquad w=\del_t ,\qquad \text{and}\qquad v= \cos\theta \del_s + \sin\theta \del_t.
\end{align*}
The connection can be split into $\hat A= \hat{A}_s {\rm d}s + A$, where $A$ is the part of the connection on $\mathbb{R}^4$.
Furthermore, we can combine the remaining component $\hat{A}_s$ with the 3 components $\phi_i$ of $B$ into a one-form ${\phi = \hat{A}_s {\rm d}t + \sum_i \phi_i {\rm d}x^i}$ on $\mathbb{R}^4$.

The definitions in the Euclidean case are the local model underlying the following identifications for general manifolds $\R_s \times W^4$:
\begin{align*}
	A := i^\ast \hat{A} , \qquad
	\phi := \hat{A}_s \wedge w^\flat + \imath_{v^\perp} \hat{B}.
\end{align*}
\begin{Remark}
To shed some light on the definition of $\phi$, it might be helpful to directly compare the situations for $\theta = 0$ and $\theta\neq 0$.
In both cases, the pullback of $\hat{A}$ provides a connection on the pullback bundle $i^\ast E\to W^4$ and projects out the component $\hat{A}_s$, which on its own can be viewed as an $\ad E$-valued function.
The difference arises in the reinterpretation of $\hat{B}$.
If $\theta \neq 0$ and we pull back $\hat{B}$ to a two-form on $W^4$, we on the one hand project out any components that annihilate~$u$ and on the other hand won't obtain a generic (self-dual) two-form on $W^4$.
Instead, we consider the contraction $\imath_{v^\perp} \hat{B}$ as a section of \smash{$\bigl(\Delta_{(u,v)}^\perp\bigr)^\ast$}, which contains neither $u^\flat$- nor $w^\flat$-components.
The absence of terms proportional to $u^\flat$ ($={\rm d}s$) ensures that the pullback is injective, i.e., does not project out any components of $\imath_{v^\perp} \hat{B}$.
Using $\hat{A}_s$ as the missing $w^\flat$\nobreakdash-component, we then obtain a~generic one-form~$\phi$ on~$W^4$.
\end{Remark}

To determine the reduction of the Haydys--Witten equations in terms of $(A,\phi)$ on $W^4$, it is sufficient to investigate the differential equations~\eqref{eq:background-dim-red-haydys-witten-equations} in arbitrarily small neighbourhoods of a~point~$x$.
Hence, choose normal coordinates $(s,t,x^i)_{i=1,2,3}$ at $x$ such that $u=\del_s$, $w=\del_t$, and ${v=\cos\theta \del_s + \sin\theta \del_t}$.
Due to $u$-invariance and after setting ${\hat{A}_s=\phi_t}$, we find
\begin{align*}
 \nabla^{\hat A}_{v^\perp} = -\sin\theta [\phi_t, \,\cdot\,] + \cos\theta \nabla^A_t,\qquad
 \bigl(F_{\hat A}\bigr)_{s\mu} = - \nabla^{A}_\mu \phi_t.
\end{align*}
The second of the Haydys--Witten equations~\eqref{eq:background-dim-red-haydys-witten-equations} thus becomes
\begin{align*}
	0 ={}& \imath_v F_{\hat A} - \delta_{\hat A}^+ \hat{B} \\
	={}& ( \cos\theta \imath_{\del_s} + \sin\theta \imath_{\del_t}) F_{\hat A} + \Bigl(\nabla^{\hat A}_{v^\perp} \imath_{v^\perp} + \sum \nabla^{\hat A}_i \imath_{\del_i} \Bigr) \hat{B} \\
	={}& \Biggl(- \nabla^{A}_t \phi_t - \sum_{i=1}^3 \nabla^{A}_i \phi_i \Biggr) \eta^\perp \\
	 &{}{+}\, \sum_{(ijk)} \bigl( \sin\theta ( F_{ti} - [\phi_t,\phi_i]) + \cos\theta\bigl( \nabla^{A}_t\phi_i - \nabla^{A}_i \phi_t \bigr) - \bigl( \nabla^{A}_j \phi_k + \nabla^{A}_k \phi_j \bigr) \bigr) {\rm d}x^i,
\end{align*}
where the sum in the last line is over cyclic permutations of $(123)$.
The $\eta^\perp$-component of this equation states ${\rm d}_{A}^{\hodge_4} \phi = 0$, which is the $\theta$-independent part of the Kapustin--Witten equations~\eqref{eq:background-Kapustin--Witten-equations}.
The ${\rm d}x^i$-components imply vanishing of the $ti$-components of the Kapustin--Witten equations as given in equation~\eqref{eq:background-Kapustin--Witten-equations-combined}.

For the evaluation of the first of the Haydys--Witten equations~\eqref{eq:background-dim-red-haydys-witten-equations}, we expand
\begin{align*}
	{F_{\hat A}^+ = \sum_{(ijk)} (-\sin\theta F_{si} + \cos\theta F_{ti} + F_{jk}) \bigl(\eta^\perp \wedge {\rm d}x^i + {\rm d}x^j \wedge {\rm d}x^k\bigr)}
\end{align*}
and similarly for ${\sigma\bigl(\hat B, \hat B\bigr) = \sum_{(ijk)} [\phi_j,\phi_k] \bigl(\eta^\perp\wedge {\rm d}x^i + {\rm d}x^j\wedge {\rm d}x^k\bigr)}$.
The equations then become
\begin{align*}
	0 &= F_{\hat A}^+ - \sigma\bigl(\hat B, \hat B\bigr) - \nabla^{\hat A}_v \hat{B} \\
	&= \sum_{(ijk)}\bigl(
		\sin\theta \nabla^A_i \phi_t + \cos\theta F_{ti} + F_{jk} - [\phi_j,\phi_k]- \cos\theta [\phi_t, \phi_i] - \sin\theta \nabla^A_t \phi_i
	\bigr) \\
	&\hphantom{= \sum_{(ijk)}}\,
		\times \bigl(\eta^\perp\wedge {\rm d}x^i + {\rm d}x^j\wedge {\rm d}x^k\bigr).
\end{align*}
This implies that also the $ij$-components of the Kapustin--Witten equations~\eqref{eq:background-Kapustin--Witten-equations-combined} vanish.
To see this, multiply by $\sin\theta$ and use the ${\rm d}x^i$-component of the earlier equation to replace $\sin\theta (F_{ti} - [\phi_t, \phi_i])$.

In summary, the key result of this section is the following statement.
\begin{Proposition} \label{prop:background-dimensional-reduction-KW}
Let $M^5 = \R\times W^4$ equipped with a product metric and a non-vanishing unit vector field $v$.
Write $u$ for the unit vector field along $\mathbb{R}$ and assume $g(u,v) = \cos \theta$ is constant.
Let \smash{$\bigl(\hat A, \hat B\bigr)$} be Haydys--Witten fields on $M^5$, write $A=i^\ast \hat{A}$ for the pullback connection on $W^4$, and depending on the value of $\theta$ define fields on $W^4$ as follows:
\begin{align*}
	& \theta = 0\colon \
	B = i^\ast \hat B , \qquad
	C = \hat{A}_{s},
	\\
	& \theta \neq 0\colon \
	\phi = \hat{A}_{s} w^\flat + \imath_{v^\perp} \hat B.
\end{align*}
Then $u$-invariant Haydys--Witten equations for $\bigl(\hat A,\hat B\bigr)$ are either equivalent to the Vafa--Witten equations for $(A,B,C)$ if $\theta \equiv 0 \pmod{\pi}$, or to the $\theta$-Kapustin--Witten equations for $(A,\phi)$ otherwise,
\begin{align*}
	\HW[v] \bigl( \hat A, \hat B \bigr) \longrightsquigarrow{\R\text{-inv.}}
	\begin{cases}
		\VW(A, B, C), & \theta \equiv 0 \pmod{\pi} , \\
		\KW[\theta](A, \phi), & \text{otherwise}.
	\end{cases}
\end{align*}
\end{Proposition}
Let us stress that dimensional reduction is not continuous at $\theta=0$.
For general four-manifolds, the Vafa--Witten equations and $\theta=0$ version of the Kapustin--Witten equations are not equivalent.
From the perspective of the Haydys--Witten equations, it should be expected that there is a possibly non-trivial relation between solutions of these equations whenever the four-manifold admits a non-vanishing vector field.
This is well known for $W^4 = \R^4$, where the Vafa--Witten and $\theta=0$ Kapustin--Witten equations are equivalent by identifying $\phi = C {\rm d}x^0 + B_{0i} {\rm d}x^i$, but it has not yet been investigated for more general four-manifolds.

The dimensional reductions of the Haydys--Witten equations have previously been carried out for the cases $\theta = 0$ and $\pi/2$ by Witten~\cite{Witten2011}, and independently for $\theta=0$ and in slightly more generality by Haydys~\cite{Haydys2015}.
Concretely, Witten considered the case where the five-manifold is $M^5 = \mathbb{R}_s \times X^3 \times \mathbb{R}_y^+$ and $v=\del_y$.
On~the~one hand, he investigates $\R_s$-invariant solutions, i.e., dimensional reduction with respect to $u=\del_s$.
In this situation, the glancing angle is $\theta = \pi/2$ (or $3\pi/2$; the two cases differ only by a reversal of orientation).
Witten explains that setting $\phi = \hat{A}_s {\rm d}y + B_{si} {\rm d}x^i$, the $u$-invariant Haydys--Witten equations are equivalent to the $\theta=\frac{\pi}{2}$ version of the Kapustin--Witten equations.
This is in accordance with the general definition $\phi = \imath_{u} \hat A \wedge w^\flat + \imath_{v^\perp} \hat B$, because in the current situation $u=\del_s$, $w^\flat = {\rm d}y$, and $v^\perp = - \del_s$.
On~the~other hand, Witten also briefly inspected $\R_y^+$-invariant solutions, which provide non-trivial boundary conditions at $y\to\infty$.
Since in that case $u = \del_y$ coincides with $v$, the glancing angle is $\theta = 0$ and the equations reduce to the Vafa--Witten equations.
This was observed more generally by Haydys for $\R_y$-invariant solutions on $M^5 = \R_y \times W^4$ with $v=\del_y$~\cite[Section~4.1]{Haydys2015}.

Let us remark that Witten explains in great detail that the full Haydys--Witten equations on $\mathbb{R}_s \times X^3 \times \mathbb{R}_y^+$ represent antigradient flow equations (with respect to a conveniently chosen functional) that interpolate between $\theta=\pi/2$ Kapustin--Witten solutions at $s\to\pm\infty$.
Haydys similarly explains that the full Haydys--Witten equations on $\R_s \times W^4$ represent antigradient flow equations that interpolate between Vafa--Witten solutions at $s\to \pm\infty$.

Both of these statements can be viewed as special case of the more general fact that on $\R_s \times W^4$ the equations $\HW[v]\bigl(\hat A, \hat B\bigr)=0$ take the form of flow equations that interpolate between $\R_s$-invariant solutions at $s\to\pm\infty$.
In particular, for $\theta \neq 0, \pi/2$, the Haydys--Witten equations on~${\R_s \times W^4}$ are equivalent to the following equations:
\begin{align*}
	&\nabla^A_s A= - \imath_{\del_s} ( {\rm d}_A \phi - \hodge_4 ( \cos\theta ( F_A - [\phi \wedge \phi] ) + \sin\theta {\rm d}_A \phi ) ) ,\\
	&\nabla^A_s \phi= - \imath_{\del_s} ( F_A - \hodge_4 ( - \sin\theta ( F_A - [\phi \wedge \phi] ) + \cos\theta {\rm d}_A \phi ) ), \\
	&{\rm d}_A^{\hodge_4} \phi= 0.
\end{align*}

\subsection[R\^{}2-invariant solutions]{$\boldsymbol{\mathbb{R}^2}$-invariant solutions}
\label{sec:background-dimensional-reduction-EBE}

Consider a product space $M^5 = \R^2 \times X^3$ equipped with a product metric and denote by ${i\colon X^3 \hookrightarrow \R^2 \times X^3}$ inclusion at the origin of $\R^2$.
Let $s_1$ and $s_2$ be Cartesian coordinates on~$\R^2$ and~$u_1$,~$u_2$ the associated coordinate vector fields.
Assume that $g(u_1, v) = \cos \theta_1$ and $g(u_2,v) = \cos\theta_2$ are constant.
We are free to choose coordinates in such a way that $u_2$ and $v$ are orthogonal, fixing $\theta_2 = \pi/2$, see Figure~\ref{fig:background-incidence-angles-1}.
The dimensional reduction only depends on the remaining parameter $\theta := \theta_1$, which is the glancing angle between $v$ and $\R^2$.

\begin{figure}[!ht]
\centering
\includegraphics[width=\textwidth, height=0.1\textheight, keepaspectratio]{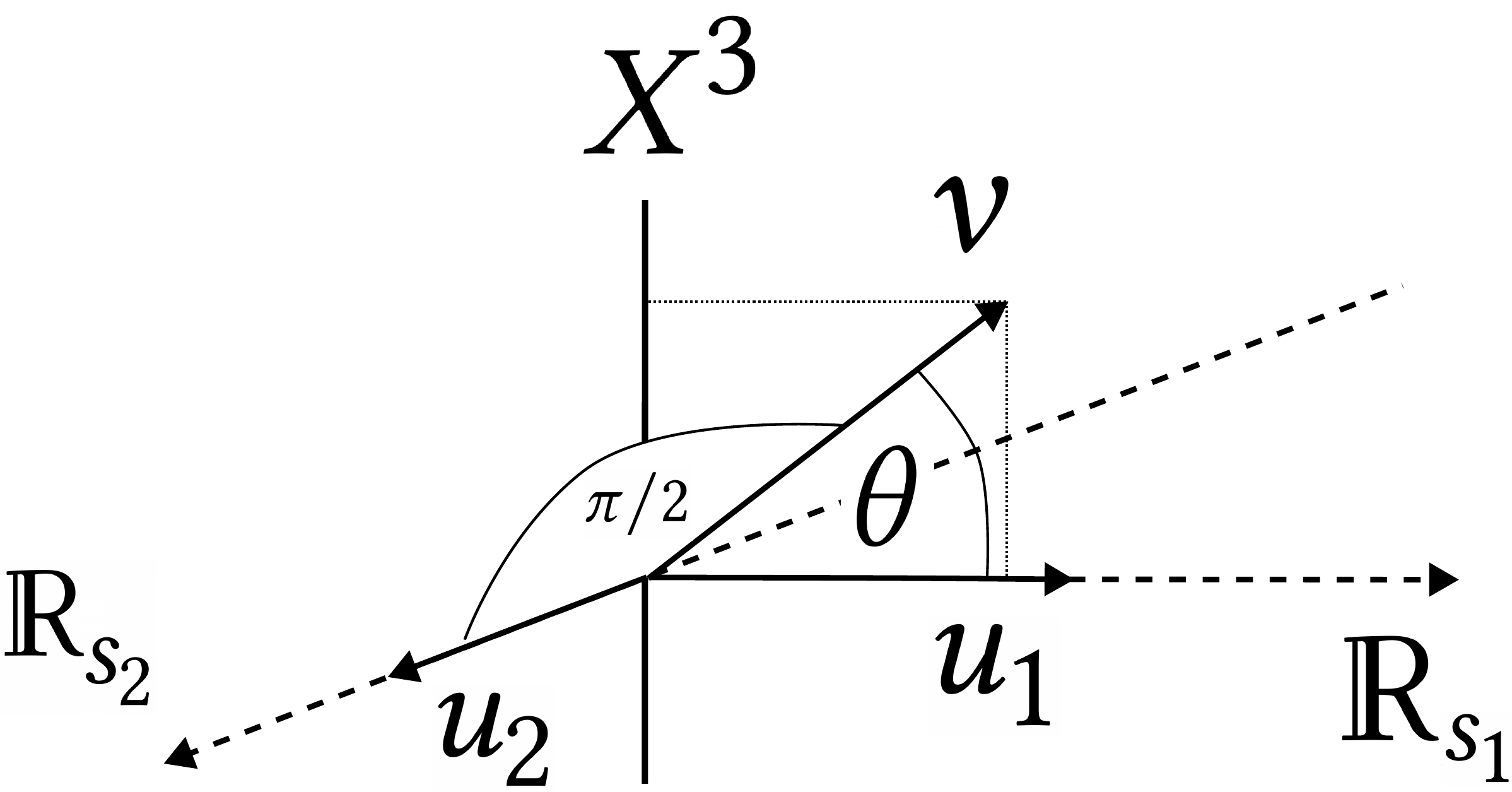}%
\caption{}%
\label{fig:background-incidence-angles-1}%
\end{figure}

Instead of imposing $\R^2$-invariance from scratch, we may proceed by iteration: First utilize the results of Section~\ref{sec:background-dimensional-reduction-KW} to determine the dimensional reduction along one of the directions and afterwards additionally demand invariance in the second direction.
Due to the results of the previous section it is clear that we need to distinguish between the cases ${\theta \equiv 0}$ and~${\theta \not\equiv 0 \pmod{\pi}}$.

In the case where $\theta = 0$ (the case $\theta=\pi$ follows from this by a reversal of orientation), we start by imposing $u_2$-invariance.
By Proposition~\ref{prop:background-dimensional-reduction-KW}, this leads to a pair $(A, \phi)$ on $\R_{s_1} \times X^3$ that satisfies the $\theta_2 = \pi/2$ version of the Kapustin--Witten equations.
The Higgs field is given by \smash{$\phi = \hat{A}_{s_2} {\rm d}s_1 - \imath_{u_2} \hat B$}, while $A$ is the pullback connection.
In the case where $\theta \neq 0$, it is more convenient to first consider $u_1$-invariance, which leads to the $\theta$-Kapustin--Witten equations for a (different) pair $(A,\phi)$.
The Higgs field is given by \smash{$\phi = \hat{A}_{s_1} w^\flat + \imath_{v^\perp} \hat B$}, where $w$ and $v^\perp$ are the sections of $\Delta_{(u_1,v)}$ that were introduced in Section~\ref{sec:background-dimensional-reduction-KW}.

It is helpful to have a closer look at the structure of \smash{$\Omega^2_{v,+}\bigl(\R^2 \times X^3\bigr)$}, to simplify the subsequent dimensional reduction along $u_2$.
Observe that $w$ is a non-vanishing vector field that is orthogonal to both $u_1$ and $u_2$, i.e., it is the pushforward of a non-vanishing vector field on $X^3$.
Let us generalize the notation from the previous section and denote by \smash{$\Delta_{(u_1,u_2,v)}$} the regular distribution spanned by the three vector fields $u_1$, $u_2$, and $v$.
If $\theta \neq 0$, this distribution is a~trivial subbundle of $TM^5$ of rank 3 that admits an orthonormal basis of sections $(u_1,u_2,w)$.
The orthogonal complement \smash{$\Delta_{(u_1,u_2,v)}^\perp$} is a rank 2 subbundle of $TX^3$ orthogonal to $w$.
It~\mbox{follows} that the tangent space $TX^3$ splits into a trivial one-dimensional part, spanned by $w$, and a~two-dimensional part that we denote \smash{$\Delta_{(u_1,u_2,v)}^\perp$}.
This provides a splitting of $\Omega^1\bigl(\R_{s_2} \times X^3\bigr)$ into sections of
\[
C^\infty\bigl(M^5\bigr) {\rm d}s_2 \oplus C^\infty\bigl(M^5\bigr) w^\flat \oplus \bigl(\Delta_{(u_1,u_2,v)}^\perp\bigr)^\ast,
\]
 which in turn induces the existence of a~global section \smash{$e_1 := \thalf (1 + T_\eta) \bigl(\eta^\perp \wedge {\rm d}s_2\bigr)$} of \smash{$\Omega^2_{v,+}$}.
As a consequence, $\hat{B}$ splits globally into
$\hat{B}= \phi_1 e_1 + \varphi$.
In the usual local basis, the two-form $\varphi$ is given by $\varphi = \phi_2 e_2 + \phi_3 e_3$.
With respect to this expression of $\hat B$ and the splitting of $\Omega^1\bigl(\R_{s_2} \times X^3\bigr)$, the Higgs field is defined as
\[\phi = \phi_1 {\rm d}s_2 + \hat{A}_{s_1} w^\flat + \imath_{v^\perp} \varphi.\]
In local coordinates \smash{$\bigl(s_2, x^i\bigr)_{i=1,2,3}$} of $\R_{s_2} \times X^3$, this reduces to
\[
\phi = \phi_1 {\rm d}s_2 + \hat{A}_{s_1} {\rm d}x^1 + \phi_2 {\rm d}x^2 + \phi_3 {\rm d}x^3.
\]

Regardless of the value of $\theta$, we are now ready to perform dimensional reduction along the second direction.
In either case, denote the remaining direction of invariance by $s$, i.e., $s=s_1$ if $\theta = 0$ and $s=s_2$ if $\theta\neq 0$.
We are in the situation of either the $\pi/2$- or $\theta$-version of the Kapustin--Witten equations for $(A,\phi)$ on $\R_s \times X^3$.
Imposing invariance in the second direction thus corresponds to a fairly straightforward dimensional reduction of the Kapustin--Witten equations:
Regardless of the five-dimensional origin of $\phi$ and $A$ they split on $\R_s \times X^3$ as~${\phi = c_1 {\rm d}s + \tilde \phi}$ and ${A = c_2 {\rm d}s + \tilde A}$.
As usual, we view the ${\rm d}s$ components $c_1$ and $c_2$ as $\ad E$-valued functions over $X^3$.
Choose an orientation, say ${\rm d}s \wedge \mu_{X^3}$, and determine the individual terms in the Kapustin--Witten equations~\eqref{eq:background-Kapustin--Witten-equations-combined} in terms of \smash{$\bigl(\tilde A, \tilde \phi, c_1, c_2\bigr)$}, whilst dropping any derivatives in the direction of $s$:
\begin{align*}
	&F_A = \quad {\rm d}_{\tilde A} c_2 \wedge {\rm d}s + F_{\tilde A}, \\
	&\tfrac12 [\phi\wedge\phi] = - \bigl[c_1, \tilde \phi\bigr] \wedge {\rm d}s + \tfrac12 \bigl[\tilde \phi \wedge \tilde \phi\bigr] ,\\
	&{\rm d}_A \phi = - \bigl(\bigl[c_2 , \tilde \phi\bigr] - {\rm d}_{\tilde A} c_1\bigr) \wedge {\rm d}s + {\rm d}_{\tilde A} \tilde\phi ,\\
	&\hodge_4 {\rm d}_A \phi = \hodge_3 {\rm d}_{\tilde A} \tilde\phi \wedge {\rm d}s + \hodge_3 \bigl( \bigl[c_2 , \tilde \phi\bigr] - {\rm d}_{\tilde A} c_1\bigr).
\end{align*}
Plugging these expressions into~\eqref{eq:background-Kapustin--Witten-equations-combined} yields the first two lines of the TEBE~\eqref{eq:background-TEBE}.
The third equation of the TEBE follows from the remaining constraint \smash{$0 = {\rm d}_A^{\hodge_4} \phi = -[c_1,c_2] + {\rm d}_{\tilde A}^{\hodge_3} \tilde\phi$}.

\begin{Proposition} \label{prop:background-dimensional-reduction-EBE}
Let $M^5 = \R^2\times X^3$ equipped with a product metric and a non-vanishing unit vector field $v$.
Write $u_i$, $i=1,2$, for Cartesian vector fields on $\mathbb{R}^2$, chosen such that $g(u_1, v) = \cos \theta$ and $g(u_2,v)=0$, and assume both angles are constant.
Let \smash{$\bigl(\hat A,\hat B\bigr)$} be Haydys--Witten fields on $M^5$, write $\tilde A=i^\ast \hat{A}$ for the pullback connection on $X^3$, and depending on the value of $\theta$ define fields on $X^3$ as follows:
\begin{align*}
	&\theta = 0\colon\
	\tilde \phi = - \imath_{u_2} \hat B ,\qquad
	c_1 = \hat{A}_{s_2} ,\qquad
	c_2 = \hat{A}_{s_1},
	\\
	&\theta \neq 0\colon\
	\tilde \phi = \hat{A}_{s_1} w^\flat + \imath_{v^\perp} \varphi ,\qquad
	c_1 = \phi_1 , \qquad
	c_2 = \hat{A}_{s_2}.
\end{align*}
Then $(u_1,u_2)$-invariant Haydys--Witten equations for \smash{$\bigl(\hat A,\hat B\bigr)$} are equivalent to the EBE if $\theta = 0$ or to the $\theta$-twisted extended Bogomolny equations for \smash{$\bigl(\tilde A, \tilde \phi, c_1, c_2\bigr)$},
\begin{align*}
	\HW[v]\bigl(\hat A, \hat B\bigr) \longrightsquigarrow{\R^2\text{\rm -inv.}}
	\begin{cases}
		\EBE\bigl( \tilde A, \tilde\phi, c_1, c_2\bigr), & \theta \equiv 0 \pmod{\pi} , \\
		\TEBE[\theta]\bigl(\tilde A, \tilde\phi, c_1, c_2\bigr), & \text{otherwise}.
	\end{cases}
\end{align*}
\end{Proposition}

The dimensional reduction of the Haydys--Witten equations to three-manifolds inherits the discontinuity at $\theta=0$ that is already present in the reduction to four-manifolds.
In particular, in the limit $\theta \to 0$ dimensional reduction does not lead to the $\theta=0$ version of the TEBE, but instead to the ``untwisted'' $\pi/2$ version.
As before, this behaviour is encoded in the rank of the regular distribution $\Delta_{(u_1,u_2,v)}$ spanned by $u_1$, $u_2$ and $v$.
If ${\theta \equiv 0 \pmod{\pi}}$, i.e., if $v$ is orthogonal to $X^3$, $\Delta_{(u_1,u_2,v)}$ is of rank two and dimensional reduction leads to the (untwisted) EBE.
If $\theta \not\equiv 0 \pmod{\pi}$, the distribution \smash{$\Delta_{(u_1,u_2,v)}$} has rank three, there exists a non-vanishing vector field $w$ on $X^3$, and dimensional reduction produces the $\theta$-TEBE.

Another special situation arises when $v$ is parallel to $X^3$, since then dimensional reduction leads to the $\pi/2$-TEBE, which we recall are just the (untwisted) EBE.
However, the existence of the vector field $w$ provides additional structure that allows us to continuously deform the $\pi/2$-TEBE to generic $\theta$-TEBE by rotating $v \mapsto v=\cos \theta u_1 + \sin\theta w$.
Such a continuous deformation does not exist if the EBE arise from a dimensional reduction for which $\Delta_{(u_1,u_2,v)}$ is of rank 2.
In that case any small deformation of $v\mapsto v+\epsilon w$, lifting $v$ off the plane spanned by $u_1$ and $u_2$, leads to a discontinuous jump from the $\pi/2$-TEBE to $\theta_\epsilon$-TEBE, where $\theta_\epsilon$ is the corresponding (small) glancing angle.

The idea of deforming (or twisting) the EBE away from $\theta=\pi/2$ is used to great effect by Gaiotto and Witten in~\cite{Gaiotto2012a}.
In their setup, they consider the dimensional reduction of the Haydys--Witten equations from
\begin{align*}
	M^5 = \R_s \times \R_t \times \Sigma \times \R^+_y \to \Sigma \times \R^+_y = X^3.
\end{align*}
The three vector fields of interest are $u_1=\del_s$, $u_2=\del_t$ and $w=\del_y$.
When $v=\del_y$, dimensional reduction results in the EBE.
However, we are in the situation where $\Delta_{(u_1,u_2,v)}$ is of rank three and we can deform the EBE equations away from $\theta=\pi/2$, e.g., by considering ${v = \cos\theta u_1 + \sin\theta w}$.

\subsection[R\^{}4-invariant solutions]{$\boldsymbol{\mathbb{R}^4}$-invariant solutions}

Consider now the case $M^5 = \R^4 \times I$, where $I$ is a real interval, and let $\del_y$ be the coordinate vector field along $I$.
Assume Haydys' preferred vector field $v$ is such that the incidence angle between~$v$ and $\R^4$ -- determined by $g(v,\del_y) = \cos\beta$ -- is constant.
On $\R^4$ fix Cartesian coordinates $(s,x^i)$, $i=1,2,3$, where $\del_s$ is the vector field that satisfies $g(v,\del_s) = \sin \beta$, while $g(v, \del_i) = 0$.
In these coordinates $v= \sin\beta \del_s + \cos\beta \del_y$, see Figure~\ref{fig:background-incidence-angles-2}.
Notice that in the current situation $\beta$ is the incidence angle between $v$ and the hyperplane of invariant directions.
This is in contrast to the preceding discussions, where it was more convenient to use the glancing angle $\theta = \pi/2 - \beta$.

\begin{figure}[!ht]
\centering
\includegraphics[width=\textwidth, height=0.1\textheight, keepaspectratio]{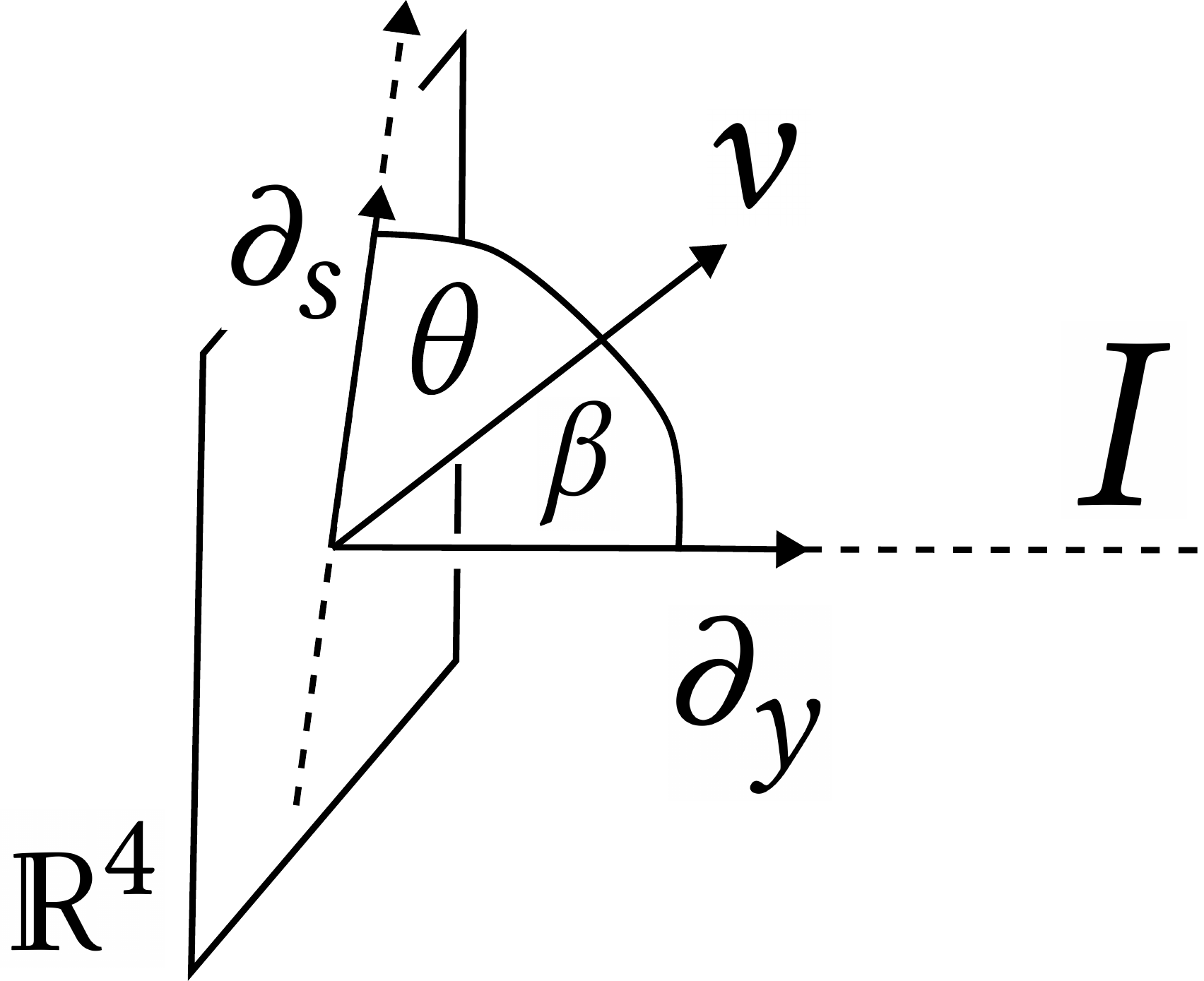}%
\caption{}\label{fig:background-incidence-angles-2}%
\end{figure}

Write \smash{$\hat A = A_s {\rm d}s + \sum_i A_i {\rm d}x^i + A_{y} {\rm d}y $} and \smash{$\hat B= \sum_i \phi_i e_i$} with \smash{$e_i = \thalf(1+T_\eta)\bigl(\eta^\perp \wedge {\rm d}x^i\bigr)$}.
Under~the assumption that \smash{$\hat A$} and \smash{$\hat B$} are independent of $\R^4$, the components $A_s$, $A_i$ and $\phi_i$, $i=1,2,3$, become a collection of seven $\mathfrak{g}$-valued functions on $I$, while $A = A_y {\rm d}y$ provides a~connection over $I$.
As a result dimensional reduction of the Haydys--Witten equations is relatively straight-forward.
The $i$-th component \big(with respect to the basis $\{e_1,e_2,e_3\}$ of $\Omega^2_{v,+}$\big) of the Haydys--Witten equations~\eqref{eq:background-Haydys--Witten-equations} is given by
\begin{align*}
	0&= \bigl( F_{\hat A}^+ - \sigma\bigl(\hat B,\hat B\bigr) - \nabla^{\hat A}_v \phi \bigr)_i \\
	&= \bigl( \cos\beta F_{si} - \sin\beta F_{yi} + \tfrac12 \epsilon_{ijk} F_{jk} \bigr) - \tfrac12 \epsilon_{ijk} [\phi_j,\phi_k] - \cos\beta \nabla^A_y \phi_i - \sin\beta \nabla^A_s \phi_i\\
	&= - \sin\beta \bigl( \nabla^A_y A_i + [A_s, \phi_i] \bigr) - \cos\beta \bigl( \nabla^A_y \phi_i - [A_s, A_i] \bigr) - \frac12 \epsilon_{ijk} ( [\phi_j,\phi_k] - [A_j, A_k] ).
\end{align*}
Meanwhile, the second of the Haydys--Witten equations becomes
\begin{align*}
	0={}& \imath_v F_{\hat A} - \delta_{\hat A}^+ {\hat B}
	= (\cos\beta \imath_{\del_y} + \sin\beta \imath_{\del_s} ) F_{\hat A}
	+ \Big(\nabla^{\hat A}_{s} \imath_{s} + \nabla^{\hat A}_{y} \imath_{y} + \sum \nabla^{\hat A}_i \imath_i \Big) B \\
	={}& \bigl( \nabla^A_y A_s + [A_i, \phi_i] \bigr) \eta^\perp
		\nonumber \\
	&{+} \bigl( \cos\beta \bigl( \nabla^A_y A_i + [A_s,\phi_i] \bigr)- \sin\beta \bigl( \nabla^A_{y} \phi_i - [A_s, A_i] \bigr) - \epsilon_{ijk} [A_j , \phi_k] \bigr) {\rm d}x^i.
\end{align*}
The component proportional to $\eta^\perp$ is exactly the third of the twisted octonionic Nahm equations~\eqref{eq:background-twisted-Nahm-equations-expanded}, while the remaining equations are just a linear combination of the first two lines of these equations.
To see this, multiply the $i$-th components of the two equations with $\sin\beta$ and $\cos\beta$, respectively, and add them up (and vice versa with subsequent subtraction).
More explicitly, the reduced Haydys--Witten equations are thus rearranged to read
\begin{align*}
	&\nabla^A_y \phi_i - [A_s, A_i] + \epsilon_{ijk} \bigl( \cos\beta \tfrac{1}{2} ([\phi_j,\phi_k] - [A_j,A_k]) + \sin\beta [\phi_j, A_k] \bigr)= 0, \\
	&\nabla^A_y A_i + [A_s, \phi_i] - \epsilon_{ijk} \bigl( - \sin\beta \tfrac{1}{2} ([\phi_j,\phi_k] - [A_j,A_k]) + \cos\beta [\phi_j, A_k] \bigr)= 0, \\
	&\nabla^A_y A_s + [A_i, \phi_i]= 0.
\end{align*}
These are exactly the $\beta$-twisted octonionic Nahm equations for $\vec{X} = (\phi_1, \phi_2, \phi_3, -A_s, A_1,A_2,A_3 )$.
\begin{Proposition} \label{prop:background-dimensional-reduction-Nahm}
Let $M^5 = \R^4\times I$, where $I$ is a connected one-dimensional manifold, equipped with a product metric and a preferred non-vanishing unit vector field $v$.
Write $u$ for the unit vector field on $I$.
Assume $g(u,v) = \cos\beta$ is constant and \smash{$\bigl(\hat A, \hat B\bigr)$} are invariant under translations in $\mathbb{R}^4$ and arranged into the pair $(A,\vec{X})$ as specified above.
Then the Haydys--Witten equations reduce to the $\beta$-twisted octonionic Nahm equations
\begin{align*}
	\HW[v]\bigl(\hat A, \hat B\bigr) \longrightsquigarrow{\R^4-\text{inv.}} \Nahm[\mathbb{O},\beta](A, \vec{X}).
\end{align*}
\end{Proposition}

\section{The Nahm pole boundary condition}
\label{sec:background-Nahm-pole-boundary-condition}

The Nahm pole boundary conditions with knot singularities play a fundamental role in the relation between Haydys--Witten theory and Khovanov homology.
They prescribe an asymptotic equivalence of the fields $(A,B)$ with a certain set of singular model solutions near the boundary.
Which model solutions to use depends on whether one is in the vicinity of a knot or not.

Recall from Section~\ref{sec:background-BPS-Kh-Homology} that in the five-dimensional setting a knot is supported on a two-dimensional surface $\Sigma_K$ inside the four-dimensional boundary of $M^5$; one direction is parallel to the longitude of the original, one-dimensional knot $K$, while the other direction arises from extending $K$ along the additional flow direction of Floer theory.
Typically $\Sigma_K$ arises in this way from either a compact knot or a collection of infinitely extended strands embedded in~$X^3$, so for all intents and purposes $\Sigma_K$ is either an embedding of~$\R\times S^1$ or a disjoint union of~$\R^2$'s.
Nevertheless, Nahm pole boundary conditions are defined for general embedded surfaces, including link cobordisms and knotted surfaces and this is how they are presented below.

The singular structure of the model solutions comes in two flavours.
First, denoting by~$y$ a~boundary-defining function, the fields diverge as $y^{-1}$.
Second, near the position of a knot they display a monopole-like behaviour associated to a singularity inside $\ad E$.

A careful analysis of the Nahm pole boundary conditions with knot singularities is described in great detail in~\cite{Mazzeo2014, Mazzeo2017}.
These articles make use of Melrose's $b$-calculus~\cite{Melrose1990}; or rather a variant of it that was introduced by Mazzeo~\cite{Mazzeo1991} and further developed by Mazzeo--Vertman~\cite{Mazzeo2013}.
In this context, it is natural to consider the geometric blowup of $M^5$ along $\Sigma_K$.
This is the manifold with corners, denoted by $\bigl[M^5; \Sigma_K\bigr]$, whose underlying set of points is the disjoint union of $M^5$ and the inward-pointing unit normal bundle of $\Sigma_K$.
Many analytic properties of differential operators and their solutions become more apparent when viewed as conormal distributions on the blowup.
While we will rely on this perspective in Section~\ref{sec:background-Nahm-pole-boundary-condition-elliptic-theory}, it will only feature implicitly in the definitions below.
As long as an explicit distinction between the original manifold and its blowup is irrelevant, we will simply denote the blowup by $M^5$, the original boundary component after removing $\Sigma_K$ by $\del_0 M^5$, and the newly introduced boundary at $\Sigma_K$ by $\del_K M^5$.

The definition of Nahm pole boundary conditions features one aspect of the blowup $\bigl[M^5; \Sigma_K\bigr]$.
Namely, the boundary conditions for the two types of boundary $\del_0 M$ and $\del_K M$ have individual descriptions.
Near $\del_0 M$ the Haydys--Witten pair $(A,B)$ is locally modeled on maximally symmetric, $\R^4$-invariant Nahm-poles on $\R^4\times\R^+_y$, while near $\del_K M$ the model solution is that of an~EBE-monopole on $\R^2 \times \R^+_y$.

Below, we first provide descriptions of these two distinct model solutions, followed by a~definition of the Nahm pole boundary conditions with knot singularities on general manifolds.
We~conclude with an investigation of the analytic properties of the Haydys--Witten equations with $\beta$-twisted Nahm pole boundary conditions, which has not previously appeared in the literature, but is readily available by combining various known results with the geometric interpretation of the twisting angle.

\subsection{Model solutions without knot singularity}
Consider Euclidean half-space $\R^4\times \R^+_y$ and denote Cartesian coordinates by $(s,x^i,y)_{i=1,2,3}$.
Assume that $v= \sin\beta \del_s + \cos\beta \del_y$, where the incidence angle $\beta$ between $v$ and the boundary is constant.
This geometry is invariant under translations parallel to the boundary, and accordingly, we demand that the model solutions are independent of the position in the boundary.
Proposition~\ref{prop:background-dimensional-reduction-Nahm} states that they must then be solutions of the $\beta$-deformed octonionic Nahm equations on $\R^+$.

Let $\rho\colon\mathfrak{su}(2) \to \mathfrak{g}$ be a Lie algebra homomorphism and denote by $(\mathfrak{t}_i)_{i=1,2,3}$ the image of the standard basis of $\mathfrak{su}(2)$ under $\rho$.
Furthermore, let us fix the anti-cyclic permutation $\tau = (132)$.
As mentioned in Section~\ref{sec:background-Nahm-equations}, one easily checks that the following is a solution of the $\beta$-twisted octonionic Nahm equations~\eqref{eq:background-twisted-Nahm-equations-expanded}:
\begin{align} \label{eq:background-Nahm-Pole-model-solutions}
	A_i = \frac{\sin\beta \mathfrak{t}_{\, \tau(i)}}{y}, \qquad
	\phi_i = \frac{\cos\beta \mathfrak{t}_i}{y}, \qquad
	A_s = A_y = 0.
\end{align}
Since the fields exhibit a pole at $y=0$, these are called ($\beta$-twisted or tilted) Nahm pole solutions.
We say the Nahm pole is \emph{regular}, if $\rho$ is a principal embedding in the sense of Kostant, i.e., if~the~commutant of $\mathfrak{su}(2)$ in $\mathfrak{g}$ is a Cartan subalgebra.
These solutions are the local model for the Nahm pole boundary condition near $\del_0 M$.

If $v=\del_y$ is orthogonal to the boundary, i.e., $\beta=0$, the gauge field $A$ vanishes and the model solution coincides with the standard, untwisted Nahm pole solution described in~\cite{Mazzeo2014, Witten2011}.
The twisted Nahm pole model solutions have previously appeared in the context of supersymmetric boundary conditions in~\cite[Section~4]{Gaiotto2009a} and their role in calculating the Jones polynomial via gauge theory was described in~\cite{Gaiotto2012a}.

\subsection{Model solutions with knot singularity}
Consider, again, Euclidean half-space $\R^4 \times \R^+_y$, but now assume that we additionally include a~'t~Hooft operator supported on a single, infinitely extended ``strand''.
More precisely, in five-dimensions this corresponds to the inclusion of a distinguished two-dimensional plane $\Sigma_K = \R^2$ in the boundary of $\R^4 \times \R^+$.
Denote Cartesian coordinates $\bigl(s,t,x^2,x^3,y\bigr)$, where $\Sigma_K$ extends along the $(s,t)$-plane.
For simplicity, assume that $v=\cos\theta \del_s + \sin\theta \del_y$.
To make contact with notation in Section~\ref{sec:background-dimensional-reduction-EBE}: the orthonormal coordinate vector fields parallel to $\Sigma_K$ coincide with $u_1= \del_{s}$ and $u_2 = \del_{t}$, while the unit normal vector $w=\del_y$ plays the role of a distinguished global vector field on the remaining three-manifold \smash{$X^3=\R_{x^2,x^3}^2\times \R^+_y$}.
For reasons that will become clear momentarily, we only consider $\theta \not\equiv 0 \pmod{\pi}$.

As a first step we demand that the model solutions are invariant with respect to translations along $\Sigma_K$.
Due to Proposition~\ref{prop:background-dimensional-reduction-EBE}, this means that the relevant model is a solution of the $\theta$-TEBE~\eqref{eq:background-TEBE} for three-dimensional fields \smash{$\bigl(\tilde A , \tilde\phi, c_1, c_2\bigr)$} on \smash{$\R^2_{x^2,x^3} \times \R^+_y$}.
Since $\theta\neq 0$, we are in the situation where the Haydys--Witten fields are expressed as $A = A_{s} {\rm d}s + A_{t} {\rm d}t + \tilde A$ and $B = \phi_1 e_1 + \varphi$, where the two-form $\varphi=\phi_2 e_2 + \phi_3 e_3$ is such that $\imath_{v^\perp} \varphi = \phi_2 {\rm d}x^2 + \phi_3 {\rm d}x^3$.
Disentangling the general definitions of Proposition~\ref{prop:background-dimensional-reduction-EBE} in this way specifies the three-dimensional fields in terms of the components of $(A,B)$ as follows:
\begin{align*}
	\tilde A = A_2 {\rm d}x^2 + A_3 {\rm d}x^3 + A_y {\rm d}y , \qquad
	\tilde \phi = \phi_2 {\rm d}x^2 + \phi_3 {\rm d}x^3 + A_{s} {\rm d}y , \qquad
	c_1 = \phi_1 , \qquad
	c_2 = A_{t}.
\end{align*}

For the case $\theta = \pi/2$ and $G={\rm SU}(2)$, the relevant model solutions of the $\pi/2$-TEBE (which are simply the untwisted EBE) were described by Witten~\cite{Witten2011}.
Introduce (hemi-)spherical coordinates $(R,\psi,\vartheta)$ on \smash{$\R^2_{x^2,x^3} \times \R^+_y \simeq [0,\infty)_R \times H^2_{\vartheta, \psi}$}, where $R\in[0,\infty)$, $\psi \in [0,\pi/2]$ and $\vartheta \in [0,2\pi]$ are given by
\begin{align*}
	R^2 = x_2^2 + x_3^2 + y^2, \qquad
	\cos \psi= \frac{y}{R}, \qquad
	\cos \vartheta = \frac{x_2}{\sqrt{x_2^2 + x_3^2}}.
\end{align*}
Let $(\mathfrak{t}_i)_{i=1,2,3}$ denote a standard basis of $\mathfrak{su}(2)$ and view $\mathfrak{t}_1$ as the generator of a fixed Cartan subalgebra.
Introduce, by abuse of notation, the $\mathfrak{sl}(2,\mathbb{C})$-valued function $\varphi = \phi_2 - {\rm i} \phi_3$ that conveniently combines the components $\phi_2$ and $\phi_3$ of the two-form $\varphi$ of the same name.
Similarly, denote by $E = \mathfrak{t}_2 - {\rm i}\mathfrak{t}_3$, $H=\mathfrak{t}_1$, and $F = \mathfrak{t}_2 + {\rm i}\mathfrak{t}_3$ the elements of an $\mathfrak{sl}(2,\mathbb{C})$-triple $(E,H,F)$.
Finally, express the components of the three-dimensional connection in terms of spherical coordinates $\tilde A = A_R {\rm d}R + A_\psi {\rm d}\psi + A_{\vartheta} {\rm d}\vartheta$.
The knot singularity solutions of the EBE with charge $\lambda \in \mathbb{Z}$ in terms of the components of $(A,B)$ are given by the following expressions:
\begin{align*}
	&A_\vartheta= - (\lambda+1) \cos^2 \psi \frac{(1+\cos \psi)^{\lambda} - (1-\cos\psi)^{\lambda}}{(1+\cos\psi)^{\lambda+1}-(1-\cos\psi)^{\lambda+1}} H ,\\
	&\phi_1= -\frac{\lambda+1}{R} \frac{(1+\cos\psi)^{\lambda+1} + (1-\cos\psi)^{\lambda+1}}{(1+\cos\psi)^{\lambda+1} - (1-\cos\psi)^{\lambda+1}} H, \\
	&\varphi= \frac{(\lambda+1)}{R} \frac{\sin^\lambda \psi\exp({\rm i}\lambda\vartheta) }{(1+\cos\psi)^{\lambda+1} - (1-\cos\psi)^{\lambda+1}} E, \\
	&A_s= A_t = A_R = A_\psi = 0.
\end{align*}
These solutions exhibit a singular behaviour in several distinct ways.
First, the components of $B$ diverge with rate $R^{-1}$ as $R\to 0$.
Second, whenever $R\neq 0$ the solution is asymptotically equivalent to the (untwisted) Nahm pole solution as we approach the original boundary component $\psi \to \pi/2$.
This compatibility will be relevant in the definition of the boundary conditions on general manifolds.
Third, and most importantly, the solutions exhibit a monopole-like singularity.
This is characterized by a nontrivial monodromy of the connection as one moves around the origin in the $(x_2,x_3)$-plane.
The monodromy is supported by a non-trivial behaviour of $\varphi$, which picks up an extra factor\footnote{$\varphi$ remains single-valued, since for $G={\rm SU}(N)$ only integer values appear, while for $G={\rm SO}(N)$ half-integer values may appear, but then the Lie algebra is really $\mathfrak{psl}(2,\mathbb{C})$, where multiplication by $-1$ is modded out.} of ${\rm e}^{2\pi {\rm i} \lambda}$ when $\vartheta$ increases by $2\pi$, and vanishes on the half-line $\psi=0$ that sits over the origin in the $\bigl(x^2,x^3\bigr)$-plane.
An insightful way to view the behaviour of $\varphi$ is to observe that it takes values in the nilpotent cone $\mathcal{N} \subset \mathfrak{g}_{\mathbb{C}} = \mathfrak{sl}(2, \mathbb{C})$ and that it approaches the cone singularity of $\mathcal{N}$ for $\psi \to 0$.

Analogous solutions for the more general case $\theta = \pi/2$ and $G={\rm SU}(N)$ have been constructed by Mikhaylov in~\cite{Mikhaylov2012}.
In this case, the solution is labeled by an element of the co-character lattice $\lambda \in \Gamma_\text{ch}^\vee = \Hom(\C^{\times}, G_{\mathbb{C}})$, or equivalently, by a representation of the Langlands dual group $G_{\mathbb{C}}^\vee$.
From the physics perspective, $\lambda$ corresponds to a choice of magnetic charge.
In~this generalization the divergence of order $R^{-1}$ remains unchanged.
However, the singular behaviour within the nilpotent cone $\mathcal{N} \subset \mathfrak{sl}(N, \mathbb{C})$ has a richer structure, since the nilpotent cone has various singularities, see, for example,~\cite{Collingwood1993}.
The co-character $\lambda$ determines which of these singularities $\varphi$ approaches as $\psi \to 0$.

With regard to a generalization by twisting, Gaiotto and Witten describe in~\cite{Gaiotto2012a} that it~is sometimes beneficial to consider the $\theta$-TEBE for $\theta \neq \pi/2$.
They predicted that there are analogous knot singularity solutions also in these cases.
This has recently been confirmed for $G={\rm SU}(2)$ by Dimakis~\cite{Dimakis2022}, who utilized a continuity argument to prove the existence of knot singularity solutions for the $\theta$-TEBE for any $\theta \in (0,\pi)$.
The deformation of $\theta$ away from $\pi/2$ to $\pi/2 - \beta$ has the effect that the Nahm pole divergence of order $R^{-1}$ that appears in $B$ is rotated into $A$, in very much the same way as is the case for the twisted Nahm pole solution in~\eqref{eq:background-Nahm-Pole-model-solutions}.
As we have seen in Proposition~\ref{prop:background-dimensional-reduction-EBE}, the dimensional reduction of the Haydys--Witten equations for $\theta=0$ is generally not continuously connected to the $\theta\neq 0$ reductions, such that continuity methods break down at $\theta=0$.
It is not currently known if there are knot singularity models for~${\theta\equiv 0\pmod{\pi}}$.

Both of these generalizations are given by less explicit descriptions than the model solutions above and the exact formulas, where available, do not provide additional insights.
For our purposes, it will suffice to assume that model solutions exist for any $\theta \in (0,\pi)$ and $G={\rm SU}(N)$, and are labeled by a magnetic charge $\lambda \in \Gamma_\mathrm{char}^\vee$.

\subsection{From model solutions to boundary conditions}

Let $M^5$ be a Riemannian manifold with a single boundary component, together with a preferred non-vanishing unit vector field $v$ that approaches the boundary at a constant angle.
We take this to mean that there is a tubular neighbourhood of the boundary $U=\del M^5 \times [0,\epsilon)_y$, on which the incidence angle of $v$ with the boundary is constant: $\cos\beta = g(v, \del_y )$.
Furthermore, let~${\Sigma_K \subseteq \del M^5}$ be an embedded surface and assume that also the glancing angle of $v$ with $\Sigma_K$, given by ${\cos \theta = \min_{u \in T\Sigma_K} g(u,v)/\norm{u}}$, is constant on $\Sigma_K$.

To simplify the discussion, we impose as an additional assumption that the incidence angle of $v$ with $\del M$ and its glancing angle with $\Sigma_K$ are related by $\theta = \pi/2 - \beta$.
This means that the glancing angle between $v$ and $\Sigma_K$ is identical to the glancing angle between $v$ and the full boundary $\del M^5$.
In this case, there is a neighbourhood of $\del_K M$ where $v = \cos\theta u + \sin\theta \del_y$ for some non-vanishing unit vector field $u\in T\Sigma_K$, i.e., the projection of $v$ to the boundary is parallel to $\Sigma_K$.
This is the situation one encounters in the context of Khovanov homology.

As explained in the introductory paragraphs of the current section, we promote $M^5$ to the geometric blowup along $\Sigma_K$, such that there are two boundary components, $\del_0 M$ and $\del_K M$.
In the preceding sections, we have described the model solutions that shall describe the local behaviour of the fields at each of the two boundaries.
A complete boundary condition requires a~global specification of the fields on ${\del M = \del_0 M \sqcup \del_K M}$.
Since the model solutions diverge at the boundary, this corresponds to specifying the leading order behaviour on tubular neighbourhoods of $\del_0 M$ and $\del_K M$, respectively.
The descriptions on these neighbourhoods must of course be compatible on intersections.

We start by fixing the global boundary data on $\del_0 M$.
Consider a tubular neighbourhood $U = \del_0 M \times [0,\epsilon)_y$ on which $g(\del_{y} , v) = \cos \beta$ is constant.
Observe that the components of~$A$ and~$B$ in the Nahm pole model~\eqref{eq:background-Nahm-Pole-model-solutions} are given by the same expression, up to a relative rotation with respect to the incidence angle $\beta$.
For this reason, it is helpful to recall that there is a~relation between one-forms and self-dual two-forms on $U$.
Hence, note that whenever the incidence angle~$\beta \neq 0$,~$v$ induces a non-vanishing vector field $u$ parallel to $\del_0 M$.
This leads to a splitting of the tangent space \smash{$TU \simeq \Delta_{(u,v)} \oplus \Delta_{(u,v)}^\perp$}, where the orthogonal complement~\smash{$\Delta_{(u,v)}^\perp$} is a~vector bundle of rank three.
As in Section~\ref{sec:background-dimensional-reduction}, contraction with $v^\perp$ then provides an isomorphism between~$\Omega^2_{v,+}(U)$ and sections of \smash{$\bigl(\Delta_{(u,v)}^\perp\bigr)^\ast$}.
In coordinates \smash{$\bigl(s, x^i, y\bigr)_{i=1,2,3}$} where ${v=\sin\beta \del_s + \cos\beta \del_y }$, the isomorphism identifies $\sum \phi_i e_i \mapsto \sum \phi_i {\rm d}x^i$, where $e_i$ denotes the usual basis of $\Omega^2_{v,+}(U)$.
It follows that any $\ad E$-valued self-dual two-form over $U$ is equivalent to an~$\ad E$-valued one-form on a subbundle of $TU$,
\begin{align*}
	\Omega^2_{v,+}(U,\ad E) \simeq \Hom\bigl( \Delta_{(u,v)}^\perp , \ad E \bigr).
\end{align*}
To keep notation at a minimum, this identification will be used implicitly in the formulas below.

Regardless of the value of $\beta$, the three-dimensional cross product on $\Omega^2_{v,+}(U)$ (cf.\ Section~\ref{sec:background-Haydys--Witten-equations}) provides each fiber with the Lie algebra structure of $\mathfrak{su}(2)$.
Thus, at every point in the tubular neighbourhood $U$, any $\ad E$-valued self-dual two-form $\phi_\rho$ that satisfies $\phi_\rho - \sigma(\phi_\rho, \phi_\rho) = 0$ gives rise to a Lie algebra homomorphism $\rho\colon \mathfrak{su}(2) \to \mathfrak{g}$.
If we denote the image of the standard basis of $\mathfrak{su}(2)$ under $\rho$ by $(\mathfrak{t}_i)_{i=1,2,3}$, then \smash{$\phi_\rho = \sum_{i=1}^3 \mathfrak{t}_i e_i$} in the usual local basis.

Conversely, a smooth family of homomorphisms $\{\rho_p\colon \mathfrak{su}(2) \to \mathfrak{g} \}_{p\in U}$ determines a unique two-form $\phi_\rho \in \Omega^2_{v,+}(U,\ad E)$ that satisfies $\phi_\rho - \sigma(\phi_\rho, \phi_\rho) = 0$.
Moreover, $\{\rho_p\}$ also induces another two-form $\phi_\rho^\tau$, related to $\phi_\rho$ by a change of orientation of $\Omega^2_{v,+}(U)$ from~$(e_1, e_2, e_3)$ to~$(e_1, e_3, e_2)$.
In a local basis \smash{$\phi_\rho^\tau = \sum_i \mathfrak{t}_{\, \tau(i)} e_i$} where $\tau$ is the anti-cyclic permutation $(132)$ from earlier.
Since $\sigma(\cdot,\cdot)$ is defined with respect to the original orientation on $\Omega^2_{v,+}(U, \ad E)$, $\phi_\rho^\tau$~satisfies \smash{${\phi_\rho^\tau + \sigma\bigl(\phi_\rho^\tau , \phi_\rho^\tau \bigr) = 0}$}.

Hence, let $\rho\colon \mathfrak{su}(2) \to \ad E$ be a Lie algebra homomorphism and $\phi_\rho$ the associated two-form.
Consider the Haydys--Witten fields on $U$ that are given by
\begin{align*}
	A^{\rho, \beta} = \frac{\sin\beta \phi_\rho^\tau}{y}, \qquad
	B^{\rho, \beta} = \frac{\cos\beta \phi_\rho}{y}.
\end{align*}
Locally $\bigl(A^{\rho, \beta}, B^{\rho, \beta}\bigr)$ coincide with the Nahm pole solutions, perhaps up to conjugation in~$\mathfrak{g}$.
The main take-away is that the boundary data at $\del_0 M$ is fully determined by a choice of $y$-independent two-form \smash{$\phi_\rho \in \Omega^2_{v,+}(U, \ad E) \simeq \Hom\bigl( \Delta_{(u,v)}^\perp, \ad E \bigr)$} in a tubular neighbourhood $U$ of $\del_0 M$, that satisfies $\phi_\rho - \sigma(\phi_\rho, \phi_\rho)=0$.

Moving on to the boundary component $\del_K M$, denote by $V=\del_K M \times [0,\epsilon)_R$ a tubular neighbourhood.
$V$ is the product of $\Sigma_K$ and the filled hemisphere \smash{$H^2_{\psi,\vartheta} \times [0,\epsilon)_R$}, where the latter admits global coordinates $(\psi, \vartheta, R)$.
We wish to impose that at leading order $R^{-1}$ the behaviour of $(A,B)$ is described completely by the knot singularity model solution.
The only degree of freedom is the choice of $\mathfrak{su}(2)$ generators $\mathfrak{t}_i \in \mathfrak{g}$ at each point in $V$.
However, except for the points at $\psi=0$ or $R=0$, every $p\in V$ is also contained in the tubular neighbourhood~$U$ of~$\del_0 M$, where $\phi_\rho$ already determines a triple $(\mathfrak{t}_i)_{i=1,2,3}$.
Since the hemisphere at $R=0$ corresponds to a single point on the original manifold and the line $\psi=0$ is of codimension two, the two-form $\phi_\rho \in \Omega^2_{v,+}(U,\ad E)$ extends uniquely to all of~$V$.

It remains to note that since $\theta\neq 0$, the two-form decomposes on $V$ into ${\phi_\rho = (\phi_\rho)_1 e_1 + \varphi_\rho}$.
This splitting provides a natural distinction between $H$ and $E$ in the model solutions, by identifying the $(\phi_\rho)_1$ component with the Cartan element $H$.
Using this to replace the generators~$\mathfrak{t}_i$ in the knot singularity solutions by $(\phi_\rho)_i$, i.e., pointwise by the image of the induced map $\rho\colon \mathfrak{su}(2) \to \mathfrak{g}$, determines a unique field configuration \smash{$\bigl(A^{\lambda,\theta} , B^{\lambda,\theta}\bigr)$} of order \smash{$\mathcal{O}\bigl(R^{-1}\bigr)$} on all of $V$.

In the more general situation with $\theta \neq \pi/2 -\beta$, the same discussion goes through with minor modifications when identifying coordinates and field components over $U$ and $V$, respectively.
We can now state the definition of the regular Nahm pole boundary conditions.

\begin{Definition}[regular Nahm pole boundary conditions with knot singularities]
\label{def:background-Nahm-pole-boundary-condition}
Assume~$v$ approaches $\del_0 M$ at a constant incidence angle $\beta$ ($\neq \pi/2$) and has glancing angle $\theta$ ($\neq 0$) with $\Sigma_K$.
Let $\{ \rho_p\colon \mathfrak{su}(2) \to \mathfrak{g} \}$ be a smooth family of principal embeddings on a tubular neighbourhood of the boundary and $\phi_\rho$, $\phi_\rho^\tau$ the associated self-dual two-forms.
The Haydys--Witten pair $(A,B)$ satisfies the regular Nahm pole boundary conditions at $\del M$, with knot singularity of weight~$\lambda\in \Gamma_{\mathrm{char}}^\vee$ along $\Sigma_K$, if for some $\epsilon>0$
\begin{enumerate}\itemsep=0pt
	\item[(1)] near $\del_0 M$: $(A,B) = \bigl( A^{\rho, \beta}, B^{\rho, \beta}\bigr) + \mathcal{O}\bigl(y^{-1+\epsilon}\bigr)$,
	\item[(2)] near $\del_K M$: $(A,B) = \bigl( A^{\lambda,\theta}, B^{\lambda,\theta}\bigr) + \mathcal{O}\bigl(R^{-1+\epsilon}\bigr)$,
\end{enumerate}
and such that the leading orders are compatible at the corner $R=y=0$.
This means that in spherical coordinates, where $y = R \cos\psi$, the expansion is of product type
\[
(A,B) = \bigl(A^{\lambda,\theta}, B^{\lambda,\theta}\bigr) + \mathcal{O}\bigl( R^{-1+\epsilon} \cos\psi^{-1+\epsilon} \bigr).\]
\end{Definition}

\begin{Remark}
There is an analogous definition associated to non-regular embeddings $\rho\colon\! {\mathfrak{su}(2) \!\to\! \mathfrak{g}}$ (see \cite{Mazzeo2017}).
However, throughout this article we only consider the \emph{regular} Nahm pole boundary conditions and omit a discussion of this more general case.
\end{Remark}

\subsection{Elliptic theory of Nahm pole boundary conditions}
\label{sec:background-Nahm-pole-boundary-condition-elliptic-theory}

In this subsection, we summarize some relevant analytic properties of the Haydys--Witten equations.
The fundamental questions include under which conditions $\HW[v]$, acting on appropriate function spaces, is Fredholm and to analyze the regularity of solutions of $\HW[v](A,B) = f$.
These properties are controlled by the fact that the Haydys--Witten and Kapustin--Witten operators are elliptic.

As is common for gauge theoretic equations, the Haydys--Witten and Kapustin--Witten equations on their own are not elliptic ``on the nose''.
But they become elliptic after choosing a~reference connection $A^0$ and imposing additionally that the linearization of the gauge action vanishes.
Since ellipticity depends only on the principal symbol, we are free to add terms in subleading orders of derivatives.
In the context of Nahm pole boundary conditions it is convenient to include, in this way, the leading order term $B^{\mathrm{NP}}$ that captures the Nahm pole behaviour of~$B$.
The gauge fixing equation we use is
\begin{align}
	{\rm d}_{A^0}^{\hodge} \bigl(A - A^0 \bigr) + \sigma\bigl( B^{\mathrm{NP}} , B - B^{\mathrm{NP}} \bigr) = 0.
	\label{eq:background-Nahm-gauge-fixing}
\end{align}
From now on, we always assume that the Haydys--Witten (and Kapustin--Witten) equations include this equation.

On closed manifolds, standard elliptic theory provides answers to many of the relevant analytic questions.
On manifolds with boundary, however, these considerations are complicated by the choice of boundary conditions.
In particular, under the assumption of Nahm pole boundary conditions with knot singularities, the associated differential operators are known to be ``depth-two incomplete iterated edge (iie) operators''.
The study of such operators is part of the larger framework of geometric microlocal analysis and may be viewed as a variant of Melrose's $b$-calculus.
We refer to~\cite[Sections~8 and 9]{Mazzeo2017} for a very clear, if concise, account of the relevant ideas.
Also see~\cite{Mazzeo1991, Mazzeo2013} for a more detailed discussion of much of the relevant background.

The results we discuss here, as well as most of the necessary background, have previously been described in great detail in the context of the $\theta=\pi/2$ version of the Kapustin--Witten equations~\cite{Mazzeo2014, Mazzeo2017}.
In these articles, Mazzeo and Witten proved that the Nahm pole boundary conditions with knot singularities amend the Kapustin--Witten operator to an elliptic system.
At the heart of the analysis lies the determination of formal rates of growth for homogeneous solutions of the linearization of the Kapustin--Witten equations.
As is usually the case, ellipticity is accompanied by a regularity theorem, showing that these formal growth rates provide the building blocks for an asymptotic expansion of solutions of $\HW[v](A,B) = f$ near the boundaries of $[M;\Sigma_K]$.

In this context, regularity is described in terms of polyhomogeneous functions, which are defined by the existence of asymptotic expansions with respect to scale functions of the form $y^\alpha (\log y)^k$.
To make this more precise, let us call $\Delta \subset \C \times \N_0$ an \emph{indicial set} if it is a countable subset that is ``bounded from the left'' in $\C$ and contains only ``finite towers'' in $\N_0$.
In other words, for any $\alpha_0 \in \R$, there are only finitely many elements $(\alpha, k) \in \Delta$ with $\operatorname{Re} \alpha \leq \alpha_0$ and if $(\alpha, k) \in \Delta$ with $k>0$, then so is $(\alpha, k-1)$.
A function $f$ is polyhomogeneous at a submanifold $\{y = 0\}$ if there is an indicial set $\Delta$ such that \smash{$f \sim \sum_{(\alpha, k) \in \Delta} f_{\alpha, k} y^\alpha (\log y)^k$} as $y\to 0$, where the functions $f_{\alpha, k}$ are independent of $y$.

\begin{Theorem}[elliptic regularity, cf.~{\cite[Proposition~5.9]{Mazzeo2014} and~\cite[Theorem~9.6]{Mazzeo2017}}]
\label{thm:background-elliptic-regularity}
The Haydys--Witten and Kapustin--Witten operators, together with regular Nahm pole boundary conditions with knot singularities, are elliptic iie operators.
In particular, assume that $(A,B)$ is a solution of the Haydys--Witten equations satisfying the $\beta$-twisted Nahm pole boundary conditions with $\theta$-twisted knot singularities as described in Definition~$\ref{def:background-Nahm-pole-boundary-condition}$.
Denote by $A^{\mathrm{NP}}$ and $B^{\mathrm{NP}}$ the leading terms of $(A,B)$ at the boundary, choose a reference connection $A^0 = A^{\mathrm{NP}} + \omega$ where $\omega$ is a~connection on the restriction of $E$ to the boundary, and write $A = A^{\mathrm{NP}} + \omega + a$ and $B = B^{\mathrm{NP}} + b$.
Then $a$ and $b$ are polyhomogeneous on $[M;\Sigma_K]$, i.e., there are asymptotic expansions
\begin{alignat*}{4}
	&a \sim \sum_{(\alpha, k) \in \Delta_0} a_{\alpha,k} y^{\alpha} (\log y)^k ,\qquad&&
	b \sim \sum_{(\alpha, k) \in \Delta_0 } b_{\alpha, k} y^\alpha (\log y)^k \qquad &&
	(y\to 0), &\\[1em]
	&a \sim \sum_{(\beta, m) \in \Delta_K} a_{\beta,m} R^{\beta} (\log R)^m ,\qquad&&
	b \sim \sum_{(\beta, m) \in \Delta_K } b_{\beta, m} R^\beta (\log R)^m \qquad &&
	(R\to 0),&
\end{alignat*}
and corresponding product-type expansions near the corner $y=R=0$, that are compatible with the fact that $y = R \cos \psi$.
Moreover, the indicial sets $\Delta_0$ and $\Delta_K$ are bounded from the left by~${\alpha \geq 1}$ and $\beta \geq 0$, respectively.
\end{Theorem}

A complete proof for the case with incidence angle $\beta = 0$ (equivalently $\theta = \pi/2$) was provided by Mazzeo and Witten in~\cite{Mazzeo2014, Mazzeo2017}.
The proof for other values of $\beta$ is completely analogous, so we refrain from reproducing the full details here.
Instead, we will only briefly recapitulate the main line of arguments and then report in Proposition~\ref{prop:background-indicial-set} the calculation of the indicial sets $\Delta_0$ and $\Delta_K$ that is required to establish the theorem for general twisting angle.

Consider the Haydys--Witten operator on the depth-2 stratified space with strata $M^5 \setminus \del M^5$, $\del M^5 \setminus \Sigma_K$, and $\Sigma_K$, as an iie operator of order 1.
This means that in a neighbourhood of each stratum its linearization takes a certain iterative form, that combines a differential operator on the stratum with an iie operator (of smaller depth) on the link.
For example, in coordinates $(s,t,R,\psi,\vartheta)$ in a neighbourhood of a point on the lowest (depth 2) stratum $\Sigma_K$, the linearization must look like
\begin{align*}
	L = ( L_s \del_s + L_t \del_t ) + L_R \del_R + \frac{1}{R} \biggl( L_\psi \del_\psi + L_\vartheta \del_{\vartheta} + \frac{1}{\cos\psi} L_0 \biggr),
\end{align*}
where each $L_i$ is a smooth (or polyhomogeneous) endomorphism of $\mathfrak{g}$.
Note that for $\psi\! \to \! \pi/2$ (i.e., $\cos\psi \to 0$) one approaches points in the middle (depth 1) stratum $\del M^5 \setminus \Sigma_K$.
Thus, the operator that is multiplied with $R^{-1}$ is itself an iie operator of depth 1, exemplifying the iterative nature of the definition.
Linearizing the Haydys--Witten operator around $\bigl(A^{\mathrm{NP}}, B^{\mathrm{NP}}\bigr)$ indeed yields an operator of that form.

The definition of full ellipticity as iie operator then involves three properties.
The first property is invertibility of the ``iie symbol'', which is a suitable analogue of the principal symbol of~$\HW[v]$, while the second and third property are about the iterative invertibility of certain model operators, called ``normal operators'', at points on strata of depth 1 and depth 2, respectively.

To determine the iie symbol, one considers the non-singular operator $R\cos\psi L$ and replaces derivatives by vectors according to a rule that takes into account the depth of the associated stratum.
In the specific case above, one replaces $R\cos\psi \del_s$ and $R\cos\psi \del_t$ by $-{\rm i} \xi_1$ and $-{\rm i}\xi_2$, respectively, while $\cos\psi \del_\psi$ and $\cos\psi \del_\vartheta$ are similarly replaced by $-{\rm i}\xi_3$ and $-{\rm i}\xi_4$, and subleading orders of differentiation (here $L_0$ of order 0) are discarded.
Invertibility of the iie symbol of the Haydys--Witten equations immediately carries over from the same result for the Kapustin--Witten equations.

The two relevant normal operators are determined as follows.
A tubular neighborhood in~$M^5$ of either of the two strata $S=\del M^5 \setminus \Sigma_K$ or $S=\Sigma_K$ is diffeomorphic to a bundle of cones~$C(Z)$, where $Z$ is either a point or the hemisphere $H^2$, respectively.
For any $p\in S$, the normal operator~$N_p(L)$ is defined as the scale- and translation-invariant operator on $T_p S \times C(Z)$, that is induced by freezing the coefficient functions $L_i$ to their values at the point $p$.
In general, the properties of~$N_p(L)$ may depend on $p \in S$ as a parameter, but this is fortunately not the case in our situation.

Invertibility of $N_p(L)$, acting on appropriately defined ``iterated edge Sobolev spaces'' on the blowup $[M;S]$, depends on the formal rates of growth for solutions of $Lu = 0$.
These rates are called the indicial roots of $N_p(L)$ and are determined by solving the condition $N_p(L)(y^\alpha u) = \mathcal{O}(y^\alpha)$ for~$\alpha$, where $y$ denotes a boundary defining function of the (blown-up) stratum under consideration.
The key property we need to show is that, under the assumption of regular Nahm pole boundary conditions, there are no indicial roots in some non-empty interval $(-1 , \overline{\alpha} )$.
Once this is established, $N_p(L)$ is invertible on function spaces associated to the scale function $y^\mu$ for any weight $\mu \in (-1, \overline{\alpha})$.
As a consequence, the Haydys--Witten operator is an elliptic iie operator and the indicial sets $\Delta_0$ and $\Delta_K$ in the polyhomogeneous expansion of $a$ and $b$ are bounded from the left by $\overline{\alpha}$.

As it turns out, the relevant indicial roots $\underline{\alpha}$ and $\overline{\alpha}$ are independent of the twisting angle $\beta$.
Abstractly, the reason for this is that the Haydys--Witten equations with $\beta$-twisted Nahm pole boundary conditions are locally equivalent to Haydys--Witten equations with untwisted Nahm pole boundary conditions, but where instead the corrections $a$ and $b$ are rotated into each other by an angle $-\beta$.
As a consequence, the only difference is that the two originally distinct indicial roots of $a$ and $b$ at $\beta=0$ are combined.
In what follows, this statement is explained in full detail for the indicial equations at the depth 1 stratum $\del M^5 \setminus \Sigma_K$.
Subsequently, we also comment on the indicial roots of the depth 2 stratum $\Sigma_K$, where indicial roots are known only implicitly.

\begin{Proposition}\label{prop:background-indicial-set}
If $G={\rm SU}(N)$ and $(A,\phi) \sim \bigl(A^{\rho,\beta}, \phi^{\rho,\beta}\bigr)$ as $y\to 0$ satisfies the regular $\beta$-twisted Nahm pole boundary condition, then the set of indicial roots is $\{-(j+1), -j, j,\allowbreak {j+1}\}_{j=1,\ldots, N-1}$ and does not depend on the twisting angle $\beta$.
In particular, there are no indicial roots in the interval $(-1, 1)$.
\end{Proposition}
\begin{proof}
Let $p \in \del M^5 \setminus \Sigma_K$, in which case the normal operator $N_p(L)$ acts on pairs of functions $(a,b)$ over \smash{$T_p \bigl(\del M^5 \setminus \Sigma_K\bigr) \times C(\{\text{pt.}\}) \simeq \R^4 \times \R^+_y$}.
Assume that $(a,b)$ are $\R^4$-invariant and, moreover, that they are of the form $a = y^\alpha a_0$ and $b=y^\alpha b_0$ for some $y$-independent, $\mathfrak{g}$-valued differential forms $a_0$, $b_0$.
Since the Nahm pole terms are proportional to $y^{-1}$ (and their contributions at order $y^{-2}$ vanish by construction), the leading order of $N_p(L)(a, b)$ is $y^{\alpha-1}$.
Terms at this order arise from the action of $\del_y$, as well as from commutators with $A^{\mathrm{NP}}$ or $B^{\mathrm{NP}}$.
The condition $N_p(L)(y^\alpha a_0, y^\alpha b_0) = \mathcal{O}(y^\alpha)$ then corresponds to equations that arise from setting to zero the terms proportional to $y^{\alpha-1}$.
These are to be interpreted as equations for $\alpha$ and determine the indicial roots.

More concretely, since the indicial equations invoke $\R^4$-invariance, it is clear from Proposition~\ref{prop:background-dimensional-reduction-Nahm} that the indicial equations are equivalent to the linearization of the $\beta$-twisted octonionic Nahm equations around $\bigl(A^{\mathrm{NP}}, B^{\mathrm{NP}}\bigr)$.
In the current situation, the leading order terms are given by the model solutions $A_i^{\mathrm{NP}} = y^{-1} \sin\beta \mathfrak{t}_{\tau(i)}$ and~${B_i^{\mathrm{NP}} = y^{-1} \cos\beta \mathfrak{t}_{i}}$.
Plugging $A^{\mathrm{NP}} + y^\alpha a$ and~${B^{\mathrm{NP}} + y^\alpha b}$ into these equations and extracting the terms at order $y^{\alpha-1}$ leads to the following set of indicial equations:
\begin{gather}
\samepage		\alpha b_i + \sin\beta\bigl[\mathfrak{t}_{\tau(i)}, a_s\bigr] - \cos\beta[\mathfrak{t}_i, a_y] \label{eq:background-Nahm-indicial-equations-1}\\
		\qquad{}+ \epsilon_{ijk} \bigl( \cos^2\beta[\mathfrak{t}_j, b_k] + \sin\beta \cos \beta[\mathfrak{t}_j, a_k] - \sin\beta \cos\beta\bigl[\mathfrak{t}_{\tau(j)}, a_k\bigr] + \sin^2\beta \bigl[\mathfrak{t}_{\tau(j)}, b_k\bigr] \bigr)
= 0, \nonumber \\
		\alpha a_i - \cos\beta[\mathfrak{t}_{i}, a_s] - \sin\beta[\mathfrak{t}_{\tau(i)}, a_y] \label{eq:background-Nahm-indicial-equations-2}\\
		\qquad{}- \epsilon_{ijk} \bigl( \cos^2\beta[\mathfrak{t}_{j}, a_k] - \sin\beta \cos\beta[\mathfrak{t}_{j}, b_k] + \sin\beta \cos \beta\bigl[\mathfrak{t}_{\tau(j)}, b_k\bigr] + \sin^2\beta \bigl[\mathfrak{t}_{\tau(j)}, a_k\bigr] \bigr)
	= 0, \nonumber \\
	\alpha a_s - \cos\beta[\mathfrak{t}_i, a_i] + \sin\beta\bigl[\mathfrak{t}_{\tau(i)}, b_i\bigr] = 0, \label{eq:background-Nahm-indicial-equations-3}\\
	\alpha a_y + \sin\beta\bigl[\mathfrak{t}_{\tau(i)}, a_i\bigr] + \cos\beta[\mathfrak{t}_i, b_i] = 0. \label{eq:background-Nahm-indicial-equations-4}
\end{gather}
The first three equations arise from the $\beta$-twisted octonionic Nahm equations~\eqref{eq:background-twisted-Nahm-equations-expanded}, while the last is the gauge fixing condition~\eqref{eq:background-Nahm-gauge-fixing}.

Observe that when $\beta=0$, and thus $\theta = \pi/2$, the equations decouple into two quaternionic Nahm-like equations for ${(a_y, \vec{b} = (b_1, b_2, b_3) )}$ and ${(a_s , \vec{a} = (a_1, a_2, a_3) )}$, respectively.
Up to a~reinterpretation of $a_s$, these are the indicial equations for the $\theta = \pi/2$ version of the Kapustin--Witten equations that were analyzed by Mazzeo and Witten.
As will be explained momentarily, an analogous decoupling also exists for $\beta\neq 0$, and this will ultimately lead to the conclusion that the indicial roots at $\del M^5 \setminus \Sigma_K$ do not depend on $\beta$ at all.
In the upcoming discussion we closely follow the exposition for the case $\beta=0$ in~\cite[Section~2.3]{Mazzeo2014}.

The Lie algebra $\mathfrak{su}(2)$ acts on the fields in~\eqref{eq:background-Nahm-indicial-equations-1}--\eqref{eq:background-Nahm-indicial-equations-4} in three important ways and this can be exploited to simplify the equations.
Below, we denote by $V_j$ the $2j+1$ dimensional representations of $\mathfrak{su}(2)$, where $j$ is a non-negative half-integer, commonly called spin.

First, consider the subalgebra $\mathfrak{su}(2)_\mathfrak{t} \subseteq \mathfrak{g}$ that is spanned by $(\mathfrak{t}_i)_{i=1,2,3}$.
It is the image of $\rho\colon \mathfrak{su}(2) \to \mathfrak{g}$ and depends on the choice of $\phi_\rho$ in the Nahm pole boundary condition.
Since we are only concerned with regular Nahm pole boundary conditions, $\rho$ is a principal embedding and $\mathfrak{su}(2)_\mathfrak{t}$ is a regular subalgebra.
Under the adjoint action of $\mathfrak{su}(2)_\mathfrak{t}$, the Lie algebra $\mathfrak{g}$ then decomposes\footnote{For a generic $\mathfrak{su}(2)$-subalgebra, this decomposition exists only for the complexification $\mathfrak{g}_{\C}$, in which case the decomposition may also involve half-integer spins.
}~into a direct sum of non-zero integer spin representations $V_j$, with $j=1,\ldots, N-1$.
Observe that none of the terms in the indicial equations mixes components with values in distinct~$V_j$'s (the action of $\mathfrak{t}_i$ preserves $V_j$ by definition).
This means that in solving the equations we can from now on assume that $a_s$, $a_y$, $\vec{a}$ and $\vec{b}$ all take values in the same representation $V_j$.

Second, there is an $\mathfrak{su}(2)$ action on the vector degrees of freedom of $\vec{a}$ and $\vec{b}$, which we denote $\mathfrak{su}(2)_\mathfrak{s}$ with generators $(\mathfrak{s}_i)_{i=1,2,3}$.
For this, recall from Section~\ref{sec:background-Nahm-equations} that in the quaternionic Nahm equations we think of the vectors $\vec{a}$ and $\vec{b}$ as elements of $\mathfrak{g} \otimes \operatorname{Im} \mathbb{H}$.
The imaginary quaternions $\operatorname{Im} \mathbb{H}$ naturally form an $\mathfrak{su}(2)$ representation, induced by acting with commutators (= cross product) on themselves.
Slightly abusing notation, this action is represented on $\R^3 \simeq \operatorname{Im} \mathbb{H}$ by multiplication with the $3\times 3$ matrices $(\mathfrak{s}_i)_{jk} = - \epsilon_{ijk}$.

If $\beta \neq 0$, we need to take into account that the two quaternionic parts combine into $(a_y, \vec{b}, a_s, \vec{a}) \in \mathfrak{g}\otimes \mathbb{O}$.
There is an analogous $\mathfrak{su}(2)_\mathfrak{s}$ action on each of the vector space summands in the decomposition $\mathbb{O} \simeq \R \oplus \operatorname{Im}\mathbb{H} \oplus \R \oplus \operatorname{Im}\mathbb{H}$, which is induced by taking commutators with elements of the first $\operatorname{Im} \mathbb{H}$-factor and subsequently discarding any terms that land outside the summand in question.
Specifically, we define $\mathfrak{su}(2)_\mathfrak{s}$ by the action of its generators $(\mathfrak{s}_i)_{i=1,2,3}$ as follows.
The action on $a_y$ and $a_s$ is trivial, i.e., generators are represented by multiplication with $\mathfrak{s}_i = 0$.
Meanwhile, the action on $\vec{b}$ is represented by the $3\times 3$ matrices $(\mathfrak{s}_i)_{jk} = -\epsilon_{ijk}$ as in the quaternionic case.
The action on $\vec{a}$ is similar but comes with an additional subtlety, since octonionic multiplication introduces an additional sign in the action of $\operatorname{Im} \mathbb{H}$.
In this case, the (octonionic) action of the imaginary quaternions is represented by the $3\times 3$ matrices $(A_i)_{jk} = +\epsilon_{ijk}$, which satisfy the commutation relations $[A_i, A_j] = - \epsilon_{ijk} A_k$.
This only provides an $\mathfrak{su}(2)_\mathfrak{s}$ representation if we let $\mathfrak{s}_i$ act via $A_{\tau(i)}$ with an additional anti-cyclic permutation $\tau=(132)$.
In conclusion, $a_s$ and $a_y$ take values in the trivial representation $V_0$, while $\vec{a}$ and $\vec{b}$ are elements of three-dimensional representations $V_1$, where the generators $\mathfrak{s}_i$ act as described above.

Third, the indicial equations are invariant under $\mathfrak{su}(2)_\mathfrak{f}$, generated by the action of $\mathfrak{f}_i := \mathfrak{t}_i \otimes 1 + 1 \otimes \mathfrak{s}_i$ on $\mathfrak{g}\otimes \mathbb{O}$.
If the fields take values in $V_j \subset \mathfrak{g}$, they decompose under the action of~$\mathfrak{su}(2)_\mathfrak{f}$ as follows:
\begin{align*}
	&a_s, a_y\in V_j \otimes V_0 = V_j^0, \qquad \vec{a}, \vec{b}\in V_j \otimes V_1 = \bigoplus_{\eta\in\{-1, 0, +1\}} V_j^\eta.
\end{align*}
Here we have introduced $V_j^\eta \simeq V_{j+\eta}$ to denote representations with total spin $j+\eta$.
It is worth pointing out that for fixed $j$, according to the first line, the components $a_s$ and $a_y$ can only be non-zero when $\eta=0$.

Now, to better understand the indicial equations, consider the ``spin-spin'' operator
\begin{align*}
	\mathfrak{J} := \mathfrak{t}\cdot\mathfrak{s} = \sum_{i=1}^3 \mathfrak{t}_i \otimes \mathfrak{s}_i.
\end{align*}
This operator yields the expressions in \eqref{eq:background-Nahm-indicial-equations-1}~and~\eqref{eq:background-Nahm-indicial-equations-2} that contain $\epsilon_{ijk}$.
Indeed, denoting $\vec{a}^\tau = (a_1, a_3, a_2)$ and $\vec{b}^\tau = (b_1, b_3, b_2)$, the action of $\mathfrak{J}$ is given by
\begin{alignat*}{3}
	&(\, \mathfrak{J}\ \vec{a}\, )_i = - \epsilon_{ijk}\bigl[t_{\tau(j)}, a_k\bigr] ,\qquad&&
	(\, \mathfrak{J}\ \vec{b}\, )_i = \epsilon_{ijk} [t_j, b_k] ,&\\
	&(\, \mathfrak{J}\ \vec{a}^{\, \tau}\, )_{\tau(i)} = \epsilon_{ijk}[t_j, a_k] ,\qquad&&
	(\, \mathfrak{J}\ \vec{b}^{\, \tau}\, )_{\tau(i)} = - \epsilon_{ijk} \bigl[t_{\tau(j)}, b_k\bigr].&
\end{alignat*}
The action of $\mathfrak{J}$ on elements of $V_j^\eta$ is determined by the quadratic Casimir operators of the three~$\mathfrak{su}(2)$ actions
$\mathfrak{J} = -\tfrac{1}{2} \bigl( C_\mathfrak{f}^2 - C_\mathfrak{t}^2 - C_\mathfrak{s}^2 \bigr)$.
In general, the quadratic Casimir operator of $\mathfrak{su}(2)$ with basis $\mathfrak{c}_i$ is defined by \smash{$C^2 = -\sum_{i=1}^3 \mathfrak{c}_i^2$}.
On a spin $J$ representation it takes the constant value $C^2 = J(J+1)$.
In our case, there are three such operators $C_\mathfrak{t}^2$, $C_\mathfrak{s}^2$ and $C_\mathfrak{f}^2$, associated to the three $\mathfrak{su}(2)$ actions on $\mathfrak{g}\otimes\mathbb{O}$.
The values \smash{$C_\mathfrak{t}^2 = j(j+1)$} and \smash{$C_\mathfrak{s}^2 = 2$} are fixed, while \smash{$C_\mathfrak{f}^2$} depends on \smash{$V_j^\eta$} and takes values $(j+\eta)(j+\eta+1)$.
It follows that \smash{$V_j^\eta$} are eigenspaces of the spin-spin operator $\mathfrak{J}$ with eigenvalues $j+1$, $1$, and $-j$, respectively.

Note that orientation reversal via $\tau$ does not preserve the total spin: if $\vec{a}\in V_j^\eta$, then $\vec{a}^\tau$ does not have definite spin with respect to $\mathfrak{su}(2)_\mathfrak{f}$, but is instead given by some non-trivial linear combination in $\oplus_\eta V_j^\eta$.
Since the indicial equations~\eqref{eq:background-Nahm-indicial-equations-1}--\eqref{eq:background-Nahm-indicial-equations-4} contain contributions from $\vec{a}$, $\vec{a}^\tau$, $\vec{b}$, and $\vec{b}^\tau$, its not possible to restrict the equations to $V_j^\eta$.
However, by taking suitable linear combinations of~\eqref{eq:background-Nahm-indicial-equations-1}~and~\eqref{eq:background-Nahm-indicial-equations-2}, one can rewrite these as a set of decoupled equations in~${\sin\beta \vec{a} + \cos\beta \vec{b}}$ and $\cos\beta \vec{a} - \sin\beta \vec{b}$.

On the one hand, if we restrict to $\eta \neq 0$, the terms containing $a_s$ and $a_y$ vanish.
In this case, the indicial equations are equivalent to
\begin{align*}
	&\alpha ( \sin\beta \vec{a} + \cos\beta \vec{b} ) + \mathfrak{J} ( \sin\beta \vec{a} + \cos\beta \vec{b} )= 0 ,\\
	&\alpha ( \cos\beta \vec{a}^{\,\tau} - \sin\beta \vec{b}^{\,\tau} ) - \mathfrak{J} ( \cos\beta \vec{a}^{\, \tau} - \sin\beta \vec{b}^{\, \tau} )= 0	.
\end{align*}
On the other hand, if $\eta = 0$, we can replace the terms containing $a_s$ and $a_y$ by utilizing that equations~\eqref{eq:background-Nahm-indicial-equations-3}~and~\eqref{eq:background-Nahm-indicial-equations-4} are solved by
\begin{align*}
	&a_i= \frac{\alpha}{j(j+1)} \bigl( \cos\beta [\mathfrak{t}_i, a_s] + \sin\beta\big[\mathfrak{t}_{\tau(i)}, a_y\big] \bigr), \\
	&b_{\tau(i)}= \frac{\alpha}{j(j+1)} \bigl( -\sin\beta [\mathfrak{t}_{i}, a_s] + \cos\beta \bigl[\mathfrak{t}_{\tau(i)}, a_y\bigr] \bigr).
\end{align*}
Solving for $[\mathfrak{t}_i, a_s]$ and \smash{$\bigl[\mathfrak{t}_{\tau(i)}, a_y\bigr]$}, plugging it into \eqref{eq:background-Nahm-indicial-equations-1}~and~\eqref{eq:background-Nahm-indicial-equations-2}, and taking appropriate linear combinations yields
\begin{align*}
	&\biggl(\alpha - \frac{j(j+1)}{\alpha} \biggr) ( \sin\beta \vec{a} + \cos\beta \vec{b} ) + \mathfrak{J} ( \sin\beta \vec{a} + \cos\beta \vec{b} )= 0, \\
	&\biggl(\alpha - \frac{j(j+1)}{\alpha} \biggr) ( \cos\beta \vec{a}^{\,\tau} - \sin\beta \vec{b}^{\,\tau} ) - \mathfrak{J} ( \cos\beta \vec{a}^{\, \tau} - \sin\beta \vec{b}^{\, \tau} )= 0	.
\end{align*}

In either case, the indicial equations reduce to an eigenvalue problem for the spin-spin operator $\mathfrak{J}$.
Using the fact that $V_j^\eta$ are eigenspaces of $\mathfrak{J}$ with the previously discussed eigenvalues leads to the following table of indicial roots:
\begin{align*}
	&(\sin\beta \vec{a} + \cos\beta \vec{b}) \in V_j^\eta \colon\ \alpha = \begin{cases} j, & \eta = 1, \\ -(j+1), j, & \eta = 0, \\ -(j+1), & \eta = -1, \end{cases} \\
	&(\cos\beta \vec{a} - \sin\beta \vec{b})^\tau \in V_j^\eta \colon\ \alpha = \begin{cases} -j, & \eta = 1, \\ j+1, -j ,& \eta = 0, \\ j+1, & \eta = -1. \end{cases}
\tag*{\qed}
\end{align*}
\renewcommand{\qed}{}
\end{proof}
This concludes the evaluation of indicial roots at points in the depth 1 stratum $\del M \setminus \Sigma_K$ needed for the proof of Theorem~\ref{thm:background-elliptic-regularity}.
Crucially, the list of indicial roots coincides with the one for the $\pi/2$-version of the Kapustin--Witten equations determined in~\cite{Mazzeo2014}.
In particular, there are no indicial roots in the interval $(-1, \overline{\alpha}) = (-1, 1)$ and the indicial set $\Delta_0$ is bounded from the left by $1$.

Moving on to a short description of the depth 2 stratum, let $p \in \Sigma_K$.
In this case, the normal operator $N_p(L)$ acts on functions over \smash{$T_p \Sigma_K \times C(H^2) \simeq \R^2_{s,t} \times [0,\infty)_R \times H^2_{\psi,\vartheta}$}.
The indicial equations now arise from considering $\R^2_{s,t}$-invariant functions of the form $(R^\alpha a , R^\alpha b)$, where~$a$,~$b$ are independent of $R$.
Equivalently, according to Proposition~\ref{prop:background-dimensional-reduction-EBE}, these are determined by plugging in \smash{$A = A^{\theta, \lambda}+R^{\alpha} a$} and \smash{$B^{\theta,\lambda} + R^\alpha b$} into the $\theta$-TEBE and extracting the terms at leading order $R^{\alpha-1}$.

The evaluation of these equations is somewhat more involved than before.
The indicial roots of $\theta$-twisted knot singularities near $\del_K M$ were determined for the $\theta$-Kapustin--Witten equations by Dimakis.
\begin{Lemma}[{\cite[Lemma 3.5]{Dimakis2022}}]
If $G={\rm SU}(2)$ and $(A,\phi) \sim \bigl(A^{\lambda,\theta}, \phi^{\lambda,\theta}\bigr)$ as $R\to 0$, then the set of indicial roots at $\psi = \pi/2$ is $\{-1,2\}$ in accordance with the Nahm pole boundary condition, at~${\psi=0}$ is $\{-\lambda-1, 0, 0 , \lambda + 1\}$, and there are no indicial roots in the interval $(-2, 0)$ at $R=0$.
\end{Lemma}
Importantly, the indicial roots at $\psi = \pi/2$ are compatible with the indicial roots at the depth 1 stratum, such that the depth 2 normal operator is iteratively invertible (roughly: it is invertible on certain rescaled versions of the function spaces on which the depth 1 normal operator is invertible).
We conclude that the Haydys--Witten operator is an elliptic iie operator, since up to a reinterpretation of field components its normal operator at glancing angle $\theta$ coincides with the normal operator of the $\theta$-Kapustin--Witten operator.
In particular, the indicial set $\Delta_K$ is bounded from the left by $\overline{\alpha} = 0$ as claimed in Theorem~\ref{thm:background-elliptic-regularity}.

\section{Haydys--Witten homology}
\label{sec:background-HWF-theory}

We now have all ingredients at hand to qualitatively define Haydys--Witten Homology, which assigns a Floer-type instanton homology $HF\bigl(W^4\bigr)$ to any Riemannian four-manifold.
The construction is a standard application of the ideas of Floer theory and is summarized in Section~\ref{sec:background-HWF-theory-general-manifolds}.
If $W^4$ admits a non-vanishing unit vector field $w$, there is actually a one-parameter family of such homologies $HF_{\theta}\bigl(W^4\bigr)$, $\theta \in [0,\pi]$.
Moreover, the construction is functorial: Any cobordism $\bigl(M^5,v\bigr)$ of four-manifolds $W^4$ and \smash{$\tilde W^4$}, where $v$ is a non-vanishing vector field on $M^5$, provides a linear map between the Floer groups associated to its boundaries.
In particular, this implies the existence of natural linear maps between Floer groups for different values of $\theta$.

It is to a large extent unclear under what conditions Haydys--Witten homology has a fully rigorous meaning.
Currently, the most important missing parts are compactness and gluing results for the moduli space of Kapustin--Witten and Haydys--Witten solutions.
There have been some important advances in this direction, mostly due to Taubes~\cite{Taubes2013, Taubes2017, Taubes2019, Taubes2021, Taubes2018}, but also see~\cite{He2019d, Tanaka2019}.

We conclude the article with a short explanation of Witten's proposal regarding Khovanov homology from the perspective of Haydys--Witten Floer theory in Section~\ref{sec:background-HWF-theory-and-Khovanov-Homology}.
This has attracted a lot of attention and provides an important testing ground for Haydys--Witten Floer theory and hints at the information that is measured by the topological invariants.
Conversely, one may hope to ``read off'' properties of Khovanov homology for general three-manifolds from the general properties of Floer-like theories.

\subsection{An instanton Floer homology for the Haydys--Witten equations}
\label{sec:background-HWF-theory-general-manifolds}

The way things have been set up, it is convenient to explain the construction of Haydys--Witten homology from the perspective of the five-dimensional Haydys--Witten geometry.
Let $M^5$ be a non-compact Riemannian manifold with corners, $G$ a simply connected compact Lie group, and $E\to M^5$ a principal $G$-bundle.
Assume $M^5$ is equipped with a non-vanishing unit vector field $v$ that approaches ends at constant angles.
The standard example to keep in mind is \smash{$M^5 = \R_s \times X^3 \times \R^+_y$} with $v= \sin\theta \del_s + \cos\theta \del_y$.
Note that $M^5$ may have ``corners at infinity'', commonly called poly-cylindrical ends, that separate non-compact ends at which $v$ has different incidence angles.
If we wish to include a~'t~Hooft operator supported on some embedded surface~${\Sigma_K \subset \del M}$ in one of the boundary components, then we implicitly take $M^5$ to be the blowup along $\Sigma_K$ and label the newly introduced boundary component $\del_K M$ with a magnetic charge~${\lambda \in \Gamma_{\mathrm{char}}^\vee}$.

Denote by $\mathfrak{B}$ a complete collection of boundary conditions for Haydys--Witten fields $(A,B)$ on $M^5$.
We think of this as a set that contains for each end of $M^5$ the information of the type of boundary condition and any choices associated to it.
For example, this might be the choice of $\beta$-twisted Nahm pole boundary conditions, which involves boundary data ${\phi_\rho \in \Omega^2_{v,+}\bigl([0,\epsilon)_y \times W^4\bigr)}$.
Similarly, if the boundary arises from the blowup of a surface $\Sigma_K$ and is labeled with a non-zero charge \smash{$\lambda \in \Gamma_{\mathrm{char}}^\vee$}, then $\mathfrak{B}$ associates the description of a knot singularity within a Nahm pole boundary conditions.
At non-compact ends we generally demand that the fields approach $\R$-invariant solutions, so $\mathfrak{B}$ specifies a choice of $\theta$-Kapustin--Witten solution (or Vafa--Witten solution if $\theta = 0$).
We write $\mathcal{M}^{\mathrm{HW}}\bigl(M^5, v;\mathfrak{B}\bigr)$ for the space of Haydys--Witten solutions that satisfy the boundary conditions determined by $\mathfrak{B}$, modulo gauge transformations that act trivially at boundaries and non-compact ends.

Let us now consider five-manifolds of the form $M^5 = \R_s \times W^4$, where $W^4$ is a smooth Riemannian manifold with corners, not necessarily compact.
$M^5$ always admits the non-vanishing vector field $v = \del_s$, which approaches the ends at $s=\pm\infty$ with incidence angle $\theta = 0$.
Whenever $W^4$ admits\footnote{%
The existence of a non-vanishing vector for compact $W^4$ implies that its Euler characteristic vanishes.
If $W^4$ is open there is no such obstruction and there always exists a non-vanishing vector field $w$, since any zeroes can be ``combed to infinity''.
In neither case is it necessary for $W^4$ to be of product type $W^4 = X^3 \times I$.
} a non-vanishing unit vector field $w$, there is a natural one-parameter family of non-vanishing vector fields $v = \cos\theta \del_s + \sin\theta w$, with $\theta \in [0,\pi]$, that interpolates between~$\del_s$ and $w$.
There are, of course, many other possible choices of $v$ on $\R_s \times W^4$; in particular, $\theta$ could vary along $\R_s$, such that the angles at $s=\pm \infty$ need not coincide.
We will come back to this later, in the more general context of cobordisms \smash{$\bigl(M^5, v\bigr)$} between four-manifolds with associated incidence angles \smash{$\bigl(W^4, \theta\bigr)$} and \smash{$\bigl(\tilde{W}^4, \tilde{\theta}\bigr)$}.
For now consider the cylinder $M^5 = \R_s \times W^4$ and fix $v = \cos\theta \del_s + \sin\theta w$ for some constant $\theta$.

The boundary conditions at $s\to\pm \infty$ correspond to solutions of the $\theta$-Kapustin--Witten equations on $W^4$.
Denote the moduli space of Kapustin--Witten solutions modulo gauge transformations by \smash{$\mathcal{M}^{\mathrm{KW}}\bigl(W^4, \theta\bigr)$}.
Given \smash{$x,x^\prime \in \mathcal{M}^{\mathrm{KW}}\bigl(W^4, \theta\bigr)$}, a complete set of boundary conditions on $M^5 = \R_s \times W^4$ is given by additionally specifying boundary conditions $\mathfrak{b}$ for each of the remaining ends of $M^5$,
\begin{align*}
	\mathfrak{B} = \mathfrak{b} \sqcup \bigl\{\lim_{s\to-\infty} (A,B) = x \bigr\} \sqcup \bigl\{\lim_{s \to +\infty} (A,B) = x^\prime \bigr\}.
\end{align*}
The choices collected in $\mathfrak{b}$ have to be compatible with $x$ and $x^\prime$ at corners of $M^5$.
This means that~$\mathfrak{b}$ has to be chosen in such a way that it interpolates between the boundary conditions that~$x$ and~$x^\prime$ satisfy.
One can think of this as a collection of four-dimensional ``boundary instantons'', one for each end of $M^5$.

A simple, yet non-trivial example of such a boundary instanton arises for $M^5 = \R_s\times X^3 \times \R^+_y$ at $y\to\infty$.
First note that a natural boundary condition for a Kapustin--Witten solution ${x=(A,\phi)}$ on $X^3 \times \R^+_y$ is that it approaches a flat connection $(A^{\sigma} ,0 )$ as $y\to \infty$, specified by a~choice of a group homomorphism \smash{$\sigma\colon \pi_1\bigl(X^3\bigr) \to G$}.
Hence, assume that $x, x^\prime \in \mathcal{M}^{\mathrm{KW}}\bigl(X^3\times\R^+_y , \theta\bigr)$ approach flat connections associated to $\sigma$ and $\sigma^\prime$, respectively.
A consistent boundary condition at the non-compact end $y\to \infty$ of \smash{$ M^5 = \R_s\times X^3 \times \R^+_y$} must then be a solution of $\beta$-Kapustin--Witten equations on \smash{$\R_s \times X^3$} that interpolates between the two flat connections~$A^\sigma$ and~$A^{\sigma^\prime}$.
In the special case where $\beta=0$ and $\phi=0$, this is equivalent to a choice of self-dual connection, i.e., a Donaldson--Floer instanton, on the four-manifold $\R_s \times X^3$ that sits at ${y = \infty}$.\looseness=1

Let us now define the Floer chains, also known as the Morse--Smale--Witten complex, that underlies Haydys--Witten homology.
For simplicity assume that there is only a finite set\footnote{%
This is equivalent to the statement that the moduli space is a compact manifold of dimension zero, which is something one would ultimately like to prove.
}~of Kapustin--Witten solutions on $W^4$ and consider the free abelian group generated by these \mbox{solutions}:
\begin{align*}
	CF_{\theta} = \bigoplus_{ x \in \mathcal{M}^{\mathrm{KW}}(W^4,\theta)} \Z\, \left[ x \right].
\end{align*}
Note that, by definition, elements of $CF_{\theta}$ are stationary solutions of the Haydys--Witten equations -- or equivalently, critical points of an appropriate Kapustin--Witten functional.
This coincides with the usual construction of a Morse--Smale--Witten complex in Morse theory.

The Morse--Smale--Witten complex is equipped with a differential ${\rm d}_{v}$ that counts Haydys--Witten instantons that interpolate between $x$ and $y$.
To make this precise, consider the moduli space of Haydys--Witten solutions where the fields $(A,B)$ approach $x$ and $y$ as $s\to\pm\infty$, respectively, and moreover satisfy some fixed boundary conditions $\mathfrak{b}$ at the remaining boundaries and non-compact ends.
This moduli space always admits an $\R$-action by translation $s\mapsto s+c$ along the flow direction $\R_s$, which maps one solution to an equivalent one that differs only by the parametrization of $\R_s$.
To count instantons, we thus consider the quotient of the moduli space by this action.
Also, we need to take into account that several instantons at the boundary might provide a consistent choice of boundary conditions $\mathfrak{b}$.
As a result Haydys--Witten instantons that interpolate from $x$ to $y$ are classified by
\begin{align*}
\begin{split}
	\mathcal{M}(x,y):= \bigcup_{\mathfrak{b}}
	\bigl(
		\mathcal{M}^{\mathrm{HW}} \bigl(
			\R_s \times W^4,\, v, \,\mathfrak{B} = \mathfrak{b} \sqcup \!\bigl\{ \lim_{s\to -\infty}\! (A,B) = x
\bigr\} \sqcup \bigl\{ \lim_{s\to +\infty}\! (A,B) = y \bigr\}\bigr)\bigr)\big/\R.
\end{split}
\end{align*}
On grounds of general properties of elliptic differential operators, this is expected to be a smooth oriented manifold.

Note that the boundary conditions $\mathfrak{b}$ are sometimes classified by an analogous moduli space of instanton solutions in one dimension less.
The disjoint union of possible boundary conditions then is equivalent to a product of smooth manifolds.
For example, in the context of boundary instantons at $y=\infty$ on the five-manifold $\R_s \times X^3 \times \R^+_y$, the moduli space is of the form
\begin{align*}
	\mathcal{M}\bigl(x,x^\prime\bigr) =
	\mathcal{M}^{\mathrm{HW}} \Biggl( \lim_{s\to\pm \infty} (A,B) = \begin{cases} x, & \lim_{y \to \infty} x = \sigma, \\ x^\prime, & \lim_{y \to \infty} x^\prime = \sigma^\prime \end{cases}\Biggr) \times \mathcal{M}^{\mathrm{asd}}\bigl(\sigma,\sigma^\prime\bigr),
\end{align*}
where $\mathcal{M}^{\mathrm{asd}}(\sigma,\sigma^\prime)$ is the moduli space of anti-self-dual connections on $\R_s \times X^3$ that interpolate between $A^\sigma$ and $A^{\sigma^\prime}$.

In the definition of ${\rm d}_v$, we rely on the dimension of $\mathcal{M}(x,y)$.
In Morse theory, i.e., on finite-dimensional manifolds, the Morse--Smale--Witten complex carries a natural $\mathbb{Z}$-grading by the Morse index, defined by the number of negative eigenvalues of the Hessian at a given critical point.
Since this determines the number of unstable flow directions in the vicinity of a critical point, the difference between the index of distinct critical points determines the dimension of the moduli space of flows $\mathcal{M}(x,y)$.
In the infinite-dimensional setting the Morse index does not make sense; the linearization of the Kapustin--Witten operator typically has infinitely many negative eigenvalues.\footnote{We have encountered a similar situation in the path-integral description, where we mentioned a grading by the ``fermion number of the filled Dirac sea''.
This makes sense only after choosing a reference vacuum for which the fermion number is \emph{defined} to be zero.}
Observe, however, that the difference of Morse indices only depends on the relative change of negative eigenvalues along a flow line.
This is known as the spectral flow of an operator and has an analogue in Floer theory.
Thus, as is common in Floer theory, we define a relative index $\mu(x,y)$ for any pair of generators $x,y \in \mathcal{M}^{\mathrm{KW}}\bigl(W^4, \theta\bigr)$, by the spectral flow of the Kapustin--Witten differential operator along a Haydys--Witten instanton.
This, in turn, coincides with the index of the Haydys--Witten differential operator when it acts on fields that are subject to the complete set of boundary conditions~$\mathfrak{B}$,
$\mu(x,y) := \operatorname{ind} \HW[v]\restr_{\mathfrak{B}}$.
The moduli space of Haydys--Witten instantons $\mathcal{M}(x,y)$ is expected to have dimension ${\mu(x,y) - 1}$.
Notably, it is zero-dimensional whenever $\mu(x,y)=1$, in which case we denote by $\hash \mathcal{M}(x,y)$ the signed count of its (oriented) elements.

The Floer differential is the linear map $CF_{\theta} \to CF_{\theta}$ defined by
\begin{align*}
	{\rm d}_{v}[x] = \sum_{\mu(x,y) = 1} \hash \mathcal{M}(x,y) \cdot [ y ].
\end{align*}
One expects that ${\rm d}_{v}^2 = 0$, such that $(CF_{\theta}, {\rm d}_{v})$ is indeed a cochain complex.
The standard proof in Floer theory relies on compactness and gluing theorems for the flow equations, which are not yet available for the Haydys--Witten equations.
More precisely, consider the compactification of the moduli space of gradient flows with relative index 2, which is an oriented manifold of dimension 1.
If the Haydys--Witten equations with boundary conditions $\mathfrak{B}$ are well behaved, the compactification is fully determined by adding broken flow lines.
The latter are exactly what we need to count when calculating ${\rm d}_{v}^2$.
Oriented manifolds of dimension one are either circles, which do not have boundary components and cannot contribute to ${\rm d}_{v}^2$, or intervals with boundary components of opposite orientation.
It follows that contributions to ${\rm d}_{v}^2$ always arise in pairs of opposite orientation and consequently add up to $0$.

Assuming all foundational problems can be addressed, we define the Haydys--Witten Floer homology groups as follows.
\begin{Definition}[Haydys--Witten homology]\label{def:Haydys-Witten-Floer-homology}
\!The Haydys--Witten homology associated to a~four-manifold $W^4$ is the homology of the chain complex $(CF_{\theta}, {\rm d}_v)$,
\begin{align*}
	HF_{\theta} \bigl(W^4\bigr) := H( CF_{\theta}, {\rm d}_v ).
\end{align*}
\end{Definition}
\begin{Remark}
It might be helpful to emphasize that the one-parameter family of homology groups only exists if $W^4$ admits a non-vanishing vector field $w$.
Since any non-compact manifold automatically admits a non-vanishing vector field, this is only an obstruction for compact manifolds.
Thus, for compact manifolds generally only the $\theta=0$ version of the Floer homology exists and Proposition~\ref{prop:background-dimensional-reduction-KW} states that in that case the Morse--Smale--Witten complex is generated by Vafa--Witten solutions.
However, according to Theorem~\ref{thm:background-Kapustin--Witten-vanishing} finite energy solutions of the $\theta$-Kapustin--Witten equations on compact manifolds are trivial whenever $\theta\neq 0$ anyway, so we may simply adopt the convention that \smash{$HF_{\theta}\bigl(W^4\bigr)$} is the trivial group for all $\theta \neq 0$.
The only situation where \smash{$HF_{\theta}\bigl(W^4\bigr)$} with $\theta \neq 0$ remains ambiguous is then in the context of infinite energy solutions on compact manifolds, which, for example, appear in connection with Nahm pole boundary conditions.
\end{Remark}

{\bf Instanton grading.}
The Morse--Smale--Witten complex $CF_{\theta}$ is naturally $\mathbb{Z}$-graded by the instanton number of the principal bundle $E\to W^4$, which up to a constant is the integral of the first Pontryagin class \smash{$p_1(\ad E) = \frac{1}{8\pi^2} \Tr F_A \wedge F_A$}.
The complex decomposes into submodules~$CF_{\theta}^k$, spanned by Kapustin--Witten solutions with instanton number $k\in \mathbb{Z}$,
\begin{align*}
	CF_{\theta}^\bullet = \bigoplus_{k \in \mathbb{Z}} CF_{\theta}^k.
\end{align*}

To understand the interaction between the instanton grading and the differential ${\rm d}_v$, observe that $p_1(\ad E)$ is a conserved four-form current.
We can consider the associated charge at any time $s \in \R_s$,
\begin{align*}
	P(s) = \frac{1}{32 \pi^2} \int_{\{s\}\times W^4} \Tr F_A \wedge F_A.
\end{align*}
Although the integrand is conserved, the current density may disappear at boundaries and non-compact ends of $W^4$ as we follow the flow along $\R_s$.
The difference in instanton number between the start and end point $s\to \pm \infty$ of a Haydys--Witten instanton is given by Stokes' theorem,
\begin{align*}
	\Delta P := \lim_{s\to\infty} \bigl( P(s) - P(-s) \bigr) = \sum \frac{1}{32 \pi^2} \int_{\del_i M} \Tr F_A \wedge F_A.
\end{align*}
The right-hand side is a sum over all boundaries and non-compact ends of $M^5$, except the ones at $s=\pm \infty$ (which appear on the left-hand side).
Each end contributes with its own instanton number, or more precisely the instanton number associated to the pullback of $E$ to $\del_i M$.

In conclusion, the differential ${\rm d}_v$ generally doesn't preserve the grading on $CF_{\theta}^\bullet$ and con\-se\-quent\-ly there is no $\mathbb{Z}$-grading on $HF_{\theta}\bigl(W^4\bigr)$.
However, the topology of the pullback bundles at ends of $M^5$ may arguably be viewed as part of the boundary data $\mathfrak{b}$, so the change in $P$-grading is ultimately controlled by the interplay of all the boundary conditions that are imposed on the Haydys--Witten instantons.
For example, one could choose to only take into account Haydys--Witten instantons for which $\Delta P$ is fixed, such that ${\rm d}_v$ has a fixed degree $\Delta P$.

{\bf Cobordisms.}
Let us briefly comment on the functorial properties of Haydys--Witten Floer theory.
Assume \smash{$\bigl(M^5, v\bigr)$} is a cobordism that interpolates between four-manifolds \smash{$\bigl(W^4, \theta\bigr)$} and \smash{$\bigl(\tilde{W}^4, \tilde{\theta}\bigr)$}, where~$\theta$,~$\tilde{\theta}$ denote the incidence angles between $v$ and the boundaries.
We promote the boundaries to non-compact ends by gluing in cylindrical ends $(-\infty,0]_s \times W^4$ and $[0,\infty)_s \times \tilde{W}^4$, respectively.
Since we assume that $v$ is already constant in some tubular neighbourhood of the boundaries, the vector field $v$ extends to a unique constant vector field on the cylinders.

To each end associate the corresponding Morse--Smale--Witten complex of boundary conditions \smash{$CF_{\theta} \bigl(W^4\bigr)$} and \smash{$CF_{\tilde \theta}\bigl(\tilde{W}^4\bigr)$}.
We can proceed exactly as before to define a linear map
\begin{align*}
	\Phi_{(M^5,v)} \colon \ CF_{\theta}\bigl(W^4\bigr) \to CF_{\tilde{\theta}}\bigl(\tilde{W}^4\bigr), \qquad [x] \mapsto \sum_{\mu(x,y)=1} \hash \mathcal{M}(x,y) [y].
\end{align*}
The only difference is that we now count Haydys--Witten solutions on \smash{$\bigl(M^5,v\bigr)$} instead of the cylinder \smash{$\bigl(\R_s\times W^4, v= \del_y\bigr)$}.

Under appropriate compactness and gluing assumptions for the Haydys--Witten equations, the induced map \smash{$\Phi_{(M^5,v)}$} is a chain map
\begin{align*}
	\Phi_{(M^5,v)} \circ {\rm d}_v = {\rm d}_v \circ \Phi_{(M^5,v)}.
\end{align*}
One way to see this is to realize that the concatenation of \smash{$\Phi_{(M^5,v)}$} and ${\rm d}_v$ is determined by the number of broken flow lines of index 2 on $M^5$, since we glue the instantons described by~${\rm d}_v$ to either the initial or final cylindrical end of $M^5$.
As before, these broken flow lines are in correspondence with the boundary components of the moduli space of Haydys--Witten instantons of index 2.
Since the relevant moduli space is the same, regardless of the order of~\smash{$\Phi_{(M^5,v)}$} and~${\rm d}_v$, these counts coincide.

It follows that \smash{$\Phi_{(M^5,v)}$} induces a linear map on homology,
\begin{align*}
	\bigl( \Phi_{(M^5,v)} \bigr)_\ast \colon\ HF_{\theta} \bigl(W^4\bigr) \to HF_{\tilde \theta } \bigl(\tilde{W}^4\bigr).
\end{align*}
Hence, Haydys--Witten homology is a functor from the category of five-dimensional cobordisms, equipped with a non-vanishing vector field, to the category of groups.
It is, therefore, a topological quantum field theory (TQFT) in the sense of the Atiyah--Segal axioms.
From the point of view of physics, it is the TQFT that arises by a topological twist of $5d$ $\mathcal{N}=2$ super Yang--Mills theory on $\R_s\times W^4$ coupled to $4d$ $\mathcal{N}=4$ super Yang--Mills theory at $s=\pm\infty$.

\subsection{Relation to Khovanov homology}
\label{sec:background-HWF-theory-and-Khovanov-Homology}

Haydys--Witten Floer theory was introduced by Witten to describe Khovanov homology in terms of quantum field theory.
Witten showed that there is a relation between Haydys--Witten Floer homology and Chern--Simons theory on $X^3$ -- and thus knot invariants -- if one considers four-manifolds of the form $W^4 = X^3 \times \R^+_y$ with Nahm pole boundary conditions at $y= 0$~\cite{Witten2010, Witten2011}.
Under this correspondence, a knot carries over to a magnetically charged 't Hooft operator embedded in the boundary $\del W^4 = X^3$.
As explained in Section~\ref{sec:background-Nahm-pole-boundary-condition}, this setup is geometrized by considering the blowup \smash{$\bigl[W^4 ; K\bigr]$} and imposing Nahm pole boundary conditions with knot singularities at the boundaries.

Given this, we associate to a pair \smash{$\bigl(X^3, K\bigr)$} the Haydys--Witten homology of \smash{$\bigl[X^3\times \R_y^+ ; K\bigr]$}.
The vector field $w=\del_y$ provides a non-vanishing vector field on $W^4$, so there is a one-parameter family of Haydys--Witten homologies with respect to $v=\cos\theta \del_s + \sin\theta \del_y$, $\theta \in [0,\pi]$.
The associated Morse--Smale--Witten complex $CF_{\theta}$ is spanned by solutions of the $\theta$-Kapustin--Witten equations that satisfy suitable Nahm pole boundary conditions with knot singularities.

The differential ${\rm d}_v$ counts Haydys--Witten instantons on the cylinder \smash{$M^5 = \R_s \times \bigl[W^4; K\bigr]$}.
This manifold is equivalent to the one obtained by first lifting the knot to the $\R_s$-invariant surface $\Sigma_K = \R_s \times K \times \{0\}$ inside the boundary of $\R_s \times W^4$ and blowing up afterwards, i.e.,
\begin{align*}
	M^5 = \R_s \times \bigl[W^4; K\bigr] = \bigl[ \R_s \times W^4; \Sigma_K \bigr].
\end{align*}
As always, we leave the blowup mostly implicit and simply write $M^5 = \R_s \times X^3 \times \R^+_y$ with original boundary $\del_0 M$ at $y=0$ (with $\Sigma_K$ removed) and blown up boundary $\del_K M$.

To fully determine ${\rm d}_v$, it remains to specify which kind of boundary conditions $\mathfrak{b}$ the Haydys--Witten instantons on $\bigl(M^5,v\bigr)$ shall satisfy:
\begin{itemize}\itemsep=0pt
\item At $\del_0 M^5$, the fields satisfy the $\beta$-twisted regular Nahm pole boundary condition, where the incidence angle is given by $\beta = \pi/2-\theta$.
	The boundary data $\phi_\rho$ of the five-dimensional Nahm pole boundary condition is the unique $\R_s$-invariant continuation of some fixed four-dimensional Nahm pole boundary condition at $\del_0 W^4$.

\item At $\del_{K} M^5$, the fields exhibit a knot singularity and are otherwise consistent with the surrounding Nahm pole boundary conditions.
	Since the glancing angle between $v$ and $\Sigma_K$ is $\theta$, the knot singularity is modeled on solutions of the $\theta$-TEBE.

\item At $y \to \infty$, the fields approach an $\R_y$-invariant finite energy solution of the Haydys--Witten equations.
	This corresponds to a solution of the $\beta$-Kapustin--Witten equations, where $\beta = \pi/2 - \theta$.
	
\item At any non-compact end or boundary of $X^3$, the fields approach maximally symmetric, stationary solutions of the Haydys--Witten equations that are compatible with the boundary conditions at adjacent boundaries and independent of the flow direction $\R_s$.
What exactly this means is best described on a case-by-case basis.
	
\end{itemize}
The first two items just spell out the Nahm pole boundary conditions with knot singularity, as described more thoroughly in Section~\ref{sec:background-Nahm-pole-boundary-condition}.
For $y\to\infty$, there might be non-trivial boundary instantons, classified by solutions of $\beta$-Kapustin--Witten solutions on $\R_s \times X^3$.
Note that for~${X^3 = S^3}$ or $\R^3$ there are no non-trivial Kapustin--Witten solutions with finite energy, because of the vanishing result of Corollary~\ref{cor:background-Nagy-Oliveira-Conjecture}.
Since the rest of the boundary conditions are explicitly chosen to be $\R_s$-invariant, the differential ${\rm d}_v$ preserves the instanton grading and Haydys--Witten homology is $\mathbb{Z}$-graded.

By definition, Haydys--Witten homology \smash{$HF_{\theta}\bigl(\bigl[X^3\times \R_y^+ ; K\bigr]\bigr)$} is given by solutions of the $\theta$-Kapustin--Witten equations, subject to Nahm pole boundary conditions with knot singularities at $y=0$, modulo Haydys--Witten instantons.
Witten originally described the case where $v= \del_y$, in which case $CF_{\pi/2}$ is spanned by solutions of the $\theta=\pi/2$ version of the Kapustin--Witten equations, in which case boundary instantons at $y\to \infty$ are given by Vafa--Witten solutions on $\R_s \times X^3$.
As mentioned earlier, the deformation to $\theta\neq \pi/2$ was considered by Gaiotto and Witten soon afterwards.

With this, we finally arrive at Witten's gauge theoretic approach to Khovanov homology, which is succinctly summarized by the following statement.
\begin{conjecture*}[{\cite{Witten2011}}]
Haydys--Witten homology \smash{$HF_{\theta}\bigl(\bigl[X^3\times \R_y^+ ; K\bigr]\bigr)$} is a topological invariant of the pair \smash{$\bigl(X^3,K\bigr)$}.
In particular, if $X^3 =S^3$ or $\R^3$, this invariant is $\mathbb{Z}$-graded and for $\theta=\pi/2$ it coincides with the Khovanov homology of the knot,
\begin{align*}
	HF_{\pi/2}^\bullet\bigl(\bigl[S^3\times\R_y^+;K\bigr]\bigr) = \mathrm{Kh}^\bullet(K).
\end{align*}
Moreover, any knot cobordism $\Sigma$ induces a map on Haydys--Witten homology via the five-di\-men\-sional cobordism $M^5 = \bigl[\R_s \times X^3 \times \R_y^+; \Sigma\bigr]$ that coincides with the corresponding map on Khovanov homology.
\end{conjecture*}

\begin{Remark}
Let us mention a non-trivial sanity check, provided by Witten, regarding the effect of a change of framing.
Consider a knot cobordism $\Sigma$ with $\del\Sigma = - K_0 \sqcup K_1$.
When $\Sigma$ is a~general knot cobordism, the Nahm pole boundary condition on \smash{$M^5 = \bigl[\R_s \times S^3 \times R_y^+; \Sigma\bigr]$} can no longer be an $\R_s$-invariant continuation.
Instead we use parallel transport to extend the initial boundary data along the flow direction to \smash{$\del_0 M = \bigl(\R_s \times S^3\bigr) \setminus \Sigma$}.
This may introduce non-trivial topology to the bundle $\ad E$ at the boundary and thus a possibly non-zero boundary instanton number.
Associated to this is a change in instanton number $\Delta P \neq 0$ between the initial and final Kapustin--Witten solution.
$\Delta P$ depends on the topology of the embedded surface $\Sigma$; more precisely, it is determined by its Euler characteristic and self-intersection number.
For example, consider $\Sigma = \R_s \times K$ with a trivialization of the normal bundle that corresponds to a single self-intersection (with respect to some normalization).
Then $\Sigma$ interpolates between different framings of $K$ and the induced map on homology encodes a shift in instanton degree for the two knots.
This matches an analogous shift in grading upon changing the framing of a knot in~Khovanov homology.
We refer to~\cite[Section~5.4]{Witten2011} for a vastly more detailed explanation.
\end{Remark}

\subsection*{Acknowledgements}
I would like to thank Johannes Walcher and Fabian Hahner for discussions and comments on an early draft of this article.
I am also grateful to anonymous referees for constructive suggestions that improved the exposition.
This work is funded by the Deutsche Forschungsgemeinschaft (DFG, German Research Foundation) under Germany's Excellence Strategy EXC 2181/1~- 390900948 (the Heidelberg STRUCTURES Excellence Cluster).


\pdfbookmark[1]{References}{ref}
\LastPageEnding

\end{document}